\def\submission{0}
\def\arxiv{1}
\ifnum\arxiv=1
\pdfoutput=1
\fi
\ifnum\submission=0
    \documentclass{llncs}
    \newcommand{\Miryam}[1]{{\footnotesize\color{cyan}[Miryam: #1]}}
    \newcommand{\yuhsuan}[1]{{\footnotesize\color{blue}[Yuyu: #1]}}
    \newcommand{\ychsieh}[1]{{\footnotesize\color{red}[yaoching: #1]}}
    \usepackage{authblk}
\else
    \documentclass{llncs}
    \newcommand{\Miryam}[1]{}
    \newcommand{\yuhsuan}[1]{}
    \newcommand{\ychsieh}[1]{}
\fi
\usepackage{tikz-cd}
\usepackage{comment}

\usepackage{amsmath,kbordermatrix,blkarray}

\renewcommand{\kbldelim}{(}
\renewcommand{\kbrdelim}{)}

\usepackage{scalerel,stackengine}
\stackMath
\newcommand\reallywidetilde[1]{%
\savestack{\tmpbox}{\stretchto{%
  \scaleto{%
    \scalerel*[\widthof{\ensuremath{#1}}]{\kern-.6pt\sim\kern-.6pt}%
    {\rule[-\textheight/2]{1ex}{\textheight}}
  }{\textheight}%
}{0.5ex}}%
\stackon[1pt]{#1}{\tmpbox}%
}
\usepackage{enumitem}
\usepackage{graphicx,color,url,booktabs,comment}
\usepackage{amsthm}
\usepackage{amsmath, amssymb,amsbsy}
\usepackage{mathrsfs}
\usepackage{qcircuit}
\usepackage{braket}
\usepackage{algorithm}
\usepackage{caption}
\usepackage[noend]{algorithmic}
\usepackage{cite,xcolor,float}
\usepackage[utf8]{inputenc}
\usepackage{tabularx}
\usepackage[pdfstartview=FitH,colorlinks,linkcolor=blue,filecolor=blue,citecolor=blue,urlcolor=blue]{hyperref}
\usepackage{fixltx2e}
\usepackage{xspace}
\usepackage{cleveref} 
\usepackage{mdframed}
\usepackage{amsfonts}
\usepackage{tikz}
\usetikzlibrary{positioning,fit,calc}

\usepackage{setspace}
\newcounter{protocol}

\newcommand{\Acal}{\mathcal{A}}

\newcommand{\Ccal}{\mathcal{C}}

\newcommand{\Ecal}{\mathcal{E}}

\newcommand{\Ical}{\mathcal{I}}
\newcommand{\Jcal}{\mathcal{J}}
\newcommand{\Kcal}{\mathcal{K}}
\newcommand{\Lcal}{\mathcal{L}}
\newcommand{\Mcal}{\mathcal{M}}

\newcommand{\Ocal}{\mathcal{O}}
\newcommand{\Pcal}{\mathcal{P}}
\newcommand{\Qcal}{\mathcal{Q}}

\newcommand{\Scal}{\mathcal{S}}

\newcommand{\Fbb}{\mathbb{F}}

\newcommand{\Qbb}{\mathbb{Q}}

\newcommand{\Zbb}{\mathbb{Z}}

\newcommand{\secparam}{\kappa}
\newcommand{\pk}{\mathsf{pk}}
\newcommand{\sk}{\mathsf{sk}}
\newcommand{\bbN}{\mathbb{N}}
\newcommand{\ok}{\mathsf{msk}}

\newcommand{\MGen}{\mathbf{MKeygen}}
\newcommand{\Gen}{\mathbf{Keygen}}
\newcommand{\Com}{\mathbf{Commit}}

\newcommand{\Resp}{\mathbf{Resp}}
\newcommand{\Ver}{\mathbf{Verify}}
\newcommand{\Open}{\mathbf{Open}}
\newcommand{\Sim}{\mathbf{Sim}}
\newcommand{\Hyb}{\mathbf{Hyb}}
\newcommand{\Ext}{\mathbf{Ext}}

\newcommand{\Sign}{\mathbf{Sign}}
\newcommand{\Vrfy}{\mathbf{Verify}}
\newcommand{\Trans}{\mathbf{Trans}}
\newcommand{\inv}{^{-1}}
\newcommand{\Cl}{\mathsf{Cl}(\mathcal{O})}
\newcommand{\nset}[2]{\{#1_#2\}_{#2\in[n]}}
\newcommand{\aset}[3]{\{#1_#2\}_{#2\in#3}}
\newcommand{\pset}[3]{\{#1\}_{#2\in#3}}
\newcommand{\nmk}{[n]\setminus\{k\}}
\newcommand{\Usample}{\xleftarrow{\$}}
\newcommand{\sch}{\{1,2,3,4\}}

\usepackage{multicol}

\newcommand*{\tikzmk}[1]{\tikz[remember picture,overlay,] \node (#1) {};\ignorespaces}
\newcommand{\boxit}[1]{\tikz[remember picture,overlay]{\node[yshift=3pt,fill=#1,opacity=.25,fit={(A)($(B)$)}] {};}\ignorespaces}

\title{Isogeny-based Group Signatures and Accountable Ring Signatures in QROM}
\ifnum\submission=0

\author{ Kai-Min Chung\inst{1} \and Yao-Ching Hsieh\inst{2} \and
  Mi-Ying (Miryam) Huang\inst{3} \and \\\vspace{-1em} Yu-Hsuan Huang\inst{1,4}
  \and Tanja Lange\inst{5} \and Bo-Yin Yang\inst{1} } \institute{
  Academia Sinica, Taiwan \\
  \and University of Washington, United States \\
  \and University of Southern California, United States \\
  \and Centrum Wiskunde \& Informatica, The Netherlands \\
  \and Eindhoven University of Technology, The Netherlands \\
  \email{kmchung@iis.sinica.edu.tw, ychsieh@cs.washington.edu, miying.huang@usc.edu, Yu-Hsuan.Huang@cwi.nl, byyang@iis.sinica.edu.tw, tanja@hyperelliptic.org}}

\date{}

\else
    \author{}
    \institute{}
\fi

\begin{document}

\maketitle

\ifnum\submission=1\vspace{-4em}\fi

\begin{abstract}
We present the first provably secure isogeny-based group signature (GS) and accountable ring signature (ARS) in the quantum random oracle model (QROM). We do so via introducing and constructing an intermediate primitive called the {\em openable sigma protocol} and demonstrating that any such protocol gives rise to a secure GS and ARS. Furthermore, QROM security is guaranteed if an additional perfect unique-response property (which is achieved via our tailored construction) is satisfied.

Previous works by Beullens et al. (Eurocrypt 2022, Asiacrypt 2020) proposed isogeny-based GS and ARS with better efficiency but were only analyzed in the classical random oracle model (CROM). It is well-known that CROM security does not generally translate to QROM security; with the growing relevance of isogeny-based constructions in post-quantum cryptography, the current state of the art is unsatisfactory. Moreover, the aforementioned existing isogeny-based signatures were recently affected by the Fiat-Shamir with aborts (FSwA) flaw discovered by Barbosa et al. and Devevey et al. (CRYPTO 2023), leaving the provable security of isogeny-based signatures open to question once again. Our constructions are not only immune to the FSwA flaw but also provide stronger QROM security. As current QROM-secure ARS and GS schemes are mostly lattice-based, we offer a robust post-quantum alternative should lattice assumptions weaken.

\end{abstract}

\ifnum\submission=1\vspace{-3em}\fi

\ifnum\submission=0
\begingroup
\makeatletter
\def\@thefnmark{} \@footnotetext{
Authors list in alphabetical order; see
\url{https://www.ams.org/profession/leaders/culture/CultureStatement04.pdf}.}
\endgroup
\fi

\section{Introduction}

\noindent\textbf{Group signatures and accountable ring signature} {\em Group signatures} (GS), first proposed in
\cite{chaum91}, are signature schemes that permit signing for a {\em group} of players chosen by a prescribed {\em group manager}. Each player can generate publicly verifiable signatures on behalf of the
group while keeping himself anonymous to everyone else except for the group
manager. The group manager has the authority to {\em open}, i.e., to
reveal the signer's identity from a signature with its {\em master secret key}. There have been numerous works devoted to group
signatures. Many of them aimed to give refinements and extensions to
the primitive \cite{CM98,BSZ05}. An important line of research on group signatures studies variants with {\em dynamic groups}. In contrast to the original formulation where only {\em static groups} are supported \cite{chaum91, bellare2003foundations},
a dynamic group signature scheme allows updating the group of players after setup. The notion of {\em partially dynamic} group signatures was
formulated in \cite{BSZ05,KY06}, where parties can join a group but cannot be removed. {\em Accountable ring signatures} (ARS), first proposed in \cite{xu04}, provides the ``dynamic property for groups'' in a
different aspect. While having ``ring signature''
\cite{rivest2001how} within its name, ARS can also be viewed as a variant
of group signatures where groups are {\em fully dynamic} but {\em not authenticated}. In an ARS scheme, the manager no longer controls the group. Instead, a signer can freely decide which
master public key to use and which group to sign for, and the corresponding master secret key can then open its identity. Though seemingly incomparable to a standard group signature, an ARS scheme can, in fact, trivially imply a group signature scheme simply by fixing the
group at the setup stage. Later, \cite{bootle15}
proposed a stringent formulation for ARS and a provable
construction based on the DDH assumption. It is further shown in
\cite{BCC16} that such a stringent ARS scheme can be generally
transformed to a {\em fully dynamic} group signature scheme. There has been increasing attention to the importance of post-quantum security for cryptographic primitives.
Various attempts emerged to construct group signatures based on
cryptographic assumptions that resist quantum attacks. \cite{GKV10} first gave a group signature construction from
lattice-based assumptions. Several constructions of lattice-based group signatures followed this, either for static
groups \cite{LLL13, NZZ15} or dynamic groups \cite{LLM16, LNH17}.
There have also been a few attempts from other classes of post-quantum assumptions, such as code-based
assumptions \cite{ELL15} or hash-based assumptions \cite{BM18}.

\noindent\textbf{Quantum random oracle model.}
Cryptographic hash functions are often modeled as random functions with public query access, a framework known as the random oracle model (ROM), which has been successful in analyzing the security of practical signature schemes. However, in the presence of full-fledged quantum adversaries who can evaluate the hash function in superposition, extending this model to allow quantum query access becomes essential, leading to the quantum random oracle model (QROM). In post-quantum cryptography, analyses under QROM are now considered crucial. In fact, \cite{ambainis2014quantum} has demonstrated that signature schemes proven secure in the  CROM can still be vulnerable to quantum attacks. Despite this, most ring and group signatures are often found in the CROM, with only a few lattice-based constructions offering QROM security \cite{csahin2022constant}. Given that most QROM-secure group signatures are lattice-based, if lattice assumptions were to be broken, there would be nearly no QROM-secure group or ring signatures available. Therefore, exploring other post-quantum assumptions for ring and group signatures becomes important, particularly in the context of QROM security. This leads to the main question of this work:

\ifnum\submission=1\vspace{-1em}\fi\begin{quote}
    {\it Can we construct provably QROM-secure group signatures and accountable ring signatures from isogenies?}
\end{quote}
\ifnum\submission=1\vspace{-2em}\fi


\subsection{Our results}
\ifnum\submission=1\vspace{-1em}\fi
We construct the first provably QROM-secure accountable ring signatures (ARS) from {\em isogeny-based assumptions}. Moreover, since ARS can be easily transformed into group signatures and ring signatures\footnote{ARS$\Rightarrow$RS is trivial by throwing away the opening functionality.} while preserving its QROM security, we also achieve the first provably QROM secure group signature and ring signature. For convenience, since QROM-security also implies CROM-security, we will primarily present and phrase our discussion in terms of QROM.
On top of it, we introduce a primitive for constructing ARS called {\em openable sigma protocol}, which is simple and fits well with the Fiat-Shamir methodology: Any openable sigma protocol $\Sigma$ can be securely transformed into an ARS ${\cal ARS}^t_\Sigma$ using Fiat-Shamir transformation. Furthermore, we show that a typical requirement for QROM-secure signatures, the {\em unique-response property}, suffices to provide QROM security to the abovementioned ARS. Note that it is non-trivial to obtain (or avoid relying on) the unique-response property, and we will direct the readers to our technical overview in Section~\ref{sec:technical_overview} for further details.

\ifnum\submission=1\vspace{-1em}\fi \begin{theorem}{(Informal)}\label{informal}
    Let $\Sigma$ be a secure openable sigma protocol. Then ${\cal ARS}_{\Sigma}^t$ is a classically secure ARS. Furthermore, if $\Sigma$ is perfect unique-response, then ${\cal ARS}_\Sigma^t$ is QROM-secure.
\end{theorem}

We base our construction on the {\em decisional CSIDH assumption} (D-CSIDH).
From an abstract viewpoint, D-CSIDH is a natural generalization of DDH that
is built on the weaker group-action structure.
Due to the lack of homomorphic properties in group-action assumptions,
it is usually infeasible to transform results obtained from group-based assumptions to those from group-action-based assumptions. Our work demonstrates the
possibility of constructing advanced cryptographic primitives with group-action-based assumptions despite its limited expressiveness. For future works and open problems, please refer to Supplementary~\ref{sec:open_problem} for more detail.

\paragraph{\bf Related Work}
A line of works \cite{beullens2020calamari, lai2021collusion, BDKLP22} result in efficient isogeny-based ring and group signatures in CROM. However, they are also affected by a recently discovered flaw \cite{FSwA23, FSwA23_Devevey} in the analyses of the Fiat-Shamir with Aborts (FSwA) paradigm. \footnote{We furthermore confirmed with Yi-Fu Lai, one of the co-authors of \cite{lai2021collusion,BDKLP22}, that it is indeed affected by the FSwA flaw.} While \cite{FSwA23, FSwA23_Devevey} have proposed patches for this flaw, the aforementioned isogeny-based group signature schemes are neither covered nor easily fixed.

Moreover, the works \cite{beullens2020calamari, lai2021collusion, BDKLP22} do not seem to achieve QROM security primarily because their underlying sigma protocols violate the so-called collapsing property, which is required for most currently available quantum rewinding techniques, e.g. \cite{Unr12,chiesa2022post}. Specifically, in the underlying sigma protocols of these works, a prover can produce two distinct responses for the same challenge—particularly when the challenge is zero. 
An alternative strategy is to rely on the so-called Unruh's transformation, a method used to enforce the unique-response property by having the prover commit to a single response per challenge. This does not seem directly applicable to \cite{lai2021collusion, BDKLP22} either due to the extra RO queries involved in the underlying sigma protocol. Therefore, even if the aforementioned FSwA flaws in these works are patched in the future, 
an upgrade to QROM security remains non-trivial. 
We refer readers to \cref{sec:hardnessQROM} and \cref{Tech:FSwA} for more details.

\section{Technical overview}\label{sec:technical_overview}

\ifnum\submission=1 \vspace{-1em}\fi In this overview, we assume some familiarity for sigma protocols and the Fiat-Shamir transformation \cite{fiat1987how}.
\ifnum\submission=1\vspace{-1em}\fi

\subsection{Why is it hard to obtain QROM-security?}\label{sec:hardnessQROM}

To start off, we explain why one needs to be extra careful when analyzing security proofs in QROM. In Fiat-Shamir paradigm, one often starts with constructing a so-called sigma protocol, where an adversary can make random challenges for the prover to respond. Then, the interaction is removed by substituting the verifier's random challenges with hashes of earlier transcripts.

In proving security for this kind of construction, two commonly used techniques are {\em rewinding} and {\em reprogramming}. We briefly describe their use-case as follows.
\begin{enumerate}[label=(\arabic*)]
    \item The properties of a sigma protocol $\Sigma$ can often be lifted to its corresponding Fiat-Shamir signatures, which run almost as $\Sigma$, except challenges are computed as the hashes (partly) of to-be-signed messages. An adversary may therefore bias the challenge distribution, by post-selecting the messages depending on previous hash queries. One typically resorts to the reprogramming argument to control such bias.
    \item A common way to argue the underlying sigma protocol indeed proves the knowledge of a certain relation, is by means of running a prover multiple times, and each time rewinding to the point where fresh challenges are about to be made, in order to collect its responses for many different challenges.
\end{enumerate}
The reprogramming argument typically involves reasoning about earlier query inputs. For example, in the classical setting, one might argue that a certain random variable $x$ has likely not been queried to the random oracle, except with a small probability. However, in the quantum setting, where queries can be made in superposition, it is not possible to read out the previously queried inputs due to the no-cloning principle, nor to measure them, as this may disturb the running states and be noticed by the adversary.
Similar challenges arise in rewinding arguments, where we rerun the same prover multiple times. Each time a response is generated, the prover's running state is disrupted, affecting subsequent executions.

\noindent\textbf{Mitigating the unique-response property?}\\
Unruh \cite{Unr12} demonstrated how to perform the quantum rewinding argument with an additional so-called collapsing property. The collapsing property is a slight relaxation of the unique-response property, which turns out necessary in the currently available quantum rewinding technique, e.g., \cite{chiesa2022post}. In particular, if the prover can produce two distinct responses for the same challenge, as in \cite{beullens2020calamari,lai2021collusion,BDKLP22}, then the collapsing property is violated.

A natural question is: Can we easily "uniquify" their responses for a construction lacking collapsing property? The answer is that it depends. There is a rule of thumb along the lines of the so-called Unruh's transformation \cite{DFMS19,Unr12}. The prover commits to one response for each challenge and publishes it in the first message, which the verifier can open and check in the later phase. However, the situation is often more complex, and we need to carefully examine whether specific constructions can be modified to achieve security in the QROM.


The primary obstacles preventing existing group and ring signatures \cite{beullens2020calamari,lai2021collusion,BDKLP22} from achieving QROM security\footnote{
As noted in \cref{Tech:FSwA}, the security claims in \cite{beullens2020calamari,lai2021collusion,BDKLP22} are compromised by the FSwA flaw \cite{FSwA23}, and even with future fixes, these constructions will still not achieve QROM security.} are the following two reasons: 
\begin{enumerate}[label=(\arabic*)]
\item Within their underlying sigma protocols, a prover is able to compute two distinct responses when the challenge is $0$. This violates the collapsing property, which is required in presently available quantum rewinding techniques, e.g., \cite{Unr12,chiesa2022post}.\\
\item Unruh's transformation \cite{unruh2015non} premises a sigma protocol in the plain model and is not directly applicable to \cite{lai2021collusion,BDKLP22} due to their extra hash commitments involved.
\end{enumerate}

\noindent\textbf{How about a polynomial number of valid responses?}\\
One thing worth mentioning is that there is a variant of Unruh's rewinding tolerating up to polynomially many responses (for the same challenge). However, this variant of rewinding still does not help either because typical constructions involve parallel repetition, which executes the same protocol multiple times independently in order to amplify soundness probability. The number of valid responses will then have an exponential blow-up as the number of repetitions scales, from which the QROM analysis falls apart.
\ifnum\submission=1\vspace{-1em}\fi
\subsection{A Technical Trailer}
\noindent\textbf{Signatures based on isogeny class group action.}
Stolbunov \cite{stolbunov2012cryptographic} attempted an isogeny-based signature scheme by applying the Fiat-Shamir transformation \cite{fiat1987how} to Couveignes' sigma protocol \cite{couveignes2006hard}. Though similar to the discrete log-based protocol by Chaum and van Heyst \cite{chaum91}, Couveignes' challenge space cannot be extended like Schnorr's protocol \cite{Schnorr91}, making parallel repetition necessary.

Later, following the proposal of an efficient class group action
implementation by CSIDH \cite{castryck2018csidh}, SeaSign
\cite{luca2019seasign}, and CSI-FiSh \cite{beullens2019csi} separately
gave efficient signature constructions based on Stolbunov's approach.
One main contribution of their works is that they overcome the lack of
canonical representation for elements in the class group $\Cl$. In
Stolbunuov's scheme, the signer would reveal $rs$ for $r\Usample\Cl$
and secret $s\in\Cl$. However, since $r$ and $s$ are represented as
element-wise bounded vectors in the CSIDH representation, a naive
representation for $rs$ {\em does not} hide the information of $s$. To
cope with this issue, SeaSign proposed a solution using the
Fiat-Shamir with abort technique \cite{Lyu09}, while CSI-FiSh computes
the whole class group structure and its relation lattice for a
specific parameter set, CSIDH-512. In this work, we will adopt the
latter approach, where we can simply assume canonical representation
for elements in $\Cl$.

Beullens, Katsumata, and Pintore \cite{beullens2020calamari} constructed an isogeny-based ring signature using a sigma protocol for an OR-relation. Similarly, our work begins with a sigma protocol that supports an opening operation. It takes $n$ statements and a master public key as inputs, proving one statement while embedding its "identity" into the transcript, which can be extracted with the master secret key. We first discuss how to embed information for opening the transcript.

\noindent\textbf{Embedding opening information.}
In a group signature scheme, the signer's identity must be embedded in the signature for the master to open it. One natural approach to embed opening information is to
encrypt the information with the master public key. Such an approach
is proven successful in a few previous works on group signatures
\cite{boneh2004short, bootle15}. However, since the opening
information is now a ciphertext under the master key, a verifier can only check the validity of the ciphertext via homomorphic operations
or NIZK. Unfortunately, unlike group-based assumptions, it is not yet
known how to achieve such homomorphic property from the weaker
\textit{group action structure} given by isogeny-based assumptions.
There is also no isogeny-based NIZK construction in the literature.
Thus, we must find a more straightforward way to
encode our opening information.

In light of this, we construct our opening functionality in a very
naive way. For a signature with group/ring size $n$ and a master
secret key $s_m$ for opening, we embed the signer identity by one DDH
tuple and $n-1$ dummies. Namely, the opening information is in the
form
\ifnum\submission=1\vspace{-1em}\fi
\begin{equation*}
  \tau = ((r_1E, r_2E, \dots, r_nE), r_kE_m), \text{ where } r_1,\dots, r_n\Usample G \text{ and } E_m=s_mE\;,
\end{equation*}
\ifnum\submission=1\vspace{-0.2 em}\fi
which embeds the signer's identity $k\in[n]$ through
\textit{position}, and is extractable for the manager holding $s_m$.
Note that such $\tau$ keeps all its elements in the form of curves/set
elements, hence the verifier can do further group action on $\tau$ for
consistency checking. This circumvents the previous difficulty, but
with the cost of a larger payload.

\noindent\textbf{Openable sigma protocol.}
To construct a group signature/accountable ring signature scheme
through Fiat-Shamir transformation, we first introduce an intermediate
primitive called openable sigma protocol. We refer the reader to
Section~\ref{sec:openablesigma} for more details.

The openable sigma protocol is similar to the standard OR sigma protocol, as both take 
$n$ statements and one witness as input. While the OR sigma protocol is a proof of knowledge for the OR-relation, the openable sigma protocol is a proof of knowledge
for the relation of the $k$th statement, where $k$ is chosen at the
proving stage and embedded in the first message $\mathsf{com}$,  which can be extracted using the master secret key $s_m$.

For our openable sigma protocol, the special soundness would thus
require an extractor that extracts the $k$th witness, which matches the
opening result. Such a stronger extractor is crucial for proving
unforgeability for group signatures, in which we transform a forger
for party $k$ into the extractor for the $k$th witness. Extractors for
standard OR sigma protocols cannot provide such reduction.

Also, unlike an OR sigma protocol, an openable sigma protocol cannot
get anonymity directly from the HVZK property, as the proving
statement is now embedded in $\mathsf{com}$. To achieve anonymity, we
need an extra property \textit{computational witness
  indistinguishability (CWI)}, which states that for an honest master
key pair $(\mathsf{mpk}, \mathsf{msk})$, the proof for the $i$th
the statement is indistinguishable from the proof for the $j$th
statement. This promises that when transformed into signatures, the signer would be anonymous as long as the manager has not colluded.

The construction of our openable sigma protocol is built on top of the
previous identity embedding component. For statements $E_1,\dots, E_n$
along with the $k$th witness $s_k$ s.t. $E_k=s_kE$ and the master key
pair $(s_m, E_m=s_mE)$, the opening information in our protocol is set
to
\ifnum\submission=1\vspace{-1 em}\fi
\begin{equation*}
  \tau = (E^\beta, E^{\mathsf{Open}}) = ((r_1E_1, r_2E_2, \dots, r_nE_n), r_ks_kE_m), r_1,\dots, r_n\Usample G
\end{equation*}
\ifnum\submission=1\vspace{-0.1em}\fi
As argued earlier, the manager can extract $k$ from $\tau$ with $s_m$.
To complete a proof of knowledge protocol, we use two challenges
($\mathsf{ch}=1,2$) to extract the knowledge of each $r_i$, and use
another two challenges ($\mathsf{ch}=3,4$) to extract ``some
$d=r_ks_k$ s.t. $dE\in E^\beta$ and $dE_m=E^{\mathsf{Open}}$.'' This
gives us a four-challenge openable sigma protocol with a corresponding
special soundness property.

We detail the full construction and security proofs in Section~\ref{sec:openablesigma}.

\noindent\textbf{Parallel repetitions and Fiat-Shamir transformation.}
From our 4-challenge sigma protocol with opening property, we
immediately obtain an identification scheme with a soundness error
$\frac{3}{4}$. In order to amplify away the soundness error, the designed
openable sigma protocol is executed in parallel repetitions. The parallel
repeated proof can be opened by taking the majority of the opening results
from each of its sessions.

It may be tempting to claim that we can achieve
soundness $(\frac{3}{4})^{\lambda}$ through a $\lambda$ repetition.
Unfortunately, this is not the case because each parallel session can
be independently generated with a different witness, and some of
the witnesses might be validly owned by the adversary. As a concrete
example, in a $\lambda$-parallel protocol, an adversary that owns 3
keys can generate $\lambda/4-1$ honest parallel sessions on behalf of
each key, and then cheat on only $\lambda/4+3$ sessions to achieve a
successful forgery.
Looking ahead, for an adversary owning $n_A$ keys, we would need
$n_A\cdot {\sf poly}(\lambda)$ repetitions to keep the error negligibly bounded, where the bound holds with its polynomial degree depending on the adversary being classical or quantum.

\ifnum\submission=0
With an identification scheme under a negligible soundness error, we can
now apply the Fiat-Shamir transformation and obtain a signature
scheme. This is done by substituting the verifier's random message with the hash of earlier transcripts. Classically, a $\mu$-special-sound protocol with some constant $\mu$, under parallel repetitions, preserves its proof of knowledge (PoK) property after being Fiat-Shamir transformed. This is done by adopting the so-called improved forking lemma, which rewinds a soundness adversary under the hood multiple times in order to collect multiple outputs and extract a secret from it. However, in the quantum setting, this does not work trivially. Each time measuring an output of a quantum adversary potentially corrupts its internal state and would require a different set of techniques for analysis.

A recently developed technique, measure-and-reprogram \cite{DFMS19,DFM20}, gives a non-trivial reduction from non-interactive PoK to the interactive one. Then, by means of generalized Unruh's rewinding, one can again rewind the interactive adversary and use it to extract a secret, assuming the {\em collapsing property} of the underlying protocol. Keeping the goal to construct a signature in mind, at some point, we need to reduce security against {\em chosen-message attacks} to {\em no-message attacks}. The crucial part of such reduction lies in the simulation of the signing oracle. One typically needs to reprogram the random oracle and argue that such reprogramming is not noticeable by the adversary except with some small probability. Assuming the reprogrammed places are of high min-entropy, a recent technique {\em adaptive reprogramming} \cite{grilo2021tight} helps one show that such reprogramming is not noticeable, even if the distribution of reprogrammed places is chosen by the quantum adversary on-the-fly.

Now, it becomes retrospectively clear why an openable sigma protocol should be defined in such a way that the index $k$ of the proven statement is embedded in its first message ${\sf com}$, and not in the responses:  only then, obtaining $\mu$ accepted responses to distinct challenges at the same position $i\in[t]$ implies extraction because they share the same first message and therefore the same opening result. In order to obtain these $\mu$ responses, the abovementioned QROM tools are used in a non-blackbox manner. In particular, we are taking advantage of the fact that the extracted responses are subject to uniformly and independently chosen challenges in each repetition and rewinding. We can then fix the sessions that open to the desired $k$, and argue about the probability where $\mu$ distinct challenges occur at the same position.
\fi

\ifnum\submission=1\vspace{-1em}\fi
\section{Preliminary}
\ifnum\submission=1\vspace{-1em}\fi

\subsection{Conventions on efficiency} \label{sec:L_eff}

For the rest of this paper, an algorithm always takes (the unary representation of) an security parameter $\lambda$ as input, which we may keep implicit. Let $\Lcal$ be a certain {\em polynomially-closed} class of non-negative functions $\Lcal\supseteq{\sf poly}(\lambda)$, i.e.
$$
f(\lambda)^C + C\in\Lcal\;,
$$
for every $f(\lambda)\in\Lcal$ and $C>0$. We define an algorithm to be {\em efficient}, if the running time falls into $\Lcal$. In addition, for an oracle algorithm, we further require that it makes at most ${\sf poly}(\lambda)$ many queries. Likewise, a hard problem is to be understood subject to such a notion of efficiency. Note that $\Lcal$ may be strictly bigger than ${\sf poly}(\lambda)$, giving rise to a stronger hardness assumption, e.g. we may assume that breaking CSI-FiSh takes up sub-exponential time, which is super-polynomially larger than the time to evaluate a CSI-FiSh group action, which is unfortunately not polytime (see \url{https://yx7.cc/blah/2023-04-14.html} for more detail).

\subsection{Isogeny and class group action}
\label{sec:classgroupaction}
\ifnum\submission=1\vspace{-1em}\fi
At the bottom level of our construction is the so-called {\em isogeny class group action}, which considers a commutative class group $\Cl$ acting on the set of supersingular elliptic curves $\Ecal\ell\ell_p(\Ocal,\pi_p)$ up to $\Fbb_{p}$ isomorphisms. The group action is free and transitive: for every $E_1,E_2\in\Ecal\ell\ell_p(\Ocal,\pi_p)$, there is exactly one ${\frak a}\in\Cl$ such that $E_2\cong_{\Fbb_p}{\frak a}E_1$. For the use of cryptography,
we note that computing the action is efficient while extracting ${\frak a}$ from the end-point curves is considered intractable. This introduces a hard-to-compute relation while regarding the curves as public keys and the group element ${\frak a}$ as secret. Note that validating the public key is efficient because it is efficient to validate the supersingularity of a curve. We refer readers to Supplementary~\ref{sec:CSIDH} for a guided walk-through.

\noindent\textbf{Hardness assumptions.} Hardness for the {\em group
  action inverse problem} (GAIP) in Definition~\ref{def:GAIP} is
commonly assumed for the above-mentioned group action, which has been
shown useful on constructing signature schemes such as CSI-FiSh
\cite{beullens2019csi} and SeaSign \cite{luca2019seasign}.

\begin{definition}[Group Action Inverse Problem (GAIP)]\label{def:GAIP} On inputs $E_1, E_2 \in \Ecal\ell\ell_p(\mathcal{O},\pi_p)$,
  find ${\frak a}\in\Cl$ such that $E_2 \cong_{\Fbb_p} {\frak a}\cdot
  E_1$.
\end{definition}

In this work, we need to assume hardness for a weaker problem, the
{\em decisional CSIDH problem} (abbreviated as D-CSIDH\footnote{This
  problem is called the decisional Diffie-Hellman group action problem
  (DDHAP) in \cite{stolbunov2012cryptographic}.}) in
Definition~\ref{def:DDH}, which was considered already in 
\cite{couveignes2006hard,stolbunov2012cryptographic}, and is
the natural generalization of the decisional
Diffie-Hellman problem for group actions.

\begin{definition}[Decisional CSIDH (D-CSIDH) / DDHAP]\label{def:DDH}
  For $E\in \Ecal\ell\ell_p(\mathcal{O},\pi_p)$, distinguish the two
  distributions
  \begin{itemize}
  \ifnum\submission=1\vspace{-0.5em}\fi
  \item $(E, {\frak a}E, {\frak b}E, {\frak c}E)$, where
    $\mathfrak{a,b,c}\Usample \Cl$,
  \item $(E, {\frak a}E, {\frak b}E, \mathfrak{ab}E)$, where
    $\mathfrak{a,b}\Usample \Cl$.
  \end{itemize}
\end{definition}
\ifnum\submission=1
    In order to enable the sampling of a random element from the class group (as in Definition~\ref{def:DDH}), we are adopting the setting of \cite{beullens2019csi}, in which a class group with specific parameters is precomputed.
\fi
\ifnum\submission=0
    We note that for typical cryptographic constructions such as CSIDH,
    additional heuristic assumptions are required to sample a random
    element from the class group (as in Definition~\ref{def:DDH}). This is
    because the ``CSIDH-way'' for doing this is by sampling exponents 
    $(e_1,\dots,e_n)$ satisfying $\forall i: |e_i|\leq b_i$,
    and the resulting distribution for ideals ${\frak
      l}_1^{e_1}\dots{\frak l}_n^{e_n}$ is generally non-uniform within
    $\Cl$. To get rid of such heuristics, one could instead work with
    specific parameters, where a bijective (yet efficient) representation
    of ideals is known. For instance, in \cite{beullens2019csi}, the
    structure of $\Cl$ is computed, including a full generating set of
    ideals ${\frak l}_1,\dots,{\frak l}_n$ and the entire lattice
    $\Lambda:=\{(e_1,\dots,e_n)| {\frak l}_1^{e_1}\dots{\frak
      l}_n^{e_n}=\mathsf{id}\}$. Evaluating the group action is just a
    matter of approximating a {\em closest vector} and then evaluating the
    residue as in CSIDH. In this work, we will be working with such a
    ``perfect'' representation of ideals, unless otherwise specified.
    
    As a remark, we note that the D-CSIDH problem for characteristic
    $p=1\mod 4$ is known to be broken \cite{castryck2020breaking}.
    Nevertheless, the attack is not applicable to the standard CSIDH
    setting where $p=3\mod 4$.
\fi
\ifnum\submission=1\vspace{-0.2em}\fi
\subsection{Group action DDH}

In this section, we give an abstract version of the CSIDH group
action. Such formulation will simplify our further construction and
security proof.

A commutative group action $\mathcal{GA}_\lambda = (G_\lambda, \Ecal_\lambda)$
with security parameter $\lambda$ (we will omit the subscripts for
simplicity) is called a DDH-secure group action if the following
holds:
\ifnum\submission=1\vspace{-1em}\fi
\begin{itemize}
\item $G$ acts freely and transitively on $\Ecal$.
\item DDHAP is hard on $\mathcal{GA}_\lambda$. i.e., for any efficient
  adversary $A$ and $E\in\Ecal$, the advantage for $A$ distinguishing
  the following two distributions is $\mathsf{negl}(\lambda)$.
  \begin{itemize}
  \item $(E, aE, bE, cE)$, $a,b,c\Usample G$
  \item $(E, aE, bE, abE)$, $a,b\Usample G$
  \end{itemize}
\end{itemize}
\ifnum\submission=1\vspace{-1em}\fi
As a side remark, the GAIP problem is also hard on a DDH-secure group
action.

For a DDH-secure group action, we can also have a natural parallel
extension for DDHAP. Such extension is also discussed in
\cite{kaafarani2020lossy}.

\ifnum\submission=1
\begin{samepage}
\fi

\ifnum\submission=1\vspace{-0.6em}\fi
\begin{definition}[Parallelized-DDHAP (P-DDHAP)] Given $E\in \Ecal$,
  distinguish the two distributions
  \ifnum\submission=1\vspace{-0.8em}\fi
  \begin{itemize}
  \item $(aE, \pset{b_iE}{i}{[m]}, \pset{c_iE}{i}{[m]})$, where
    $a,\aset{b}{i}{[m]},\aset{c}{i}{[m]}\Usample G$,
  \item $(aE, \pset{b_iE}{i}{[m]}, \pset{ab_iE}{i}{[m]})$, where
    $a,\aset{b}{i}{[m]}\Usample G$.
  \end{itemize}
\end{definition}
By a simple hybrid argument, we can easily see that if $DDHAP$ is
$\epsilon$-hard, then P-DDHAP is $m\epsilon$ hard. 
To see this, note that a single DDHAP can be turned into
a P-DDHAP as $(aE, \pset{r_ibE}{i}{[m]},\pset{r_icE}{i}{[m]})$
for $\aset{r}{i}{[m]}\Usample G$.

In the following we will use this in the form $(aE, \pset{b_iE}{i}{[m]},
\pset{c_iE}{i}{[m]})\approx_c (aE, \pset{c_ia\inv E}{i}{[m]},
\pset{c_iE}{i}{[m]})$.

\ifnum\submission=1
\end{samepage}
\fi

\ifnum\submission=1\vspace{-1em}\fi
\subsection{Sigma protocol}
\ifnum\submission=1\vspace{-1em}\fi
\label{sec:sigma}
A sigma protocol is a three-message public coin proof of knowledge protocol. For interested readers, we provide a brief introduction in Supplementary~\ref{sec:sigpro}.
\ifnum\submission=1\vspace{-1.5 em}\fi
\subsection{The forking lemma}
\label{sec:forking}
\ifnum\submission=1\vspace{-1em}\fi
A sigma-protocol-based signature naturally allows witness extraction
from the special soundness property. By extracting the witness
from signature forgeries, one can reduce the unforgeability property
to the hardness of computing the witness. However, the main gap
between special soundness and unforgeability is that special soundness
needs multiple related transcripts to extract the witness, while a
signature forging adversary only provides one. The forking lemma
\cite{pointcheval1996security} is thus proposed to close this gap. For our particular application, as elaborated in Supplementary~\ref{sec:forkl}, a generalized variant is adopted for the classical analysis.

\ifnum\submission=1\vspace{-1.5em}\fi
\subsection{Group signature}
\ifnum\submission=1\vspace{-1em}\fi
A group signature scheme consists of one manager and $n$ parties. The
manager can set up a group and provide secret keys to each party.
Every party is allowed to generate signatures on behalf of the whole
group. Such signatures are publicly verifiable without
revealing the corresponding signers, except the manager can open signers'
identities with his master's secret key. We refer readers to Supplementary~\ref{sec:gsig} for group signature syntax and formal definitions.

\ifnum\submission=1\vspace{-1 em}\fi
\subsection{Accountable ring signature}
\ifnum\submission=1\vspace{-1em}\fi
Accountable ring signatures (ARS) are a natural generalization for
both group and ring signatures. Compared to a group signature, ARS gives the power of group decision to the signer. On
signing, the signer can sign for an arbitrary group (or ring, to fit
the original naming) and can decide on a master independent of the
choice of the group. The master can open the signer's identity
among the group without needing to participate in the key generation
of parties in the ring. Note that accountable ring signatures directly
imply group signatures simply by fixing the group and the master
party at the key generation step. Thus, ARS can be viewed as a
more flexible form of group signature.

\noindent\textbf{Syntax.} An accountable ring signature scheme
$\mathcal{ARS}$ is associated with the following sets
$\Mcal$, $\Kcal_m$, $\Kcal$, $\Kcal\Pcal_m$, $\Kcal\Pcal$ and (efficient) algorithms
$\MGen$, $\Gen$, $\Sign$, $\Vrfy$, $\Open$, as elaborated below.

\begin{itemize}
\item $\MGen(1^\lambda)\rightarrow (\mathsf{mpk},\mathsf{msk})\in\Kcal\Pcal_m$
  generates a master public-secret key pair.
\item $\Gen(1^\lambda)\rightarrow (\pk,\sk)\in\Kcal\Pcal$ generates a public
  key-secret key pair for a ring member.
\item $\Sign(\mathsf{mpk}, S, m, \sk
  )\rightarrow \sigma$, for a message $m\in\Mcal$, a finite set of public keys $S\subset_{\sf fin}\Kcal$ with the existence of $\pk\in S$ such that $(\pk,\sk)\in\Kcal\Pcal$, generates a signature $\sigma$.
\item $\Vrfy(\mathsf{mpk},S , m, \sigma) \rightarrow {\sf acc}\in\{0,1\}$ verifies whether the signature is valid.
\item $\Open(\mathsf{msk}, S, m, \sigma) \rightarrow
  \pk\in S\cup\{\perp\} $ reveals an identity $\pk$, which presumably should be the public key of the signer of $\sigma$. It outputs
  $\pk=\perp$ when the opening fails, (e.g. when $\sigma$ is
  malformed).
\end{itemize}
\ifnum\submission=1\vspace{-1em}\fi
We refer to $\Mcal$ as the message space, $\Kcal_m$ as the master public key space
and $\Kcal$ as the public key space. We
also define $\Kcal\Pcal_m$ to be the set of all master key pairs
$(\mathsf{mpk}, \mathsf{msk})$, and $\Kcal\Pcal$ to be the set of all
public-private key pairs $(\pk,\sk)$. For simplicity, we keep the
parameter $\lambda$ implicit for the before-mentioned key spaces, and
additionally require public keys to be all distinct for a set $S$ of size
$|S|\leq \mathsf{poly}(\lambda)$.

An accountable ring signature scheme should satisfy the following
security properties.

\noindent\textbf{Correctness.} An ARS is said to be correct if every
honest signature can be correctly verified and opened.
\begin{definition}
  An accountable ring signature scheme $\mathcal{ARS}$ is correct if
  for any master key pair $(\mathsf{mpk},
  \mathsf{msk})\in\Kcal\Pcal_m$, any key pair
  $(\pk,\sk)\in\Kcal\Pcal$, and any set of public keys $S$ such that
  $\pk\in S$,
  \begin{equation*}
        \Pr\left[
        \substack{
        \mathsf{acc}=1\land \mathsf{out}=\pk
        } \middle| 
        \substack{
        \sigma\leftarrow\Sign(\mathsf{mpk}, S, m, \sk),\\
        \mathsf{acc} \leftarrow \Ver(\mathsf{mpk}, S, m, \sigma),\\
        \mathsf{out} \leftarrow\Open(\mathsf{msk}, S, m, \sigma)
        } \right]> 1-\mathsf{negl}(\lambda).
    \end{equation*}
\end{definition}

\noindent\textbf{Anonymity.} An ARS is said to be anonymous if no
adversary can determine the signer's identity within the set of
signers of a signature without using the master secret key.
\begin{definition}\label{def:anonymity}
  An accountable ring signature scheme $\mathcal{ARS}$ is anonymous if
  for any efficient adversary $A$ and any two key pairs $(\pk_0, \sk_0),
  (\pk_1,\sk_1)\in \Kcal\Pcal$,
    \begin{equation*}
    \left|
    \Pr\left[1\leftarrow A^{\Sign^*(\mathsf{mpk}_\bullet,\bullet, \bullet, \sk_0),{\sf mpk}_\bullet}(x)\right]
    - \Pr\left[1\leftarrow A^{\Sign^*(\mathsf{mpk}_\bullet,\bullet, \bullet, \sk_1),{\sf mpk}_\bullet}(x)\right]
    \right|\leq \mathsf{negl}(\lambda)\;,
    \end{equation*}
    with each query $\Sign^*(\mathsf{mpk}_\nu, S, m, \sk_b)$ returning an honest signature only when both $\pk_0, \pk_1\in S$ and otherwise abort, where each master key pairs $(\mathsf{mpk}_\nu,\mathsf{msk}_\nu)\leftarrow\MGen(1^\lambda)$ are sampled honestly.
\end{definition}

\ifnum\submission=1\vspace{-1em} \fi
\begin{remark}
As we do not forbid $x$ to contain information about the secret keys, adversaries in Definition~\ref{def:anonymity} are referred to as being under the full key exposure.
\end{remark}
\ifnum\submission=1\vspace{-1em} \fi
\noindent\textbf{Unforgeability.} An ARS is said to be unforgeable if
no adversary can forge a valid signature that fails to open or opens
to some non-corrupted party, even if the manager has also
colluded. We model this property with the unforgeability game
$G_{A}^{\mathsf{UF}}$ as defined below. Note that below we quantify over $A$ that may have the master key pairs $(\mathsf{mpk},\mathsf{msk})$ hard-coded within.

\begin{algorithm}[H]
\caption{$G_{A}^{\mathsf{UF}}({\sf mpk},{\sf msk})$: Unforgeability game}
\begin{algorithmic}[1]
    \STATE $(\pk,\sk)\gets\Gen(1^\lambda)$
    \STATE $(S^*, m^*, \sigma^*)\leftarrow A^{\textbf{Sign}(\bullet, \bullet,
    \bullet, \sk), H}(\pk)$
    \STATE {\bf check} $\sigma^*$ is not produced by querying $\textbf{Sign}({\sf mpk},S^*,m^*,\sk)$
    \STATE {\bf check} $1\leftarrow \Vrfy(\mathsf{mpk}, S^*, m^*, \sigma^*)$
    \STATE {\bf check} $\pk\text{ or }\bot\leftarrow\Open(\mathsf{msk}, S^*, m^*, \sigma^*)$
    \STATE $A$ wins if all check pass
\end{algorithmic}
\end{algorithm}

\ifnum\submission=1\vspace{-2.5em}\fi
\begin{definition}
  An accountable ring signature scheme $\mathcal{ARS}$ is unforgeable
  if for any efficient adversary $A$, any valid master key pair
  $(\mathsf{mpk}, \mathsf{msk})\in\Kcal\Pcal_m$
    \begin{equation*}
        \Pr[A\ \text{wins}\ G_{A}^{\mathsf{UF}}(\mathsf{mpk}, \mathsf{msk})] < \mathsf{negl}(\lambda).
    \end{equation*}
\end{definition}
\ifnum\submission=1\vspace{-1em}\fi
\noindent\textbf{Transforming ARS to GS.}
\label{sec:ARStoGS}
As mentioned earlier, an accountable ring signature can be viewed as a
generalization of a group signature. We give here the general
transformation from an ARS scheme $\mathcal{ARS}$ to a group signature
scheme $\mathcal{GS}^\mathcal{ARS}$.

The algorithms of the group signature scheme
$\mathcal{GS}^\mathcal{ARS}$ are detailed as follows:

\begin{itemize}
    \item $\mathbf{GKeygen}(1^\lambda,1^n)$:
    \begin{algorithmic}[1]
        \STATE $(\mathsf{mpk}, \mathsf{msk})\leftarrow \mathcal{ARS}.\MGen(1^\lambda)$
        \STATE $\forall i\in[n], (\pk_i, \sk_i)\leftarrow\mathcal{ARS}.\Gen(1^{\lambda})$ and $S=\nset{\pk}{i}$
        \RETURN $(\mathsf{gpk}=({\sf mpk, S}), \nset{\sk}{i}, \mathsf{msk})$
    \end{algorithmic}
    
    \item $\mathbf{GSign}(\mathsf{gpk}=(\mathsf{mpk}, S),m,\sk_k)$
    \begin{algorithmic}[1]
    	\RETURN $\sigma\leftarrow\mathcal{ARS}.\Sign(\mathsf{mpk}, S, m, \sk_k)$
    \end{algorithmic}
    
    \item $\mathbf{GVerify}(\mathsf{gpk}=(\mathsf{mpk}, S),m,\sigma)$:
    \begin{algorithmic}[1]
    	\RETURN $\sigma\leftarrow\mathcal{ARS}.\Ver(\mathsf{mpk}, S, m, \sigma)$
    \end{algorithmic}
    
    \item $\mathbf{GOpen}(\mathsf{gpk}=(\mathsf{mpk}, S),\mathsf{msk}, m, \sigma)$: 
    \begin{algorithmic}[1]
    	\STATE $\pk\leftarrow\mathcal{ARS}.\Open(\mathsf{msk}, S, m, \sigma)$
    	\RETURN $k$ s.t. $\pk = \pk_k\in S$ or $\perp$ otherwise
    \end{algorithmic}
\end{itemize}

Note that the transformation only changes the formulation of the setup
stage. Thus, the security properties from $\mathcal{ARS}$ transfer
directly to the induced group signature scheme
$\mathcal{GS}^\mathcal{ARS}$.

\ifnum\submission=1\vspace{-1.5 em}\fi
\section{Openable sigma protocol}
\label{sec:openablesigma}
\ifnum\submission=1\vspace{-1.5 em}\fi

In this section, we will introduce the openable sigma protocol, which
is an intermediate primitive toward group signatures and accountable
ring signatures. We will first give some intuition on how we formulate
this primitive, and then give a formal definition and construction
from DDH-hard group actions.

\ifnum\submission=1\vspace{-2em}\fi
\subsection{Intuition}
\ifnum\submission=1\vspace{-1em}\fi
Typical construction of a Fiat-Shamir-based signature starts from a
sigma protocol. As introduced in Section~\ref{sec:sigma}, the three
message protocol $(\mathsf{com}, \mathsf{ch}, \mathsf{resp})$ only
requires special soundness, which is, informally speaking, weaker than
the unforgeability property in the sense that multiple transcripts are
required in order to break the underlying hardness. The forking lemma
closes this gap with the power of rewinding and random oracle
programming. As stated in Section~\ref{sec:forking}, the lemma takes a
forger that outputs a single forgery and gives an algorithm that
outputs multiple instances of valid $(\mathsf{com}, \mathsf{ch}_j,
\mathsf{resp}_j)$'s. This gives a transformation from a signature
breaker to a witness extractor, bridging the two security notions.
	
For our accountable ring signature, we thus plan to follow the
previous roadmap. We design a sigma protocol that supports an extra
``opening'' property. Our openable sigma protocol takes $n$ statements
as input and additionally requires the prover to take a master public
key $\mathsf{mpk}$ as input on generating the first message
$\mathsf{com}$. The function $\Open$, with the master secret key
$\mathsf{msk}$, can then extract the actual statement to which the
proving witness corresponds. For a $\mathsf{com}$ generated from
statement $(x_1, \dots, x_n)$ and witness $w_i$ with $(x_i, w_i)\in
R$, we have $x_i = \Open(\mathsf{com}, \mathsf{msk})$. As our
target is a signature scheme, $(x_i, w_i)$ would be set to public
key/secret key pairs, and thus the open function outputs the
signer's identity.

To achieve the stronger security property of ARS after the Fiat-Shamir
transformation, our openable sigma protocol needs to have
modified security properties correspondingly. For \textit{special
  soundness}, we would not be satisfied with extracting only ``one of
the witnesses''; instead, we need to build an extractor that extracts
a witness which matches the opening result. Such a stronger extractor
will allow us to extract secret keys from adversaries that can
impersonate other players. For \textit{honest
  verifier zero knowledge} (HVZK), we require the transcript to
be ZK even when given the master secret key $\mathsf{msk}$. This is
crucial for proving that the impersonating attack cannot succeed even
with a corrupted manager. Note that when given $\mathsf{msk}$, one
cannot hope to hide the signer's identity, so we only require ZK against
the signer's witness. The formulation for the HVZK simulator thus
takes the signer identity as input. Finally, we need an extra
property to provide anonymity for the signer, which we named 
\textit{computational witness indistinguishability} (CWI). CWI
requires that, given honest master key pairs, the transcript generated
from two different witnesses/identities should be indistinguishable.
This property is formulated as the indistinguishability of two signing
oracles.

\ifnum\submission=1\vspace{-2em}\fi
\subsection{Definition}
\ifnum\submission=1\vspace{-1em}\fi

An openable sigma protocol $\Sigma$ is defined with respect to two relations. A base relation
$R\subset X\times W$ and an opening relation $R_m$. Each ${\cal R}\in\{R, R_m\}$ of the both relations is efficiently samplable with respect to some distribution, but for a fresh sample $(x,s)\gets {\cal R}(1^\lambda)$, it is hard to derive the witness $s$ from the statement $x$ as the security parameter $\lambda$ scales. We will keep $\lambda$ implicit for convenience if the context is clear. Additionally, we define the OR-relation for $R$, i.e. $(\nset{x}{i},
s)\in R_n$ if and only if all $x_i$ are distinct and $\exists i\in[n]$
s.t. $(x_i,s)\in R$. The openable sigma protocol $\Sigma$ contains the following four algorithms.
\ifnum\submission=1\vspace{-0.5em}\fi
\begin{itemize}
\item $\Com(x_m, \nset{x}{i}, s)\rightarrow (\mathsf{com},st)$
  generates a commitment $\mathsf{com}$ based on $(\nset{x}{i}, s)\in
  R_n$. $\Com$ also generates a state $st$ which is shared with
  $\Resp$ and will be kept implicit for convenience.

\item $\Resp(x_m, \nset{x}{i}, s, \mathsf{com}, \mathsf{ch},
  st)\rightarrow \mathsf{resp}$ computes a response $\mathsf{resp}$
  relative to a challenge $\mathsf{ch}\Usample\Ccal$.
\item $\Ver(x_m, \nset{x}{i}, \mathsf{com}, \mathsf{ch},
  \mathsf{resp})\rightarrow 1/0$ verifies whether a tuple
  $(\mathsf{com}, \mathsf{ch}, \mathsf{resp})$ is valid. $\Ver$
  outputs $1$ if the verification passes and $0$ otherwise.

\item $\Open(s_m, \nset{x}{i}, \mathsf{com}) \rightarrow x\in
  \nset{x}{i}\cup \{\perp\} $ reveals some $(x,s)\in R$, where $s$ is
  the witness used to generate the commitment $\mathsf{com}$. It
  outputs $x=\perp$ when the opening fails. (i.e. when $\mathsf{com}$
  is malformed)
\end{itemize} 
An openable sigma protocol is secure if it is high min-entropy, computational unique-response, correct, $\mu$-special sound for some constant $\mu$ and statistical honest-verifier zero-knowledge, as defined below.
\begin{definition}[High min-entropy]
  An openable sigma protocol $\Sigma$ is of high min-entropy if the for any possible commitment ${\sf com}_0$
  $$\Pr[\Com(x)\rightarrow{\sf com}={\sf com}_0]\leq {\sf negl}(\lambda)\;.$$
\end{definition}

\begin{definition}[Unique-response property]
    An openable sigma protocol $\Sigma$ is computational unique-response if for every $(x_m,s_m)\in R_m$ and every efficient algorithm $A$
    $$
    \Pr_{(x,s)\gets R(1^\lambda)}\left[\substack{(S,{\sf com},{\sf ch},{\sf resp}_1,{\sf resp}_2))\gets A(x)\\
    \forall i\in[2]:1\gets \Vrfy(x_m,S,{\sf com},{\sf ch}_i,{\sf resp}_i)\\
    \pk\text{ or }\bot\gets \Open(s_m,S,{\sf com})\\
    {\sf resp}_1\neq {\sf resp}_2}\right]\leq{\sf negl}(\lambda)\;.
    $$
    Furthermore, $\Sigma$ is called perfect unique-response if for every $x_m, S, {\sf com}, {\sf ch}$ there is at most one ${\sf resp}$ such that $1\gets\Vrfy(x_m,S,{\sf com},{\sf ch},{\sf resp})$.
\end{definition}

\begin{definition}[Correctness]
  An openable sigma protocol $\Sigma$ is correct if
  for all $n=poly(\lambda)$, $(x_m, s_m)\in R_m$, $(\nset{x}{i}, s)\in
  R_n$, $\mathsf{ch}\in\Ccal$, and $x\in\nset{x}{i}$ such that $(x,
  s)\in R$,
  \begin{equation*}
    \Pr\left[
      \mathsf{acc}=1\wedge id=x\middle|
      \substack{
        \mathsf{com} \leftarrow \Com(x_m, \nset{x}{i}, s),\\
        \mathsf{resp} \leftarrow \Resp(x_m, \nset{x}{i}, s, \mathsf{com}, \mathsf{ch}),\\
        \mathsf{acc} \leftarrow \Ver(\{x_m, \nset{x}{i}, \mathsf{com}, \mathsf{ch}, \mathsf{resp}),\\
        id \leftarrow\Open(s_m, \nset{x}{i}, \mathsf{com})
      } \right] \geq 1-\mathsf{negl}(\lambda).
    \end{equation*}
\end{definition}

\begin{definition}[$\mu$-Special Soundness]
  An openable sigma protocol $\Sigma$ is
  $\mu$-special sound if for all $n=poly(\lambda)$ there exists an
  efficient extractor $\mathbf{Ext}$ such that, for all $(x_m, s_m)\in
  R_m$ and any $(\nset{x}{i} , \mathsf{com},
  \aset{\mathsf{ch}}{j}{[\mu]}, \aset{\mathsf{resp}}{j}{[\mu]})$ such that
  each $\mathsf{ch}_j\in\Ccal$ are distinct, then
    \begin{equation}\label{eq:openable_special_sound}
        \Pr\left[
        \substack{
        (\forall j\in[\mu],\ \mathsf{acc}_j=1)\land\\ (x=\perp\lor (x, s)\notin R)}\middle|
        \substack{
        \forall j\in\Ccal,\ \mathsf{acc}_j \leftarrow \mathbf{Ver}(x_m, \nset{x}{i}, \mathsf{com}, \mathsf{ch}_j, \mathsf{resp}_j),\\
        x \leftarrow\mathbf{Open}(s_m,\nset{x}{i}, \mathsf{com}),\\
        s \leftarrow\mathbf{Ext}(\nset{x}{i}, \mathsf{com}, \aset{\mathsf{ch}}{j}{[\mu]},  \aset{\mathsf{resp}}{j}{[\mu]})
        } \right] =0.
    \end{equation} 
\end{definition}

\begin{definition}[Statistical honest-verifier zero-knowledge / sHVZK]
  An openable
  sigma protocol $\Sigma$ is statistical HVZK if there exists an
  efficient simulator $\mathbf{Sim}$ such that, for any $x_m\in X_m$,
  any $(\nset{x}{i}, s)\in R_n$, and $x\in\nset{x}{i}$ such that $(x,
  s)\in R$,
    \begin{equation*}
        \Trans(x_m, \nset{x}{i}, s) \approx_s \mathbf{Sim}(x_m, \nset{x}{i}, x)
    \end{equation*}
    where $\mathbf{Trans}$ outputs honest transcript $(\mathsf{com},
    \mathsf{ch}, \mathsf{resp})$ generated honestly by $\Com$ and
    $\Resp$ with honestly sampled $\mathsf{ch}\Usample\Ccal$.
\end{definition}

\begin{definition}[Computational witness indistinguishability / CWI]\label{def:CWI}
  An openable sigma protocol $\Sigma$ is computational witness indistinguishable, if
  for any two $(x_i, s_i), (x_j, s_j) \in R$ and any efficient
  adversary $A$, with $x_m^\bullet(\nu)$ returning $x_m^\nu$ where $(x_m^\nu,s_m^\nu)\gets R_m$ is freshly sampled for each $\nu$, we have
     \begin{equation*}
        \left|
        \begin{aligned}
            \Pr\left[1\leftarrow A^{ \mathbf{Trans^*}(x_m^\bullet, \bullet, s_i),x_m^\bullet}(x)\right] 
            - \Pr\left[1\leftarrow A^{ \mathbf{Trans^*}(x_m^\bullet, \bullet, s_j),x_m^\bullet}(x)\right]
        \end{aligned}
        \right|\leq \mathsf{negl}(\lambda)
    \end{equation*}
    where $\mathbf{Trans}^*(\mathsf{mpk}_\nu, S, s_k)$ for whichever $k\in\{i,j\}$ returns an honest transcript $(\mathsf{com}, \mathsf{ch}, \mathsf{resp})$ tuple from $\Sigma$ if both $x_i, x_j\in S$ and aborts otherwise.
\end{definition}

\ifnum\submission=1\vspace{-2 em}\fi
\subsection{Construction}\label{sec:sigma_construction}
\ifnum\submission=1\vspace{-1em}\fi
Here, we give our construction to an openable sigma protocol
$\Sigma_{GA}$ for relations from our DDH-secure group action
$\mathcal{GA} = (G, \Ecal)$. We let $E\in\Ecal$ be some fixed
element in $\Ecal$. When implemented with CSIDH, we can choose the
curve $E_0:y^2=x^3+x$ for simplicity. Let the relation $R_E = \{(aE,
a)| a\in G\}\subset \Ecal\times G$.

For our $\Sigma_{GA}$, we set its opening and base relations $R_m=R=R_E$, with
the natural instance generator that samples $a\Usample G$ and outputs $(aE,a)$. For inputs $E_m\in \Ecal$ and $(\nset{E}{i}, s)\in R_n$ with
any $n=poly(\lambda)$, the algorithms for $\Sigma_{GA}$ are
constructed as follow.

\begin{itemize}
    \item $\Com(E_m, \nset{E}{i}, s)$
    \begin{algorithmic}[1]
        \STATE set $k\in [n]$ s.t. $(E_k, s) \in R$.
        \STATE $\nset{\Delta}{i}, \nset{\Delta'}{i}, b \Usample G$
        \STATE $\tau \Usample \textit{sym}(n)$ \COMMENT{$\tau$ is a random permutation}
        \STATE $\forall i\in[n] :  E_i^\alpha := \Delta_iE_i$
        \STATE $\forall i\in[n] :  E_i^\beta := \Delta'_iE_i^\alpha = \Delta_i\Delta'_iE_i$
        \STATE $\forall i\in[n] :  E_i^\gamma := bE_i^\beta = \Delta_i\Delta'_ibE_i$
        \STATE $E^{\mathsf{Open}} := \Delta_k\Delta'_ksE_m$
        \STATE $E^{\mathsf{Check}} := \Delta_k\Delta'_kbsE_m = bE^{\mathsf{Open}}$
        \STATE $\mathsf{st} = (\nset{\Delta}{i}, \nset{\Delta'}{i}, b, l=\Delta_k\Delta_k'bs)$
        \RETURN $(\mathsf{com}, {\sf st})=((\nset{E^\alpha}{i}, \nset{E^\beta}{i}, \tau(\nset{E^\gamma}{i}), E^{\mathsf{Open}}, E^{\mathsf{Check}}), {\sf st})$\\ \COMMENT{We use $\tau(\bullet)$ as a lazy convention of sending a permuted list}
    \end{algorithmic}
    
    \item $\Resp(E_m, \nset{E}{i}, s, \mathsf{com}, \mathsf{ch}, st)$:
    \begin{algorithmic}[1]
        \IF{$\mathsf{ch} = 1$}
        \RETURN $\mathsf{resp} := \nset{\Delta}{i}$
        \ENDIF
        \IF{$\mathsf{ch} = 2$}
        \RETURN $\mathsf{resp} := \nset{\Delta'}{i}$
        \ENDIF
        \IF{$\mathsf{ch} = 3$}
        \RETURN $\mathsf{resp} := b$
        \ENDIF
        \IF{$\mathsf{ch} = 4$}
        \RETURN $\mathsf{resp} := l= \Delta_k\Delta'_kbs$
        \ENDIF
    \end{algorithmic}
    
    \item $\Ver(E_m, \nset{E}{i}, \mathsf{com}, \mathsf{ch}, \mathsf{resp})$:
    \begin{algorithmic}[1]
        \RETURN 0 if $\nset{E}{i}$ or $\nset{E^\beta}{i}$ are not all distinct
        \IF{$\mathsf{ch} = 1$}
        \STATE $\mathbf{check\ } \forall i\in[n] :  E_i^\alpha = \Delta_iE_i$
        \ENDIF
        \IF{$\mathsf{ch} = 2$}
        \STATE $\mathbf{check\ } \forall i\in[n] :  \Delta'_iE_i^\alpha = E_i^\beta$
        \ENDIF
        \IF{$\mathsf{ch} = 3$}

        \STATE $\mathbf{check\ } \exists \tau'\in \textit{sym}(n)$ s.t. $\tau'(\nset{bE^\beta}{i}) = \tau(\nset{E^\gamma}{i})$
        \STATE $\mathbf{check\ } E^{\mathsf{Check}} = bE^{\mathsf{Open}}$
        \ENDIF
        \IF{$\mathsf{ch} = 4$}
        \STATE $\mathbf{check\ }  E^{\mathsf{Check}} = lE_m$
        \STATE $\mathbf{check\ }  \exists E^\gamma \in \tau(\nset{E^\gamma}{i})\ s.t.\ E^\gamma= lE$
        \ENDIF
        \RETURN 1 $\mathbf{\ if\ all\ checks\ pass}$
    \end{algorithmic}
    
    \item $\Open(s_m:=\mathsf{msk}, \nset{E}{i}:=\nset{\pk}{i}, \mathsf{com})$: 
    \begin{algorithmic}[1]
        \FOR{$i\in[n]$}
            \IF {$s_mE_i^\beta = E^{\mathsf{Open}}$}
                \RETURN $E_i$
            \ENDIF
        \ENDFOR
        \RETURN $\perp$
    \end{algorithmic}
\end{itemize}

\ifnum\submission=1\vspace{-1em}\fi
The construction of our openable sigma protocol looks complicated, but
the intuition is simple. The core section of the message
$\mathsf{com}$ is $(E^\beta, E^{\mathsf{Open}})$, which allows
opening. The other parts of $\mathsf{com}$ are to ensure that the
opening section is honestly generated. $E^\alpha$ along with the
challenge/response pair on $\mathsf{ch}=1,2$ allows extraction for
$\Delta_i\Delta'_i$'s, ensuring that $E^\beta$ is honestly generated.
$(E^\gamma, E^{\mathsf{Check}})$ along with the challenge/response pair
on $\mathsf{ch}=3,4$ verifies the relation between $E^\beta$ and
$E^{\mathsf{Open}}$. By using a permuted $E^\gamma$, the CWI property
is preserved through such a verification process. Combined together,
we complete the proof of knowledge protocol.

\begin{theorem}
    \label{Sigmasecure}
    $\Sigma_{GA}$ is an openable sigma protocol with $R_E$ being both
    the opening relation and the base relation
\end{theorem} 

\ifnum\submission=1\vspace{-2 em}\fi
\subsection{Security}\label{sec:openable_sigma_security}

The proof for Theorem~\ref{Sigmasecure} is broken down into proving
each of the required properties. First, by construction one immediately get $\Sigma_{GA}$ being perfect unique-response, and high min-entropy. It is also easy to show that $\Sigma_{GA}$ is correct and statistical HVZK (see Supplementary~\ref{sec:sigma_GA_proof} for full proof).

\begin{lemma}\label{lem:sigma_correct}
    $\Sigma_{GA}$ is \textbf{correct} and \textbf{statistical honest-verifier zero-knowledge}.
\end{lemma}


\begin{lemma}
   $\Sigma_{GA}$ is \textbf{4-special sound}.
\end{lemma}
\begin{proof}
  For any $E_m\in\Ecal$ and any $(\nset{E}{i} , \mathsf{com},
  \aset{\mathsf{resp}}{j}{\Ccal})$ where\\ $\mathsf{com} =
  (\nset{E^\alpha}{i}, \nset{E^\beta}{i}, \sigma(\nset{E^\gamma}{i})$,
  and $\aset{\mathsf{resp}}{j}{[4]} = (\nset{\Delta}{i},
  \nset{\Delta'}{i}, b, l)$. Suppose that $\forall j\in[4],
  1\leftarrow \mathbf{Ver}(E_m, \nset{E}{i}, \mathsf{com}, j,
  \mathsf{resp}_j)$, then by the definition of $\Ver$, we can get the
  following equations:\\
\begin{equation*}
    \begin{cases}
    \nset{E}{i}, \nset{E^\beta}{i} \text{ are both pairwise distinct sets}\\
    \forall i\in[n] : 
    E_i^\alpha = \Delta_iE_i, 
    E_i^\beta = \Delta'_iE_i^\alpha\\
    \exists \tau'\in \textit{sym}(n)\ s.t.\ 
    \tau'(\nset{bE^\beta}{i}) = \tau(\nset{E^\gamma}{i})\\
    \exists E^\gamma \in \tau(\nset{E^\gamma}{i})\ s.t.\ E^\gamma= lE\\
    E^{\mathsf{Check}} = lE_m= bE^{\mathsf{Open}}
    \end{cases}
\end{equation*}

Thus, there exists a unique $k\in [n]$ such that $lE = bE_k^\beta =
\Delta'_kbE_k^\alpha = \Delta_k\Delta'_kbE_k $, which means
$l(\Delta_k\Delta'_kb)\inv E = E_k$. This implies that $(E_k,
l(\Delta_k\Delta'_kb)\inv)\in R_E$. Furthermore, we also have
$E^{\mathsf{Open}} = b\inv lE_m = s_m b\inv lE = s_m E_k^\beta$. This
implies that $E_k\leftarrow\Open(s_m, \nset{E}{i}, \mathsf{com})$.
Thus $\Open$ does not output $\perp$.
From these observations, we can easily construct the extractor
$\Ext(\mathsf{com}, \aset{\mathsf{resp}}{j}{\Ccal})$, which simply
searches through $k\in [n]$ for $k$ satisfying $lE = bE_k^\beta$, then
output $s=l(\Delta_k\Delta'_kb)\inv$. This concludes the proof that
$\Sigma_{GA}$ is \textbf{4-special sound}.
\end{proof}

\ifnum\submission=1\vspace{-1.5em}\fi
\begin{lemma}
  \label{sigmaCWI}
  $\Sigma_{GA}$ is \textbf{computational witness indistinguishable} (assuming DDHAP
  is hard for $\mathcal{GA}$).
\end{lemma}
\ifnum\submission=1\vspace{-1em}\fi
Here we will finally use the fact that $\mathcal{GA}$ is DDH-hard. We
will prove this theorem through two hybrids. We highlight the changes
between $\Trans$ and $\Hyb_1$ and between $\Hyb_1$ and $\Hyb_2$ with
different colors for easier comparison.

\ifnum\submission=1\vspace{-1em}\fi

\begin{proof}
  For any $E_m\in\Ecal$ and any $(\nset{E}{i} , \mathsf{com},
  \aset{\mathsf{resp}}{j}{\Ccal})$ where\\ $\mathsf{com} =
  (\nset{E^\alpha}{i}, \nset{E^\beta}{i}, \sigma(\nset{E^\gamma}{i})$,
  and $\aset{\mathsf{resp}}{j}{[4]} = (\nset{\Delta}{i},
  \nset{\Delta'}{i}, b, l)$. Suppose that $\forall j\in[4],
  1\leftarrow \mathbf{Ver}(E_m, \nset{E}{i}, \mathsf{com}, j,
  \mathsf{resp}_j)$, then by the definition of $\Ver$, we can get the
  following equations:\\
\begin{equation*}
    \begin{cases}
    \nset{E}{i}, \nset{E^\beta}{i} \text{ are both pairwise distinct sets}\\
    \forall i\in[n] : 
    E_i^\alpha = \Delta_iE_i, 
    E_i^\beta = \Delta'_iE_i^\alpha\\
    \exists \tau'\in \textit{sym}(n)\ s.t.\ 
    \tau'(\nset{bE^\beta}{i}) = \tau(\nset{E^\gamma}{i})\\
    \exists E^\gamma \in \tau(\nset{E^\gamma}{i})\ s.t.\ E^\gamma= lE\\
    E^{\mathsf{Check}} = lE_m= bE^{\mathsf{Open}}
    \end{cases}
\end{equation*}

Thus, there exists a unique $k\in [n]$ such that $lE = bE_k^\beta =
\Delta'_kbE_k^\alpha = \Delta_k\Delta'_kbE_k $, which means
$l(\Delta_k\Delta'_kb)\inv E = E_k$. This implies that $(E_k,
l(\Delta_k\Delta'_kb)\inv)\in R_E$. Furthermore, we also have
$E^{\mathsf{Open}} = b\inv lE_m = s_m b\inv lE = s_m E_k^\beta$. This
implies that $E_k\leftarrow\Open(s_m, \nset{E}{i}, \mathsf{com})$.
Thus $\Open$ does not output $\perp$.
From these observations, we can easily construct the extractor
$\Ext(\mathsf{com}, \aset{\mathsf{resp}}{j}{\Ccal})$, which simply
searches through $k\in [n]$ for $k$ satisfying $lE = bE_k^\beta$, then
output $s=l(\Delta_k\Delta'_kb)\inv$. This concludes the proof that
$\Sigma_{GA}$ is \textbf{4-special sound}.
\end{proof}

\begin{lemma}
  \label{CWIlem1}
  For any $s\in G$, any efficient adversary $A$ with $E_m^\bullet(\nu)$ generating $E_m^\nu$ from $(E_m^\nu,s_m^\nu)\gets R_E$ for each $\nu$, we have
     \begin{equation*}
        \left|
        \begin{aligned}
            &\Pr\left[1\leftarrow A^{ \mathbf{Trans}(E_m^\bullet, \bullet, s),E_m^\bullet}(x)\right] 
            - &\Pr\left[1\leftarrow A^{ \mathbf{\Hyb_1}(E_m^\bullet,\bullet, s),E_m^\bullet}(x)\right]
        \end{aligned}
        \right|\leq \mathsf{negl}(\lambda) \;,
    \end{equation*}
    where $\Hyb_1$ is as specified below.
\end{lemma}

\begin{algorithm}[h]
\caption{$\Hyb_1(E_m, \nset{E}{i}, s)$}
\setlength\multicolsep{0pt}
\begin{multicols}{2}
\begin{algorithmic}[1]
    \STATE $\mathsf{ch}\Usample \sch$
    \STATE set $k\in [n]$ s.t. $(E_k, s) \in R$.
    \STATE $\nset{\Delta}{i}, \nset{\Delta'}{i}, b \Usample G$
    \STATE $\tau \Usample \textit{sym}(n)$
    \STATE $\forall i\in[n] :  E_i^\alpha = \Delta_iE_i$, $E_i^\beta = \Delta'_iE_i^\alpha$
    \STATE \tikzmk{A}$\forall i\in[n] :  E_i^\gamma = bE_i^\beta$ \tikzmk{B}\boxit{green!40}
    \STATE \tikzmk{A} $r\Usample G$, $E^{\mathsf{Open}} = rE$
    \IF{$\mathsf{ch}=1,2,3$}
        \STATE $E^{\mathsf{Check}} = bE^{\mathsf{Open}}$ \tikzmk{B}\boxit{red!40}
    \ELSIF{$\mathsf{ch}=4$}
        \STATE $l=\Delta_k\Delta'_kbs$, $E^{\mathsf{Check}} = lE_m$
    \ENDIF
    \STATE set $\mathsf{resp}$ honestly w.r.t $\mathsf{ch}$
    \RETURN $(\mathsf{com},\mathsf{ch},\mathsf{resp})$
\end{algorithmic}
\end{multicols}
\end{algorithm}

\ifnum\submission=1 \vspace{-2em}\fi
\begin{proof}
  Each query input of $\Trans$ and $\Hyb_1$ is of form $(E_m, \nset{E}{i}, s)$ where $(\nset{E}{i},s)\in R_n$ and $E_m$ is the curve correspoinding to the random master public key. We first note that the difference between honest transcript $\Trans$
  and $\Hyb_1$ is that $\Hyb_1$ replaces honest $E^{\mathsf{Open}}$
  with $rE$ for a random $r\in G$. For $\mathsf{ch}\neq 4$,
  $E^{\mathsf{Check}}$ is also replaced accordingly to
  $E^{\mathsf{Open}}$.

  We will prove the indistinguishability of $(\mathsf{com},
  \mathsf{ch}, \mathsf{resp})\leftarrow\Trans$ and
  $(\mathsf{com}', \mathsf{ch}', \mathsf{resp}')\leftarrow\Hyb_1$ for
  each different challenge $\mathsf{ch}\in \Ccal$ separately. In the
  following proof, we set $k$ s.t. $(E_k,s)\in R$, as in both
  $\Trans$ and $\Hyb_1$

  For $\mathsf{ch}'=1$, we have $\mathsf{resp}' = \nset{\Delta}{i}$,
  which is honestly generated and thus identical to $\Trans$. We thus
  focus on the $\mathsf{com}'$ part.

  By the hardness of P-DDHAP, for random $\Delta'_k, r\Usample G$, we
  have
\begin{equation*}
    (E_m, \Delta'_kE, \Delta'_kE_m)\approx_c(E_m, \Delta'_kE, rE)
\end{equation*}
Hence, for random $\Delta_k, \Delta'_k, b, r\Usample G$ and honestly
generated $(E_m, E_k^\beta, E_k^\gamma, E^{\mathsf{Open}},\\
E^{\mathsf{Check}})$, we have
\begin{equation*}
    \begin{aligned}
        &(E_m, E_k^\beta, E_k^\gamma, E^{\mathsf{Open}}, E^{\mathsf{Check}})\\
        =&(E_m, \Delta_ks(\Delta'_kE), \Delta_kbs(\Delta'_kE), \Delta_ks(\Delta'_kE_m), \Delta_kbs(\Delta'_kE_m))\\
        \approx_c &(E_m, \Delta_ks(\Delta'_kE), \Delta_kbs(\Delta'_kE), \Delta_ks(rE), \Delta_kbs(rE))\\
        =&(E_m, E_k^\beta, E_k^\gamma, r'E, br'E)
    \end{aligned}
\end{equation*}
Where the left-hand side is the output $\mathsf{com}$ from $\Trans$, restricted to
the variables dependent on $s_m$ or $\Delta'_k$. The right-hand side is the
corresponding partial output from $\Hyb_1$. As the remaining parts of
$\Trans$ and $\Hyb_1$ are equivalent, this equation shows that the
output distributions of $\Trans$ and $\Hyb_1$ are indistinguishable for
$\mathsf{ch}=1$.

For the case $\mathsf{ch}=2,3$, the indistinguishability can be proved
in a similar fashion. Notice again that for random $\Delta_k,
r\Usample G$, $(E_m, \Delta_kE, \Delta_kE_m)\approx_c(E_m, \Delta_kE,
rE)$. Thus for random $\Delta_k, \Delta'_k, b, r\Usample G$
\begin{equation*}
    \begin{aligned}
        &(E_m, E_k^\alpha, E_k^\beta, E_k^\gamma, E^{\mathsf{Open}}, E^{\mathsf{Check}})\\
        =&(E_m, s(\Delta_kE), \Delta'_ks(\Delta_kE), \Delta'_kbs(\Delta_kE), \Delta'_ks(\Delta_kE_m), \Delta'_kbs(\Delta_kE_m))\\
        \approx_c &(E_m, s(\Delta_kE), \Delta'_ks(\Delta_kE), \Delta'_kbs(\Delta_kE), \Delta'_ks(rE), \Delta'_kbs(rE))\\
        =&(E_m, E_k^\alpha, E_k^\beta, E_k^\gamma, r'E, br'E)
    \end{aligned}
\end{equation*}

For the case $\mathsf{ch}=4$, we would need a slight change. First we
recall the fact that, since $\mathcal{GA}$ is free and transitive, for
every $E_i$ there exists a unique $s_i\in G$ s.t. $s_iE = E_i$. Thus,
sampling $\nset{D}{i}, b\Usample G$ and letting $\Delta'_i=(bs_i)\inv
D_i$ gives us a uniformly distributed $\nset{\Delta'}{i}$.

Now, again from P-DDHAP, for random $b, r\Usample G$, 
\begin{equation*}
    (E_m, b\inv E, b\inv E_m)\approx_c(E_m, b\inv E, rE)
\end{equation*}

Thus, for random $\nset{\Delta}{i}, \nset{D}{i}, b, r\Usample G$ where
$D_i = \Delta'_ibs_i$, we have
\begin{equation*}
    \begin{aligned}
        &(E_m, \nset{E^\alpha}{i}, \nset{E^\beta}{i}, \nset{E^\gamma}{i}, E^{\mathsf{Open}}, E^{\mathsf{Check}}, l)\\
        =&(E_m, \pset{\Delta_iE_i}{i}{[n]}, \pset{\Delta_i\Delta'_iE_i}{i}{[n]}, \pset{\Delta_i\Delta'_ibE_i}{i}{[n]}, \Delta_k\Delta'_ks_kE_m, \Delta_k\Delta'_kbs_kE_m, \Delta_k\Delta'_kbs_k)\\
        =&(E_m, \pset{\Delta_iE_i}{i}{[n]}, \pset{\Delta_iD_i(b\inv E)}{i}{[n]}, \pset{\Delta_iD_iE}{i}{[n]}, \Delta_kD_k(b\inv E_m), \Delta_kD_kE_m, \Delta_kD_k)\\
        \approx_c &(E_m, \pset{\Delta_iE_i}{i}{[n]}, \pset{\Delta_iD_i(b\inv E)}{i}{[n]}, \pset{\Delta_iD_iE}{i}{[n]}, \Delta_kD_k(rE), \Delta_kD_kE_m, \Delta_kD_k)\\
        =&(E_m, \nset{E^\alpha}{i}, \nset{E^\beta}{i}, \nset{E^\gamma}{i}, r'E, E^{\mathsf{Check}}, l)
    \end{aligned}
\end{equation*}

Finally, since both $\mathsf{ch}$ and $\mathsf{ch}'$ are sampled
randomly in $\sch$, we can conclude that $\Trans$ and $\Hyb_1$
are computationally indistinguishable.
\end{proof}

\begin{lemma}
\label{CWIlem2}

    For any $s\in G$, any efficient adversary $A$ with $E_m^\bullet(\nu)$ generating $E_m^\nu$ from $(E_m^\nu,s_m^\nu)\leftarrow R_E$ for each $\nu$, we have
     \begin{equation*}
        \left|
        \begin{aligned}
            \Pr\left[1\leftarrow A^{ \mathbf{\Hyb_1}(E_m^\bullet, \bullet, s),E_m^\bullet}(x)\right]
            - \Pr\left[1\leftarrow A^{\mathbf{\Hyb_2}(E_m^\bullet, \bullet, s),E_m^\bullet}(x)\right]
        \end{aligned}
        \right|\leq \mathsf{negl}(\lambda) \;,
    \end{equation*}
    where $\Hyb_2$ is as defined below.
\end{lemma}

\begin{algorithm}[H]
\caption{$\Hyb_2(E_m, \nset{E}{i}, s)$}
\setlength\multicolsep{0pt}
\begin{multicols}{2}
\begin{algorithmic}[1]
    \STATE $\mathsf{ch}\Usample \sch$
    \STATE set $k\in [n]$ s.t. $(E_k, s) \in R$.
    \STATE $\nset{\Delta}{i}, \nset{\Delta'}{i}, b \Usample G$
    \STATE $\tau \Usample \textit{sym}(n)$
    \STATE $\forall i\in[n] :  E_i^\alpha = \Delta_iE_i$, $E_i^\beta = \Delta'_iE_i^\alpha$
    \STATE $r\Usample G$, $E^{\mathsf{Open}} = rE$
    \IF{$\mathsf{ch}=1,2,3$}
        \STATE $E^{\mathsf{Check}} = bE^{\mathsf{Open}}$
        \STATE \tikzmk{A}$\forall i\in[n] :  E_i^\gamma = bE_i^\beta$ \tikzmk{B}\boxit{green!40}
    \ELSIF{$\mathsf{ch}=4$}
        \STATE \tikzmk{A}$\forall i\in[n] : r_i\Usample G, E_i^\gamma = r_iE$
        \STATE $l=r_k$, $E^{\mathsf{Check}} = lE_m$\quad\quad\quad\tikzmk{B}\boxit{green!40}
    \ENDIF
    \STATE set $\mathsf{resp}$ honestly w.r.t $\mathsf{ch}$
    \RETURN $(\mathsf{com},\mathsf{ch},\mathsf{resp})$
\end{algorithmic}
\end{multicols}
\end{algorithm}

\ifnum\submission=1\vspace{-3em}\fi
\begin{proof}
  The hybrids $\Hyb_1$ and $\Hyb_2$ differ only in the case
  $\mathsf{ch}=4$, in which we replace the whole $E^\gamma$ with
  random curves, $E^{\mathsf{Check}}$ and $l$ are also changed
  correspondingly. As in the previous proof, we use the fact
  that sampling $\nset{D}{i}, b\Usample G$ and letting
  $\Delta'_i=(bs_i)\inv D_i$ gives us uniformly random
  $(\nset{\Delta'}{i}, b)$.

By P-DDHAP, for random $b, \aset{D}{i}{\nmk}, \aset{r}{i}{\nmk}$, 
\begin{equation*}
    \begin{aligned}
        &(b\inv E, \pset{D_iE}{i}{\nmk}, \pset{D_ib\inv E}{i}{\nmk})\\
        \approx_c &(b\inv E, \pset{r_iE}{i}{\nmk}, \pset{D_ib\inv E}{i}{\nmk})
    \end{aligned} 
\end{equation*}

For simplicity, we let $S=\nmk$. Now, for random $\nset{\Delta}{i},
\nset{D}{i}, b, \aset{r}{i}{S}$ where $D_i = \Delta'_ibs_i$, and
$(E_m, \nset{E^\alpha}{i}, \aset{E^\beta}{i}{S},
\aset{E^\gamma}{i}{S}, E_k^\beta, E_k^\gamma, E^{\mathsf{Check}}, l)$
are the elements output from $\Hyb_1$, we have

\begin{equation*}
    \begin{aligned}
        &(E_m, \nset{E^\alpha}{i}, \aset{E^\beta}{i}{S}, \aset{E^\gamma}{i}{S}, E_k^\beta, E_k^\gamma, E^{\mathsf{Check}}, l)\\
        =&(E_m, \pset{\Delta_iE_i}{i}{[n]}, \pset{\Delta_i(D_ib\inv E)}{i}{S}, \pset{\Delta_i(D_iE)}{i}{S},\Delta_kD_k(b\inv E),\\ &\qquad \Delta_kD_kE, \Delta_kD_kE_m, \Delta_kD_k)\\
        \approx_c&(E_m, \pset{\Delta_iE_i}{i}{[n]}, \pset{\Delta_i(D_ib\inv E)}{i}{S}, \pset{\Delta_i(r_iE)}{i}{S}, \Delta_kD_k(b\inv E),\\ &\qquad  \Delta_kD_kE, \Delta_kD_kE_m, \Delta_kD_k)\\
        =&(E_m, \nset{E^\alpha}{i}, \aset{E^\beta}{i}{S}, \pset{r'_iE_i}{i}{S}, E_k^\beta, E_k^\gamma, E^{\mathsf{Check}}, l)\\
    \end{aligned}
\end{equation*}

Finally we let $r'_k = \Delta_kD_k$, which is obviously independent from 
all other $r'_i$, then $(E_k^\beta, E_k^\gamma, E^{\mathsf{Check}}, l)
= (r'_kb\inv E, r'_kE, r'_kE_m, r'_k)$. Note that $r'_kb\inv$ gives
fresh randomness since $b$ is now independent from all other elements in
the right-hand side. Thus the right-hand side perfectly fits the distribution for $\Hyb_2$. This
concludes that $\Hyb_1$ and $\Hyb_2$ are computationally
indistinguishable.

\end{proof}

\ifnum\submission=1\vspace{-3em}\fi
\begin{lemma}
\label{CWIlem3}
     For any $E_m\in X_m$, $\nset{E}{i}\in X^n$, and $s_{k_0},s_{k_1}$ s.t. both $(\nset{E}{i}, s_{k_0})$, $(\nset{E}{i}, s_{k_1})\in R_n$ then
     \begin{equation*}
         \Hyb_2(E_m, \nset{E}{i}, s_{k_0}) = \Hyb_2(E_m, \nset{E}{i}, s_{k_1})\;,
     \end{equation*}
    where ``$=$'' is understood as the output distribution being identical.
\end{lemma}

\begin{proof}
  We always have $\Hyb_2(E_m, \nset{E}{i}, s_{k_0}) = \Hyb_2(E_m,
  \nset{E}{i}, s_{k_1})$ for $\mathsf{ch}=1,2,3$, as every elements in
  the output is generated independently from $k$. For $\mathsf{ch}=4$,
  we can give a deeper look on elements in $(\mathsf{com},
  \mathsf{resp}) = (E^\alpha, E^\beta, E^\gamma, E^{\mathsf{Open}},
  E^{\mathsf{Check}}, l)$. The part $(E^\alpha, E^\beta,
  E^{\mathsf{Open}})$ is generated independent from $k$, and the part
  $(E^\gamma, E^{\mathsf{Check}}, l)$ is of the form
  $(\tau(\pset{r_iE}{i}{[n]}), r_kE_m, r_k)$. Since $\tau$ is a
  random permutation and $r_i$'s are independent randomness, the two
  distributions $(\tau(\pset{r_iE}{i}{[n]}), r_{k_0}E_m, r_{k_0})$
  and $(\tau(\pset{r_iE}{i}{[n]}), r_{k_1}E_m, r_{k_1})$ are
  obviously identical. Hence $\Hyb_2(E_m, \nset{E}{i}, s_{k_0}) =
  \Hyb_2(E_m, \nset{E}{i}, s_{k_1})$.

\end{proof}

Finally, by combining Lemma \ref{CWIlem1}, Lemma \ref{CWIlem2}, and
Lemma \ref{CWIlem3}, we conclude that for any efficient adversary $A$ with ${\sf E}_m^\bullet$ and $\Trans^*$ defined as usual, and any $s_{i},s_{j}\in G$, we have
 \begin{equation*}
    \left|
    \begin{aligned}
        \Pr\left[1\leftarrow A^{ \mathbf{Trans}^*(E_m^\bullet, \bullet, s_{i}),E_m^\bullet}(x)\right]
        - \Pr\left[1\leftarrow A^{ \mathbf{Trans}^*(E_m^\bullet, \bullet, s_{j}),E_m^\bullet}(x)\right]
    \end{aligned}
    \right|\leq \mathsf{negl}(\lambda)\;,
\end{equation*}
by restricting the query inputs $(E_m,\nset{E}{i},s_k)$ to those $(\nset{E}{i},s_i), (\nset{E}{i},s_j)\in R_n$ for whichever $k\in\{i,j\}$.
This concludes the proof of Lemma \ref{sigmaCWI}, and thus $\Sigma_{GA}$ is indeed an openable sigma protocol.

\ifnum\submission=1\vspace{-1em}\fi
\section{Constructing accountable ring signatures}
\label{sec:ARS}
\ifnum\submission=1\vspace{-1em}\fi

In this section, we will show how to obtain an accountable ring
signature scheme from our openable sigma protocol.
The construction
can be decomposed into two parts. We first take multiple parallel repetitions
to the protocol for soundness amplification; then, we apply the Fiat-Shamir transformation on the parallelized protocol to
obtain the full construction. One subtle issue is that since every
sigma protocol in the parallel repetition is generated independently,
each parallel session of the transcript may open to a different party.
Hence, we need an opening function for the parallelized protocol, which returns the majority output over the opening results of the
parallel sessions.

\ifnum\submission=1\vspace{-1em}\fi
\subsection{Construction}
\ifnum\submission=1\vspace{-1em}\fi

More generally, we are going to construct our ARS scheme $\mathcal{ARS}_{\Sigma}^t$ by performing Fiat-Shamir transformation to the parallel repeated protocol $\Sigma^{\otimes t}$ where the number of repetitions $t(\lambda,n)$ depends on the security parameter $\lambda$ and the number of members $n$. The construction of $\mathcal{ARS}_{\Sigma}^t$ is detailed as follows.

\begin{remark}
This can later be instantiated with $\mathcal{ARS}_{GA}:=\mathcal{ARS}_{\Sigma_{GA}}^t$ by choosing $\Sigma:=\Sigma_{GA}$ to be our previously constructed protocol over the group action and $t(\lambda,n):=2\lambda n$.
\end{remark}

\begin{itemize}
    \item $\MGen(1^\lambda)$:
    \begin{algorithmic}[1]
        \RETURN $(\mathsf{mpk}, \mathsf{msk}) \gets R_m(1^\lambda)$
    \end{algorithmic}
    
    \item $\Gen(1^\lambda)$:
    \begin{algorithmic}[1]
        \RETURN $(\pk, \sk) \gets  R(1^\lambda)$
    \end{algorithmic}
    
    \item $\Sign(\mathsf{mpk}, S, m, \sk)$
    \begin{algorithmic}[1]
        \STATE $t:=t(\lambda, |S|)$
        \STATE $\forall j \in [t], (\mathsf{com}_j, st_j)\leftarrow\Sigma_{GA}.\Com(\mathsf{mpk}, S, \sk)$
        \STATE $(\mathsf{ch}_1, \dots, \mathsf{ch}_t) \leftarrow H(m,\mathsf{com}_1,\dots, \mathsf{com}_t, S)$
        \STATE $\forall j \in [t], \mathsf{resp}_j\leftarrow \Sigma.\Resp(\mathsf{mpk}, S, \sk, \mathsf{com}_j, \mathsf{ch}_j, st_j)$
        \RETURN $\sigma = ({\mathsf{com}}, {\mathsf{resp}}) := ((\mathsf{com}_1,\dots, \mathsf{com}_t), (\mathsf{resp}_1, \dots, \mathsf{resp}_t)) $
    \end{algorithmic}
    
    \item $\Vrfy(\mathsf{mpk}, S, m, \sigma)$:
    \begin{algorithmic}[1]
        \STATE $t:=t(\lambda, |S|)$
        \STATE $\mathbf{parse\ } \sigma = ({\mathsf{com}}, {\mathsf{resp}})$
        \STATE ${\mathsf{ch}} := H(m, {\mathsf{com}}, S)$
        \STATE $\mathbf{check\ } \forall j\in[t] :  1\leftarrow \Sigma.\Vrfy(\mathsf{mpk}, \nset{\pk}{i}, \mathsf{com}_j, \mathsf{ch}_j, \mathsf{resp}_j)$
        \RETURN 1 $\mathbf{\ if\ all\ checks\ pass}$
    \end{algorithmic}
    
    \item $\Open(\mathsf{msk}, S, m, \sigma)$: 
    \begin{algorithmic}[1]
        \STATE $t:=t(\lambda, |S|)$
        \STATE $\mathbf{parse\ } \sigma = ({\mathsf{com}}, {\mathsf{ch}}, {\mathsf{resp}})$
        \STATE $\forall j \in [t], \mathsf{out}_j \leftarrow \Sigma.\Open(\mathsf{msk}, S, \mathsf{com}_i)$
        \STATE $\pk = \mathbf{Maj}(\aset{\mathsf{out}}{j}{[t]})$ \COMMENT{$\mathbf{Maj}$ outputs the majority element from its input list. In case of ties, it outputs a random choice of the majority elements.}
        \vspace{-1em}\RETURN $\pk$
    \end{algorithmic}
\end{itemize}

\begin{theorem}\label{ARSsecure}
    Let $\Sigma$ be a secure openable sigma protocol. Then ${\cal ARS}_{\Sigma}^t$ is secure for every $t(\lambda,n)= n\cdot {\sf poly}(\lambda)$. If $\Sigma$ is furthermore perfect-unique-response, then ${\cal ARS}_\Sigma^t$ is QROM-secure.
\end{theorem}
\begin{proof}
See Section~\ref{sec:classical_sec},~\ref{sec:qrom_sec} for the proof. This is concluded jointly from Lemma~\ref{lem:correct},~\ref{lem:anonymity_ROM},~\ref{lem:unforgeable_ROM},~\ref{lem:anonimity_qrom},~\ref{lem:unforgeable_QROM}.
\end{proof}

\ifnum\submission=1 \vspace{-1em}\fi
From Section~\ref{sec:openable_sigma_security} we know that ${\Sigma_{GA}}$ is a secure openable sigma protocol being $4$-special sound, and by applying the transformation from Section~\ref{sec:ARStoGS}, we
immediately get the following corollaries.

\begin{corollary}
    Assuming DDHAP is hard, then
    $\mathcal{ARS}_{GA}$ is a QROM-secure ARS scheme, and $\mathcal{GS}^{\mathcal{ARS}_{GA}}$ is a QROM-secure GS scheme
\end{corollary}
\ifnum\submission=1\vspace{-1em}\fi
This completes our construction of \textit{both} an accountable ring
signature scheme and a group signature scheme.

\ifnum\submission=0
\begin{remark}
One additional benefit of using class group
action as the key relation is that honest public keys can be
efficiently verified. As discussed in
Section~\ref{sec:classgroupaction}, any $E_i\in
\Ecal\ell\ell_p(\mathcal{O},\pi_p)$ is a valid public key since the
group action is transitive, and furthermore, any $E_i\notin
\Ecal\ell\ell_p(\mathcal{O},\pi_p)$ can be efficiently detected. This
prevents the possibility of a malformed master key or malformed public
keys, which is a potential attacking interface of an ARS scheme.
\end{remark} 
\fi

\subsection{Classical Security}\label{sec:classical_sec}
In this section, we provide the classical security proof for the ${\cal ARS}^t_\Sigma$ described earlier. We will outline the lemmas here, with the full formal proofs available in \cref{sec:classical_sec}.

\begin{lemma}\label{lem:correct}
    Let $\Sigma$ be a secure openable sigma protocol, then $\mathcal{ARS}_{\Sigma}^t$ is \textbf{correct}.
\end{lemma}

\begin{lemma}\label{lem:anonymity_ROM}
  Let $\Sigma$ be a secure openable sigma protocol, then $\mathcal{ARS}_{\Sigma}^t$ is \textbf{anonymous} for every $t(\lambda,n)\leq{\sf poly}(\lambda,n)$ in CROM.
\end{lemma}

\begin{lemma}\label{lem:unforgeable_ROM}
    Let $\Sigma$ be a secure openable sigma protocol. Then $\mathcal{ARS}_{\Sigma}^t$ is \textbf{unforgeable} for every $t(\lambda,n)=n\cdot{\sf poly}(\lambda)$ in the CROM.
\end{lemma}

\if\submission=1\subsection{Classical Security}\label{sec:classical_sec}
\ifnum\submission=1\vspace{-1em}\fi
In this section, we are going to provide classical security proof of the ${\cal ARS}^t_\Sigma$ described earlier. Starting from CROM, and then lifting those to QROM.

For the proof of Theorem \ref{ARSsecure} we again break down the theorem
into proving each security property, i.e. correctness, anonymity and unforgeability. For correctness, there is no difference between classical and quantum settings, but since the proof does not exploit ``quantum-ness'' of an adversary, we put it in this section as well.

\ifnum\submission=1\vspace{-0.5em}\fi
\begin{lemma}\label{lem:correct}
    Let $\Sigma$ be a secure openable sigma protocol, then $\mathcal{ARS}_{\Sigma}^t$ is \textbf{correct}.
\end{lemma}
\ifnum\submission=1\vspace{-1em}\fi
\begin{proof}
  For any master key pair $(\mathsf{mpk},
  \mathsf{msk})\in\Kcal\Pcal_m$, any key pair
  $(\pk,\sk)\in\Kcal\Pcal$, and any set of public keys $S$ such that
  $\pk\in S$, we directly have $(\mathsf{mpk}, \mathsf{msk})\in R_m$
  and $(S, \sk)\in R_n$ where $n=|S|$. Let
  $\sigma\leftarrow\Sign(\mathsf{mpk}, S, m, \sk)$ be an honest
  signature on message $m$ and ring $S$. Notice that in an honest
  execution of $\Sign$, each $\mathsf{com}_j$ and $\mathsf{resp}_j$ is
  honestly generated according to $\Sigma$. Thus by the
  correctness of $\Sigma$, we know for ${\sf ch}:=H({\sf com},m)$ and every $j\in[t]$ with
  probability $1-\mathsf{negl}(\lambda)$, that $1\leftarrow
  \Sigma.\Ver(\mathsf{mpk}, S, \mathsf{com}_j, \mathsf{ch}_j,
  \mathsf{resp}_j)$ and $\pk\leftarrow\Sigma.\Open(\mathsf{mpk},
  S, \mathsf{com}_j)$. Hence we directly obtain that, with probability
  $1-t\cdot \mathsf{negl}(\lambda) = 1-\mathsf{negl}(\lambda)$, we have that
  $1\leftarrow \Ver(\mathsf{mpk}, S, m, \sigma)$ and $\pk\leftarrow
  \Open(\mathsf{msk}, S, m, \sigma)$. This concludes the proof that
  $\mathcal{ARS}_{\Sigma}$ is correct.
\end{proof}

\ifnum\submission=1\vspace{-1em}\fi
\begin{lemma}\label{lem:anonymity_ROM}
  Let $\Sigma$ be a secure openable sigma protocol, then $\mathcal{ARS}_{\Sigma}^t$ is \textbf{anonymous} for every $t(\lambda,n)\leq{\sf poly}(\lambda,n)$ in CROM.
\end{lemma}
\ifnum\submission=1\vspace{-1.5 em}\fi
\begin{proof}
The anonymity of $\mathcal{ARS}_{\Sigma}$ follows immediately from the CWI
property of $\Sigma$. For any efficient adversary $A$ with at most $q$
queries to the random oracle, it can have at most
$q\cdot{\sf negl}(\lambda)\leq \mathsf{negl}(\lambda)$ advantage on distinguishing $\Sign^*$
and $(\Trans^*)^t$. And by CWI from $\Sigma$, we have
$\Trans^*(\mathsf{mpk}, S, \sk_{id_0})\approx_c \Trans^*(\mathsf{mpk},
S, \sk_{id_1})$. Hence we can directly conclude that
$\Sign^*(\mathsf{mpk}, S, \sk_{id_0})\approx_c \Sign^*(\mathsf{mpk},
S, \sk_{id_1})$, which proves that $\mathcal{ARS}_\Sigma^t$ is anonymous.
\end{proof}

\begin{lemma}\label{lem:unforgeable_ROM}
    Let $\Sigma$ be a secure openable sigma protocol. Then $\mathcal{ARS}_{\Sigma}^t$ is \textbf{unforgeable} for every $t(\lambda,n)=n\cdot{\sf poly}(\lambda)$ in the CROM.
\end{lemma}
We refer readers to \ref{sec:lem10} for the proof.
\fi
\ifnum\submission=1\vspace{-1em}\fi
\subsection{QROM security}\label{sec:qrom_sec}
\ifnum\submission=1\vspace{-0.5em}\fi

To start, we show the anonymity first, where an adversary is asked to distinguish the signing oracles $\Sign^*({\sf mpk}_\bullet,\bullet,\bullet,\sk_k)$ for $k\in\{i,j\}$. Recall that $\Sign^*$ is defined with respect to two fixed public-secret key pairs $(\pk_k,\sk_k)\in\Kcal\Pcal$ for $k\in\{i,j\}$, a query $\Sign^*({\sf mpk}_\nu,S,m,\sk_k)$ must be such that $\{\pk_i,\pk_j\}\subseteq S$. 

The idea is that $\Sign^*$ behaves {\em almost} as if running $t$ repetitions of the openable sigma protocol $\Trans^*({\sf mpk}_\bullet,\bullet,\sk_k)$, which one cannot distinguish between $k\in\{i,j\}$. There is one exception: namely, $\Sign^*$ computes the challenges by hash evaluation ${\sf ch}:=H(m,{\sf com}, S)$, but then since $m$ and $S$ are chosen by the adversary, this may cause bias to the challenge distribution.

Such bias is handled by reprogramming techniques. Note that the first message ${\sf com}$ is freshly sampled with high min-entropy in each query to $\Sign^*$. Therefore, it is unlikely that $H(m,{\sf com}, S)$ has been queried, and thus ${\sf ch}$ is almost unbiased. In the quantum setting, one cannot simply identify previous queries to $H$, but the adaptive reprogramming technique \cite{grilo2021tight} can still be used to mimic this line of reasoning.

For convenience, we will use the prefix ``$\Sigma^{\otimes t}.$'' to specify that the scope of the object lies in the $t$-time repetitions of $\Sigma$, with $\Sigma^{\otimes t}.\Vrfy$ outputting $1$ if all repetitions are accepted, and $\Sigma^{\otimes t}.\Open$ outputting the majority of the opening results.

\ifnum\submission=1\vspace{-1em}\fi
\begin{lemma}\label{lem:anonimity_qrom}
    Let $\Sigma$ be an openable sigma protocol that is high min-entropy. Then ${\cal ARS}_{\Sigma}^t$ is {\bf anonymous} for every $t(\lambda,n)\leq {\sf poly}(n,\lambda)$ in QROM.
\end{lemma}
\ifnum\submission=1\vspace{-1em}\fi
\begin{proof}
For the purpose of analysis, define the following $\Sign^*_2({\sf mpk}_\bullet,\bullet,\bullet,\sk_k)$ oracle for $k\in\{i,j\}$.
\begin{itemize}
    \item $\Sign^*_2({\sf mpk}_\nu,S,m,\sk_k)$:
    \begin{algorithmic}[1]
        \STATE {\bf abort} if $\{\pk_i,\pk_j\}\not\subseteq S$
        \STATE $({\sf com},{\sf st})\gets \Sigma^{\otimes t}.\Com({\sf mpk}, S,\sk_k)$
        \STATE {\bf program} $H(m,{\sf com}, S):={\sf ch}\gets \Sigma^{\otimes t}.\Ccal$
        \STATE ${\sf resp}\gets \Sigma^{\otimes t}.\Resp({\sf mpk},S,\sk_k,{\sf com},{\sf ch},{\sf st})$
        \STATE {\bf return} $({\sf com},{\sf resp})$
    \end{algorithmic}
\end{itemize}
Note that $\Sign_2^*$ only replaces the computation of challenge in $\Sign^*$ from ${\sf ch}:=H(m,{\sf com},S)$ using the random oracle $H$, to freshly sampling ${\sf ch}$ and reprogramming to $H(m,{\sf com},S):={\sf ch}$. As described earlier, since ${\sf com}$ is high-min-entropy, by \cite[Theorem~1]{grilo2021tight} we obtain $A^{\Sign^*({\sf mpk}_\bullet,\bullet,\bullet,\sk_k),{\sf mpk}_\bullet,H}(\pk_k)\approx A^{\Sign_2^*({\sf mpk}_\bullet,\bullet,\bullet,\sk_k),{\sf mpk}_\bullet,H}(\pk_k)$ being indistinguisable. Now, via Zhandry's comressed oracle technique, or alternatively as described in \cite[Appendix~A]{CFHL21}, there is an efficient quantum algorithm $B^{\Trans^*({\sf mpk}_\bullet,\bullet,\bullet,\sk_k),{\sf mpk}_\bullet}(\pk_k)$ that run as if $A^{\Sign_2^*({\sf mpk}_\bullet,\bullet,\bullet,\sk_k),{\sf mpk}_\bullet,H}(\pk_k)$ but emulating the random oracle and reprogramming by itself. Since $\Sigma$ is computational witness-indistinguishable, $B$ cannot distinguish between $k\in\{i,j\}$. Putting things together, we obtain the following chain of indistinguishability,
\begin{align*}
    &A^{\Sign^*({\sf mpk}_\bullet,\bullet,\bullet,\sk_i),{\sf mpk}_\bullet,H}(\pk_i)
    \approx A^{\Sign_2^*({\sf mpk}_\bullet,\bullet,\bullet,\sk_i),{\sf mpk}_\bullet,H}(\pk_i)\\
    &\approx B^{\Trans^*({\sf mpk}_\bullet,\bullet,\bullet,\sk_i),{\sf mpk}_\bullet}(\pk_i)
    \approx B^{\Trans^*({\sf mpk}_\bullet,\bullet,\bullet,\sk_j),{\sf mpk}_\bullet}(\pk_j)\\
    &\approx A^{\Sign_2^*({\sf mpk}_\bullet,\bullet,\bullet,\sk_j),{\sf mpk}_\bullet,H}(\pk_j)
    \approx A^{\Sign^*({\sf mpk}_\bullet,\bullet,\bullet,\sk_j),{\sf mpk}_\bullet,H}(\pk_j)\;.
\end{align*}
This concludes the proof.
\end{proof}

\ifnum\submission=1\vspace{-1.5em} \fi
For the rest of this section, we show unforgeability in QROM. The key to lifting Lemma~\ref{lem:unforgeable_ROM} into QROM is a quantum extraction technique. The classical forking lemma, which measures out part of the transcript before rewinding, may ruin the internal quantum state of the adversary and therefore does not trivially apply to the quantum setting.

First, we give a CMA-to-NMA reduction, i.e. transforming an adversary $A$ against $G^{\sf UF}_A$ into an adversary against $\widetilde G^{\sf UF}_A$ as defined below.

\ifnum\submission=1 \vspace{-2em}\fi
\begin{algorithm}[H]
\caption{$\widetilde G_{A}^{\mathsf{UF}}({\sf mpk},{\sf msk})$: NMA-Unforgeability game}
\begin{algorithmic}[1]
    \STATE $(\pk,\sk)\gets\Gen(1^\lambda)$
    \STATE $(S^*, m^*, \sigma^*)\leftarrow A^H(\pk)$
    \STATE {\bf check} $1\leftarrow \Vrfy(\mathsf{mpk}, S^*, m^*, \sigma^*)$
    \STATE {\bf check} $\pk\text{ or }\bot\leftarrow\Open(\mathsf{msk}, S^*, m^*, \sigma^*)$
    \STATE $A$ wins if all check pass
\end{algorithmic}
\end{algorithm}
\ifnum\submission=1 \vspace{-2.3em}\fi
This is by means of simulating the signing queries $\Sign$ via a simulator $\Sim$ as follows.
\ifnum\submission=1 \vspace{-1em}\fi
\begin{center}
    \begin{minipage}[t]{0.4\linewidth}
        $\Sign({\sf mpk},S,m,\sk)$:
        \vspace{-1em}\begin{algorithmic}[1]
            \STATE $t:=t(\lambda,|S|)$
            \STATE ${\sf com}\gets \Sigma^{\otimes t}.\Com$
            \STATE ${\sf ch}:=H(m,{\sf com},S)$
            \STATE ${\sf resp}\gets \Sigma^{\otimes t}.\Resp$
            \STATE {\bf return} $({\sf com},{\sf resp})$
        \end{algorithmic}
    \end{minipage}
    \begin{minipage}[t]{0.55\linewidth}
        $\Sim({\sf mpk},S,m)$:
        \begin{algorithmic}[1]
            \STATE $t:=t(\lambda,|S|)$
            \STATE $({\sf com},{\sf ch},{\sf resp})\gets \Sigma^{\otimes t}.\Sim({\sf mpk},S,\pk)$
            \STATE {\bf program} $H(a,m):={\sf ch}$
            \STATE {\bf return} $({\sf com},{\sf resp})$
        \end{algorithmic}
    \end{minipage}
\end{center}
\begin{lemma}\label{lem:sign_sim_ind}
    Let $\Sigma$ be a statistical HVZK, high min-entropy openable sigma protocol, the number of repetitions be $t(\lambda,n)\leq{\sf poly}(\lambda,n)$ and $(\pk,\sk)\gets\Gen(1^\lambda)$ be freshly sampled. Then for every efficient quantum algorithm $A$, we have
    $$
    \left|\Pr\left[
        1\gets A^{\Sign(\bullet,\bullet,\bullet,\sk),H}(\pk)
    \right] - \Pr\left[
        1\gets A^{\Sim,H}(\pk)
    \right]\right|\leq {\sf negl}(\lambda)\;.
    $$
\end{lemma}
\begin{proof}
Define an intermediate oracle $\Sign_2$ as follows.
\ifnum\submission=1 \vspace{-1em}\fi
\begin{itemize}
    \item $\Sign_2({\sf mpk}, S, m)$:
    \begin{algorithmic}[1]
            \STATE $t:=t(\lambda,|S|)$
            \STATE ${\sf com}\gets \Sigma^{\otimes t}.\Com$
            \STATE {\bf program} $H(m,{\sf com},S):= {\sf ch}\stackrel{\$}{\gets}\Ccal$
            \STATE ${\sf resp}\gets \Sigma^{\otimes t}.\Resp$
            \STATE {\bf return} $({\sf com},{\sf resp})$
        \end{algorithmic}
\end{itemize}
\ifnum\submission=1 \vspace{-1em}\fi
Note that $\Sign$ and $\Sign_2$ only differs at one place, where the former computes the challenge ${\sf ch}:=H(m,{\sf com},S)$ using the random oracle $H$, but the latter samples a fresh challenge ${\sf ch}$ and then reprogrammed the corresponding entry $H(m,{\sf com},S):={\sf ch}$. Since ${\sf com}$ is of high-min-entropy, a direct application of \cite[Theorem 1]{grilo2021tight} implies $A^{\Sign(\bullet,\bullet,\bullet,\sk),H}(\pk)\approx A^{\Sign_2,H}(\pk)$ being indistinguishable. Furthermore, $\Sign_2$ and $\Sim$ only differ in how the transcript is respectively generated, with the former produced via an honest execution $\Sigma^{\otimes t}.\Trans$, and the latter via the corresponding simulator $\Sigma^{\otimes t}.\Sim$. It follows directly from the HVZK property that $A^{\Sign_2,H}(\pk)\approx A^{\Sim,H}(\pk)$ is indistinguishable. This concludes the proof.
\end{proof}

\ifnum\submission=1 \vspace{-1em}\fi
Now we are ready to prove the CMA-to-NMA reduction.

\ifnum\submission=1 \vspace{-1em}\fi
\begin{lemma}\label{lem:CMA2NMA_QROM}
    Let $\Sigma$ be a statistical HVZK, high min-entropy, computationally unique-response openable sigma protocol and the number of repetitions $t(\lambda,n)\leq{\sf poly}(\lambda,n)$. For every valid master key pair $({\sf mpk},{\sf msk})\in\Kcal\Pcal_m$ efficient (CMA) quantum adversary $A$ against $G^{\sf UF}_A({\sf mpk},{\sf msk})$, there is an efficient (NMA) quantum adversary $B$ against $\widetilde G^{\sf UF}_B({\sf mpk},{\sf msk})$ such that
    $$
    \left|\Pr\left[A\text{ wins } G^{\sf UF}_A({\sf mpk},{\sf msk})\right]
    - \Pr\left[B\text{ wins } \widetilde G^{\sf UF}_B({\sf mpk},{\sf msk})\right]
    \right|\leq {\sf negl}(\lambda)\;.
    $$
\end{lemma}
\ifnum\submission=1 \vspace{-1em}\fi
\begin{proof}
Let $B^H(\pk)$ run $(S^*,m^*,\sigma^*)\gets A^{\Sim,H}(\pk)$ but emulating the reprogramming of $H$ by itself. Already from Lemma~\ref{lem:sign_sim_ind} we may conclude the following
$$
\left|\Pr\left[
    \substack{
        (S^*,m^*,\sigma^*)\gets A^{\Sign(\bullet,\bullet,\bullet,\sk),H}(\pk)\\
        1\gets \Vrfy({\sf mpk,S^*,m^*,\sigma^*})\\
        \pk\text{ or }\bot\gets\Open({\sf msk},S^*,m^*,\sigma^*)}
\right] - \Pr\left[
    \substack{
        (S^*,m^*,\sigma^*)\gets A^{\Sim,H}(\pk)\\
        1\gets \Vrfy^H({\sf mpk,S^*,m^*,\sigma^*})\\
        \pk\text{ or }\bot\gets\Open({\sf msk},S^*,m^*,\sigma^*)}
\right]\right|\leq {\sf negl}(\lambda)\;,
$$
where $\Vrfy^H$ is understood as the verification with respect to the possibly reprogrammed random oracle $H$.

Without loss of generality we may assume $A^{\Sim,H}$ never outputs $\sigma^*$ produced by querying $\Sim({\sf mpk},S^*,m^*)$ for the message $m^*$. If the produced $(S^*,m^*,\sigma^*)\gets A^{\Sim,H}(\pk)$ satisfies $1\gets \Vrfy^H({\sf mpk,S^*,m^*,\sigma^*})$ and $\pk\gets\Open({\sf msk},S^*,m^*,\sigma^*)$. It may be (1) there has been a query of form $({\sf com}^*,{\sf resp})\gets \Sim({\sf mpk},S^*,m^*,\pk)$ for some ${\sf resp}$ so that there has been the reprogramming of form $H(m^*,{\sf com}^*,S^*):={\sf ch}^*$, in which case ${\sf resp}\neq {\sf resp}^*$ so $({\sf com}^*,{\sf ch}^*,{\sf resp}^*)$ and $({\sf com}^*,{\sf ch}^*,{\sf resp})$ are distinct valid transcripts of $\Sigma^{\otimes t}$, which is hard to find due to the computational unique-response property, or (2) there has not been such a query, in which case $H(m^*,{\sf com}^*,S^*)$ would not have been reprogrammed (except with negligible probability), and so the verification $\Vrfy({\sf msk},S^*,m^*,\sigma^*)$ with respect to the un-reprogrammed $H$ will pass. This concludes the proof.

\end{proof}
\ifnum\submission=1 \vspace{-1em}\fi
Next, for every valid master key pair $({\sf mpk},{\sf msk})\in\Kcal\Pcal_m$, define the interactive unforgeability game $G^{\sf int}_A({\sf mpk},{\sf msk})$ as follows.
\ifnum\submission=1 \vspace{-2em}\fi
\begin{algorithm}[H]
\caption{$G_{A}^{\mathsf{int}}({\sf mpk},{\sf msk})$: Interactive unforgeability game}
\begin{algorithmic}[1]
    \STATE $(\pk,\sk)\gets\Gen(1^\lambda)$
    \STATE $(S^*, m^*, {\sf com}^*,{\sf st})\leftarrow A(\pk)$ and $t:=t(\lambda,|S|)$
    \STATE ${\sf ch}\stackrel{\$}{\gets}\Sigma^{\otimes t}.\Ccal$
    \STATE ${\sf resp}^*\gets A({\sf st},{\sf ch})$
    \STATE $A$ wins if the following holds: ${1\leftarrow
  \Sigma^{\otimes t}.\Vrfy(\mathsf{mpk}, S^*, {\sf com}^*,{\sf ch}^*,{\sf resp}^*)}$ and ${\pk\text{ or }\bot\leftarrow\Sigma^{\otimes t}.\Open(\mathsf{msk}, S^*, {\sf com}^*)}$
\end{algorithmic}
\end{algorithm}
\ifnum\submission=1 \vspace{-2em}\fi
We are going to reduce an NMA adversary to the another interactive adversary against the openable sigma protocol, with freedom to choose which set $S$ of instances to break on its choice, so long as the secret key $\sk$ is included in $S$. 
\begin{lemma}\label{lem:NMA2INT_QROM}
Let $\Sigma$ be an openable sigma protocol. For every $({\sf mpk},{\sf msk})\in\Kcal\Pcal_m$ and every efficient (NMA) quantum adversary $A$ against $\widetilde G^{\sf UF}_A({\sf mpk},{\sf msk})$ making at most $q$ queries to the random oracle $H$, there is an efficient (interactive) quantum adversary $B$ against $\widetilde G^{\sf int}_B$ such that the following holds
$$
\frac{\Pr\left[
    A \text{ wins }\widetilde G_{A}^{\mathsf{UF}}({\sf mpk},{\sf msk})
\right]}{(2q+1)^2} \leq \Pr\left[
    B \text{ wins }G_{A}^{\mathsf{int}}({\sf mpk},{\sf msk})
\right]\;.
$$
\end{lemma}
\ifnum\submission=1 \vspace{-1em}\fi
\begin{proof}
This is via direct application of the measure-and-reprogram technique. For every fixed choice of $\pk^\circ\in\Kcal\Pcal$, let $V_{\pk^{\circ}}$ be the predicate as described below.
\ifnum\submission=1 \vspace{-2em}\fi
\begin{itemize}
    \item $V_{\pk^{\circ}}(x=(m,{\sf com},S),{\sf ch},{\sf resp})$:
    \begin{algorithmic}[1]
        \STATE $t:=t(\lambda,|S|)$
        \STATE {\bf check} $1\leftarrow
  \Sigma^{\otimes t}.\Vrfy(\mathsf{mpk}, S, {\sf com},{\sf ch},{\sf resp})$
        \STATE {\bf check} $\pk^\circ\text{ or }\bot\leftarrow\Sigma^{\otimes t}.\Open(\mathsf{msk}, S, {\sf com})$
        \STATE {\bf return} $1$ iff all check pass
    \end{algorithmic}
\end{itemize}
\ifnum\submission=1 \vspace{-1em}\fi
By construction, for $(\pk,\sk)\gets\Gen(1^\lambda)$ we have
\begin{align*}
    &\Pr\left[A \text{ wins }\widetilde G^{\sf UF}_A\right] = \sum_{(\pk^\circ,\sk^\circ)\in\Kcal\Pcal} \Pr\left[\substack{\pk=\pk^\circ\\\sk=\sk^\circ}\right]\cdot \Pr_H\left[\substack{(S^*,m^*,({\sf com}^*,{\sf resp}^*))\gets A^H(\pk^\circ)
    \\1\gets V_{\pk^{\circ}}(x,H(x),{\sf resp}^*)}\right]\\
    &\Pr\left[B \text{ wins }\widetilde G^{\sf int}_A\right] = \sum_{(\pk^\circ,\sk^\circ)\in\Kcal\Pcal}
    \Pr\left[\substack{\pk=\pk^\circ\\\sk=\sk^\circ}\right]\cdot
    \Pr_{{\sf ch}\gets\Sigma^{\otimes t}.\Ccal}\left[\substack{(S^*,m^*,{\sf com}^*,{\sf st})\gets B(\pk^\circ)\\ {\sf resp}^*\gets B({\sf st},{\sf ch})\\ 1\gets V_{\pk^{\circ}}(x,H(x),{\sf resp}^*)}\right]\;,
\end{align*}
for every interactive algorithm $B$, where $x$ denotes $(m^*,{\sf com}^*,S^*)$. Summing over $x_\circ$ in \cite[Theorem~2]{DFM20}, we obtain the existence of an efficient $B$ such that the following holds for all $\pk^\circ,\sk^\circ$
$$
\Pr_{{\sf ch}\gets\Sigma^{\otimes t}.\Ccal}\left[\substack{(S^*,m^*,{\sf com}^*,{\sf st})\gets B(\pk^\circ)\\ {\sf resp}^*\gets B({\sf st},{\sf ch})\\ 1\gets V_{\pk^{\circ}}(x,{\sf ch}^*,{\sf resp}^*)}\right]
\geq
\Pr_H\left[\substack{(S^*,m^*,({\sf com}^*,{\sf resp}^*))\gets A^H(\pk^\circ)
    \\1\gets V_{\pk^{\circ}}(x,H(x),{\sf resp}^*)}\middle]\right/{(2q+1)^2}\;\;.
$$
Finally, summing over all choice of $(\pk^\circ,\sk^\circ)$ with suitable probability, the proof is concluded.
\end{proof}
\ifnum\submission=1 \vspace{-1em}\fi
Finally, we reduce an interactive adversary against $G^{\sf int}_A$ into another adversary that extract the secret key $\sk$ from the public key $\pk$. Note that the key generation samples a key pair $(\pk,\sk)\gets R(1^\lambda)$ with respect to a hard relation $R$, and thus the secret key $\sk$ should be hard to extract.

The idea of extraction goes as follows. Let $t(\lambda, n)=(n+1)\kappa$ be the number of repetitions, where $\kappa$ is to be decided later. If $A$ wins $G^{\sf int}_A$, i.e. producing a valid transcript that is opened to $\pk$ or $\bot$, then by the pigeonhole principle, there must be at least $\kappa$ repetitions opened to $\pk$ or at least $\kappa$ opened to $\bot$. We then perform rewinding in order to collect sufficient number of accepted responses for these repetitions. Once there are $\mu$ accepted responses in the same repetition being produced with non-zero probability, (\ref{eq:openable_special_sound}) immediately falsify them being opened to $\bot$, and so we can always extract a secret key $\sk$ using the extractor $\Sigma.\Ext$ provided by the $\mu$-special-sound property.

Note that it is not just a black-box evocation of (generalized) Unruh's rewinding because it only provides guarantee toward the number of collected valid transcripts, but not toward the content of those transcripts. When analyzing a parallel-repetition multi-special-sound protocol, one needs to open up the rewinding argument and see what's inside. On a very high-level, thanks to the fact that the opening result is determined once the first message {\sf com} is produced, one can still argue that conditioned on any fixed choice of the opening result, the collected transcripts are with challenges being uniformly random. The analysis is more involved, and we refer interested readers to Supplementary~\ref{sec:INT2EXT_QROM_proof}.

\begin{lemma}\label{lem:INT2EXT_QROM}
Let $\Sigma$ be a $\mu$-special-sound openable sigma protocol, the number of repetitions be $t(\lambda,n)=(n+1)\cdot\kappa(\lambda,n)$. For every $({\sf mpk},{\sf msk})\in\Kcal\Pcal_m$ and every efficient quantum adversary $A$ against $G^{\sf int}_A$, there exists an efficient quantum adversary $B$ such that
$$
\Pr\left[A\text{ wins }G^{\sf int}_A({\sf mpk},{\sf msk})\right]^{2\mu-1}\leq
\Pr\left[\substack{(\pk,\sk)\gets R(1^\lambda)\\ \sk\gets B(\pk)}\right] + \exp\left(\frac{-\kappa}{\mu^\mu}\right)\;.
$$
\end{lemma}
Putting everything together, we conclude unforgeability in QROM. For completion, see Supplementary~\ref{sec:unforgeable_QROM_proof} for a rather formal wrapping up.
\begin{lemma}\label{lem:unforgeable_QROM}
  Let $\mu$ be a constant and $\Sigma$ be an openable sigma protocol being correct, $\mu$-special-sound, statistical HVZK, perfect-unique-response, and high min-entropy. Then ${\cal ARS}_\Sigma^t$ is {\bf unforgeable} in QROM for every $t(\lambda,n)= n\cdot {\sf poly}(\lambda)$.
\end{lemma}

\ifnum\submission=1\vspace{-1.5em}\fi


\ifnum\submission=0
\subsection*{Acknowledgments}
Authors were supported by Taiwan Ministry of Science and Technology
Grant 109-2221-E-001-009-MY3, Sinica Investigator Award
(AS-IA-109-M01), Executive Yuan Data Safety and Talent Cultivation
Project (AS-KPQ-109-DSTCP), and Young Scholar Fellowship (Einstein
Program) of the Ministry of Science and Technology (MOST) in Taiwan,
under grant number MOST 110-2636-E-002-012, 
and by the Netherlands Organisation for 
Scientific Research (NWO) under grants 628.001.028 (FASOR)
and 613.009.144 (Quantum Cryptanalysis of Post-Quantum Cryptography),
and by the NWO funded project HAPKIDO (Hybrid Approach for quantum-safe Public Key Infrastructure Development for Organisations). Mi-Ying (Miryam) Huang is additionally supported by the NSF CAREER award 2141536, the United State. This work was carried out while the fifth author was
visiting Academia Sinica, she is grateful for the hospitality.
\fi

\if\submission=0
    \bibliography{ref-condensed}
    \bibliographystyle{splncs04}
\else
    \bibliography{ref-condensed}
    \bibliographystyle{splncs04}
\fi

\newpage
\appendix

\section*{Supplementary Material}

\section{Open problems}\label{sec:open_problem}

\noindent{\bf Extra judging functionality.} Our setting has premised an honest manager, only to whom the opening result is available. A corrupted manager can thus incriminate any party as the signer of an arbitrary signature. Many previous works on group signatures then provide an extra {\em judging function} allowing the manager to generate a
    publicly verifiable proof for its opening results. We offer (in Supplementary~\ref{sec:judge}) a weaker version with a simple tweak that will prevent a dishonest manager from incriminating honest non-signers. It remains open to constructing a QROM-secure ARS with a full-fledged judging function.

\section{Fiat Shamir with Aborts Flaw in Previous Literature}
\label{Tech:FSwA}
The flaw lies in arguing that an honestly generated signature does not leak its secret key, which is formally captured via the existence of a simulator $\Sim$ that is indistinguishable from the signing procedure $\Sign$ in the random oracle model, i.e.
$$
\left|\Pr\left[1\gets\Acal^{H,\Sign^H}\right] - \Pr\left[1\gets\Acal^{H,\Sim^H}\right]\right|
$$
is small, for an efficient oracle algorithm $\Acal$.
Here, $\Sim$ is allowed to reprogram certain entries $H(x):=y$ of a random oracle $H$, while $\Sign$ is not. For a detailed description of the flaw in general, we refer to \cite{FSwA23,FSwA23_Devevey}.

\paragraph{\bf Why \cite{beullens2020calamari, lai2021collusion, BDKLP22} cannot be easily fixed?}
First, we note that fixes of \cite{FSwA23_Devevey} premise a stronger honest-verifier zero-knowledge property of the underlying Sigma protocol, which is not available here. So we only discuss the patch as in \cite{FSwA23}. Essentially, in all the previous works, the underlying Sigma protocol is described as $\Sigma=(\Com,\Ccal,\Resp)$ where $\Ccal$ is a finite set, and $\Com^H,\Resp^H$ are oracle algorithms given access to $H$. One execution of $\Sigma$ generates $(a,{\sf st})\gets\Com^H$, $c\gets\Ccal$ and $z:=\Resp^H(a,c,{\sf st})$, while to generate a FSwA signature, the challenge is replaced with $c:=H(a,m)$ on input a message $m$, and then $(a,z)$ is returned as a signature. On the other hand, a simulated signature, is generated by first produce a non-abort transcript $(a,c,z)$, and reprogram it to the random oracle $H(a,m):=c$. The patch in \cite{FSwA23} proceeds through the following hybrid sequence
$$
\Acal^{H,\Sign}\approx\Acal^{H,{\bf Prog}}\approx \Acal^{H,\Trans}\approx\Acal^{H,\Sim}\;,
$$
where the precise definition of ${\bf Prog}$ and ${\bf Trans}$ are described below.
\begin{figure}[H]
    \centering
    \begin{minipage}[t]{0.49\linewidth}
        ${\bf Prog}^H(m)$:
        \begin{algorithmic}[1]
            \REPEAT
                \STATE ${\sf com},{\sf st}\gets\Com^H$
                \STATE $H(m,{\sf com}):={\sf ch}\gets\Ccal$
                \STATE ${\sf resp}\gets\Resp^H({\sf com},{\sf ch})$
            \UNTIL{$z\neq\bot$}
            \RETURN $({\sf com},{\sf resp})$
        \end{algorithmic}
    \end{minipage}
    \begin{minipage}[t]{0.5\linewidth}
        ${\bf Trans}^H(m)$:
        \begin{algorithmic}[1]
            \REPEAT
                \STATE ${\sf com},{\sf st}\gets\Com^H$
                \STATE ${\sf ch}\gets\Ccal$
                \STATE ${\sf resp}\gets\Resp^H({\sf com},{\sf ch})$
            \UNTIL{$z\neq\bot$}
            \STATE $H(m,{\sf com}):={\sf ch}$
            \RETURN $({\sf com},{\sf resp})$
        \end{algorithmic}
    \end{minipage}
    \caption{The oracles ${\bf Prog}$ and ${\bf Trans}$}
\end{figure}
In the game hop from $\Acal^{H,{\bf Prog}}$ to $\Acal^{H,{\bf Trans}}$, the set $S$ of entries that is reprogrammed by one but not the other, consists of pairs of the form $(m,{\sf com})$ for any transcript $({\sf com},{\sf ch},{\sf resp})$ generated at an iteration with $z=\bot$. The patch in \cite{FSwA23}, which assume $\Com$ and $\Resp$ make no query to $H$, crucially relies on ${\sf com}$ having high min-entropy in the view of $\Acal$, and so it will likely not notice the reprogramming. However, in the setting of \cite{beullens2020calamari, lai2021collusion, BDKLP22}, where $\Com^H,\Resp^H$ do make queries to $H$, the transcript ${\sf com}$ is chosen dependent on $H$, and hence it may no longer contain as high min-entropy in the view of $\Acal^{H,{\sf Prog}}$. Hence, the existing patch does not apply.

\section{Security proofs for $\Sigma_{GA}$}\label{sec:sigma_GA_proof}
\subsection{Proof of Lemma~\ref{lem:sigma_correct}}
\begin{proof}[Proof of correctness]
  By the definition of $\Com$ and $\Ver$, any honestly generated
  $(\mathsf{com}, \mathsf{ch}, \mathsf{resp})$ based on $(\nset{E}{i},
  s)\in R_n$ will be accepted as long as the set $\nset{E^\beta}{i}$
  is pairwise distinct. Since $\mathcal{GA}$ is free and transitive,
  there is a unique $g\in G$ s.t. $gE_i = E_j$. Thus,
  $E_i^\beta=E_j^\beta$ if and only if
  $(\Delta_j\Delta'_j)\inv\Delta_i\Delta'_i = g$, which happens with
  negligible probability since all $\Delta$'s are honestly sampled.
  Hence with probability $1-n\cdot \mathsf{negl}(\lambda)$, the set
  $\nset{E^\beta}{i}$ are all distinct, and hence $\Ver$ accepts.

  For the function $\Open$, note that if $(E_m, s_m)\in R_m$ and
  $(E_k, s)\in R$, then $E^{\mathsf{Open}} = \Delta_k\Delta_k'sE_m =
  \Delta_k\Delta_k'ss_mE$, hence $s_mE_k^\beta = E^{\mathsf{Open}}$.
  As argued previously, $\nset{E^\beta}{i}$ are all distinct with
  probability $1-\mathsf{negl}(\lambda)$, and $k$ would be unique if
  this is the case. Thus the probability that $\Open$ outputs $E_k$ is
  overwhelming, concluding the proof that $\Sigma_{GA}$ is correct.

\end{proof}

\begin{proof}[Proof of statistical HVZK]
  The construction of $\Sim$ is given in the following algorithm. We will
  show that $\Sim$ is in fact a perfect simulator for $\Trans$.
\begin{algorithm}[H]
\caption{$\Sim(E_m, \nset{E}{i}, E_k)$}
\setlength\multicolsep{0pt}
\begin{multicols}{2}
\begin{algorithmic}[1]
    \STATE $\mathsf{ch}\Usample \sch$
    \STATE $b\Usample G$, $\tau\Usample \textit{sym}(n)$
    \IF {$\mathsf{ch}=1$}
        \STATE $\nset{\Delta}{i}, \nset{D}{i}\Usample G$
        \STATE $\forall i\in[n] :  E_i^\alpha := \Delta_iE_i$
        \STATE $\forall i\in[n] : E_i^\beta :=  \Delta_iD_iE$
        \STATE $E^{\mathsf{Open}} := \Delta_kD_kE_m$
    \ELSIF {$\mathsf{ch}=2,3,4$}
        \STATE $\nset{D}{i},\nset{\Delta'}{i}\Usample G$
        \STATE $\forall i\in[n] :  E_i^\alpha := D_iE$
        \STATE $\forall i\in[n] : E_i^\beta := \Delta'_iE_i^\alpha$
        \STATE $E^{\mathsf{Open}} := D_k\Delta'_kE_m$
    \ENDIF
    \STATE $\forall i\in[n] :  E_i^\gamma := bE_i^\beta$
    \IF {$\mathsf{ch}=1,2,3$}
        \STATE $E^{\mathsf{Check}} := bE^{\mathsf{Open}}$
    \ELSIF {$\mathsf{ch}=4$}
        \STATE $l := \Delta_kD_kb$
        \STATE $E_{km}^{\mathsf{Check}} := lE_m$
    \ENDIF
    \RETURN $(\mathsf{com}, \mathsf{ch}, \mathsf{resp})$ with the same definition as honest $\Com$ and $\Resp$
\end{algorithmic}
\end{multicols}
\end{algorithm}
\ifnum\submission=1\vspace{-2em} \fi
Since $\mathcal{GA}$ is free and transitive, for every $E_i\in\Ecal$,
there exists a unique $s_i\in G$ s.t. $s_iE = E_i$. In $\Sim$, we can
thus set $\Delta'_i = D_is_i\inv$ in case $\mathsf{ch}=1$ and
$\Delta_i = D_is_i\inv$ in case $\mathsf{ch}=2,3,4$. Since the
distribution of $D_is_i\inv$ is uniformly random, $\Sim$ generates
identical distributions for $\Delta$'s as $\Trans$. Thus the output
distribution of $\Sim$ should also be identical to the real
transcript. Checking that verification passes for
all cases  shows that $\Sim$ is a perfect simulator.
\end{proof}

\section{Judging the opening}\label{sec:judge}

Due to the majority voting that we have adopted in our opening design, we do not know yet
how to construct a proof for the exact opening output.
However, as a natural byproduct of our construction, we could also empower the
manager to generate a proof $\pi$ additionally from $\textbf{Open}$
that could be publicly verified showing for multiple sessions
$s_m E_k^{\beta}=E^{\mathsf{Open}}$ (as in
Section~\ref{sec:sigma_construction}), which is done with a slight
twist to Couveignes' sigma protocol, as defined in $\textbf{Judge}$ below.

\begin{itemize}
    \item $\textbf{Open}(\mathsf{msk}, S=\{\pk_i\}_{i\in [n]},m,\sigma)\rightarrow (\pk,\pi)\in (S\cup\{\bot\})\times\{0,1\}^*$: The redefined open algorithm not only reveals signer identity $\pk$ but also produces a publicly verifiable proof $\pi$ for it.
    \item $\textbf{Judge}(\mathsf{mpk},S=\{\pk_i\}_{i\in [n]},\sigma,\pk,\pi)\rightarrow\mathsf{acc}\in\{0,1\}$: The judge algorithm accepts if the manager opened correctly,
\end{itemize}
Note that in Section~\ref{sec:sigma_construction}, the opening within the sigma protocol is done by picking the index $k$ such that $s_m E^\beta_k=E^{\mathsf{Open}}$. A manager could therefore prove this equality in a Schnorr-like manner, re-starting from the sigma protocol $\Sigma_{GA}$ with three additional algorithms $\textbf{JCommit},\textbf{JResp},\textbf{JVerify}$.
\begin{itemize}
    \item $\textbf{JCommit}(s_m:=\mathsf{msk},\{E_i\}_{i\in[n]}:=\{\pk_i\}_{i\in[n]}, \mathsf{com})$:
    \begin{algorithmic}[1]
        \STATE $b'\Usample G$
        \STATE $\textbf{parse}\ \mathsf{com} = \mathsf{(\nset{E^\alpha}{i}, \nset{E^\beta}{i}, \tau(\nset{E^\gamma}{i}), E^{\mathsf{Open}}, E^{\mathsf{Check}})}$
        \COMMENT{We use $\tau(\bullet)$ as a lazy convention of sending a permuted list}
        \STATE $E^{\mathsf{Judge}}:= {b'} E^{\mathsf{Open}}$
        \STATE $E^{b'}_m:=b's_mE$
        \RETURN $(\mathsf{jcom},\mathsf{jst}) = \left((E^{\mathsf{Judge}},E^{b'}_m),(b',s_m)\right)$
    \end{algorithmic}
    \item $\textbf{JResp}(E_m,\{E_i\}_{i\in[n]},\mathsf{jcom},\mathsf{jch},\mathsf{jst})$:
    \begin{algorithmic}[1]
        \STATE \textbf{parse} $\mathsf{jst} = (b',s_m)$
        \IF{$\mathsf{jch}=0$}
            \RETURN $\mathsf{jresp}:=b'$
        \ENDIF
        \IF{$\mathsf{jch}=1$}
            \RETURN $\mathsf{jresp}:=l'=b's_m$
        \ENDIF
    \end{algorithmic}
    \item $\textbf{JVerify}(E_m:=\mathsf{mpk},\{E_i\}_{i\in[n]}:=\{\pk_i\}_{i\in [n]}, E_k:=\pk, \mathsf{com},\mathsf{jcom},\mathsf{jch},\mathsf{jresp})$:
    \begin{algorithmic}[1]
        \vspace*{-1em}
        \STATE \textbf{parse} $\mathsf{com}=(\nset{E^\alpha}{i}, \nset{E^\beta}{i}, \tau(\nset{E^\gamma}{i}), E^{\mathsf{Open}}, E^{\mathsf{Check}})$
        \STATE \textbf{parse} $\mathsf{jcom}=(E^{\mathsf{Judge}},E^{b'}_m)$
        \IF{$\mathsf{jch}=0$}
            \STATE \textbf{check} $E^{\mathsf{Judge}}= b' E^{\mathsf{Open}}$
            \STATE \textbf{check} $E^{b'}_m=b'E_m$
        \ENDIF
        \IF{$\mathsf{jch}=1$}
            \STATE \textbf{check} $E^{\mathsf{Judge}}=l'E_k^{\beta}$
            \STATE \textbf{check} $E^{b'}_m=l'E$
        \ENDIF
        \RETURN 1 \textbf{if all check pass}
    \end{algorithmic}
\end{itemize}

For each run of $\textbf{Commit}\rightarrow (\mathsf{com},\mathsf{st})$, we have to do additionally
$\iota$ repetitions of $\textbf{JCommit}$ (and thus
$\iota t$ repetitions in total) to confirm that it is opened to the $k$-th signer with $\mathsf{negl}(\iota)$ probability. Similar as before, the Fiat-Shamir transform is applied for non-interactivity as follows.
\begin{itemize}
    \item $\textbf{Open}(\mathsf{msk}, S=\{\pk_i\}_{i\in [n]},m,\sigma)$
    \begin{algorithmic}[1]
        \STATE $t=t(\lambda,|S|)$; $\iota=\lambda$
        \STATE $\mathbf{parse\ } \sigma = ({\mathsf{com}}, {\mathsf{resp}})$
        \STATE $\forall j \in [t], \mathsf{out}_j \leftarrow \Sigma_{GA}.\Open(\mathsf{msk}, S, \mathsf{com}_i)$
        \STATE $\forall (i,j)\in [\iota]\times [t], \mathsf{jcom}_{i,j}\leftarrow \Sigma_{GA}.\textbf{JCommit}(\mathsf{msk},\{E_i\}_{i\in[n]},\mathsf{com}_j)$
        \STATE ${\sf jch}:=\{\mathsf{jch}_{i,j}\}_{(i,j)\in [\iota]\times [t]}\leftarrow H(\sigma,\{\mathsf{jcom}_{i,j}\}_{(i,j)\in[\iota]\times [t]})$
        \STATE $\forall (i,j)\in [\iota]\times [t], (\mathsf{jresp}_{i,j},\mathsf{jst}_{i,j})\leftarrow \Sigma_{GA}.\textbf{JResp}(E_m,\{E_i\}_{i\in[n]},\mathsf{jcom}_{i,j},\mathsf{jch}_{i,j},\mathsf{jst}_{i,j})$\vspace*{-1em}
        \STATE $\pk = \mathbf{Maj}(\aset{\mathsf{out}}{j}{[t]})$ \COMMENT{$\mathbf{Maj}$ outputs the majority element of a set. In case of ties, $\mathbf{Maj}$ outputs a random choice of the marjority elements.}
        \STATE $\pi:=\{\mathsf{jcom}_{i,j},\mathsf{jresp}_{i,j}\}_{(i,j)\in [\iota]\times [t]}$
        \RETURN $(\pk,\pi)$
    \end{algorithmic}
    \item $\textbf{Judge}(\mathsf{mpk},S=\{\pk_i\}_{i\in [n]},\sigma,\pk,\pi)$:
    \begin{algorithmic}[1]
        \RETURN $0$ \textbf{if} $\pk=\bot$
        \STATE $t=2\lambda|S|$; $\iota=\lambda$
        \STATE $\mathbf{parse\ } \sigma = ({\mathsf{com}}, {\mathsf{ch}}, {\mathsf{resp}})$
        \STATE \textbf{parse\ } $\pi=\{\mathsf{jcom}_{i,j},\mathsf{jresp}_{i,j}\}_{(i,j)\in [\iota]\times [t]}$
        \STATE ${\sf jch}:=\{\mathsf{jch}_{i,j}\}_{(i,j)\in [\iota]\times [t]}\leftarrow H(\sigma,\{\mathsf{jcom}_{i,j}\}_{(i,j)\in[\iota]\times [t]})$
        \STATE $\forall j\in[t], \mathsf{jout}_j\leftarrow \bigwedge_{i\in[\iota]}\Sigma_{GA}.\textbf{JVerify}(\mathsf{mpk},\{\mathsf{E}_i\}_{i\in[n]},\pk,\mathsf{com}_j,\mathsf{jcom}_{i,j},\mathsf{jch}_{i,j},\mathsf{jresp}_{i,j})$\vspace*{-1em}
        \RETURN $1$ \textbf{if} $\sum_{j\in[t]}\mathsf{jout}_j\geq \lambda$
    \end{algorithmic}
\end{itemize}
Here, a corrupted manager gets to selectively generate a partial proof $\{E_{\mathsf{out}_j},\mathsf{jcom}_{i,j},\mathsf{jch}_{i,j},\mathsf{jresp}_{i,j}\}_{(i,j)\in[\iota]\times\Jcal}$ where $\Jcal\subseteq[t]$ is adaptively chosen. So long as we have $\sum_{j\in\Jcal}\mathsf{jout}_j\geq \lambda$, the judged proof is accepted. This does not prevent the manager from generating accepted proofs that open to different members when $\#\{E_{\mathsf{out}_j}\}>1$, which could happen if the corresponding signature is generated by multiple colluding signers. Otherwise, incriminating an honest non-signer would require to make up at least $\lambda$ valid sessions of $\textbf{Commit}$, which will succeed with only negligible probability, i.e. for any efficient adversary $A$, any $t(\lambda,n)=n\cdot {\sf poly}(\lambda)$ and valid master key pair $(\mathsf{mpk},\mathsf{msk})\in\Kcal\Pcal_m$,
$$
\Pr[A\text{ wins }G_{A}^{\mathsf{JUF}}({\sf mpk},{\sf msk})]\leq \mathsf{negl}(\lambda),
$$
where the {\em judging unforgeability game} $G_{A}^{\mathsf{JUF}}$ is as specified below.
\begin{algorithm}[H]
\caption{$G_{A}^{\mathsf{JUF}}$: Judging unforgeability game}
\begin{algorithmic}[1]
    \STATE $(\pk,\sk)\leftarrow\Gen(1^\lambda)$
    \STATE $(S, m^*, \sigma^*)\leftarrow A^{\textbf{Sign}(\bullet, \bullet,
        \bullet, \sk),H}(\pk)$
    \STATE $A$ wins if $\sigma^*$ is not produced by querying $\textbf{Sign}({\sf mpk},S^*,m^*,\sk)$, $1\leftarrow
      \Vrfy(\mathsf{mpk}, S, m^*, \sigma^*)$, $({\sf out},\pi)\leftarrow
      \Open(\mathsf{msk}, S, m, \sigma^*)$ satisfies ${\sf out} \in
      \{\pk,\perp\}$ and $1\leftarrow\textbf{Judge}(\mathsf{mpk},S,\sigma,{\sf out},\pi)$.
\end{algorithmic}
\end{algorithm}

\section{Isogeny class group action}\label{sec:CSIDH}
Here we briefly cover the basics for {\em elliptic curve
  isogenies}. For simplicity, we consider a working (finite) field
$\Fbb_q$ with characteristic $p>3$. An isogeny $\phi$ between elliptic
curves $E_1\to E_2$ defined over an algebraic closure $\bar\Fbb_q$ is
a surjective homomorphism between the groups of rational points
$E_1(\bar\Fbb_q)\to E_2(\bar\Fbb_q)$ with a finite kernel. If,
additionally, $\phi$ is assumed {\em separable}, i.e. the induced
extension of function fields $\phi^*:\bar\Fbb_q(E_2)\hookrightarrow
\bar\Fbb_q(E_1)$ by $\bar\Fbb_p(E_2)\ni f\mapsto
f\circ\phi\in\bar\Fbb_p(E_1)$ is separable, then for any finite
subgroup $H\leq E_1(\bar\Fbb_p)$, there is an isogeny $\phi:E_1\to
E_2$ having $H$ as its kernel, and the co-domain curve is furthermore
uniquely determined up to isomorphisms (in $\bar\Fbb_q$). We
refer to the co-domain curve as the {\em quotient curve}, denoted $E_1/H$. A
corresponding isogeny could be computed using Velu's formula specified
in \cite{velu1971isogenies}, which works by expanding the coordinates
of $Q = \phi(P)$ as follows,
\begin{align*}
    &x(Q) = x(P) + \sum_{R\in H\setminus\{0\}} \left(x(P+R) - x(R)\right), \\
    &y(Q) = y(P) + \sum_{R\in H\setminus\{0\}} \left(y(P+R) - y(R)\right).
\end{align*}
The separable degree $\deg_{\mathsf{sep}}\phi$ is defined as the
separable degree for $\phi^*$, which coincides with the size of its
kernel $\#\ker\phi$, and since any isogeny could be acquired by
precomposing Frobenius maps to a separable isogeny, i.e. of form
$\phi\circ\pi_p^k$ where $\phi$ is separable, we can (equivalently)
define the (full) degree $\deg\left(\phi\circ\pi_{p^k}\right) =
\deg_{\mathsf{sep}}(\phi)p^k$. From now on, we will assume
separability of isogenies unless otherwise specified, and therefore
$\deg\phi=\deg_{\mathsf{sep}}\phi$ in this case.

For large degree $\phi$, when both domain $E_1$ and co-domain $E_2$
(supersingular) curves are prescribed, it could be hard to determine
the kernel (and thus $\phi$). The current best-known (generic) quantum
algorithm is {\em claw finding}, which takes $\tilde
O\left(\deg(\phi)^{1/3}\right)$ operations.

One important structure for isogenies is the so-called {\em isogeny
  class group action}, which was first used for cryptographic
constructions by \cite{couveignes2006hard,stolbunov2012cryptographic},
and was viewed as a weaker alternative for discrete logarithm.
However, although theoretically feasible, the instantiated group
action used to rely heavily on techniques regarding the so-called {\em
  modular polynomials}, which is computationally expensive in
practice. Later on, improvements in the {\em Commutative SIDH} (CSIDH)
\cite{castryck2018csidh} scheme got rid of these techniques. Concretely, the
space $X$ is instantiated as a set
$\Ecal\ell\ell_p(\Ocal,\pi_p)=\{E/\Fbb_p \text{ supersingular elliptic
  curves}\}/\cong_{\Fbb_p}$ acted by their ideal class group $\Cl$ of
the $\Fbb_p$-rational endomorphism ring $\Ocal = \mathsf{End}_p(E)$
where $E\in\Ecal\ell\ell_p(\Ocal,\pi_p)$ but $\Ocal\otimes\Qbb$
tensored as a $\Zbb$-module is identical regardless of the choice of
$E\in\Ecal\ell\ell_p(\Ocal,\pi_p)$ thus so is $\Cl$. The additional
parameter $\pi_p$ denotes the $p$-power Frobenius
$\pi_p:(x,y)\mapsto(x^p,y^p)$. Elements of $\Cl$ are
equivalence classes ${\frak a}$ of ideals of the (partial) endomorphism
ring $\Jcal\lhd\mathsf{End}_p(\Ocal)$. Any such ideal class
${\frak a}\in\Cl$ therefore acts on the curves by sending
$E\in\Ecal\ell\ell_p(\Ocal,\pi_p)$ to the quotient curve ${\frak
  a}\cdot E:=E/E[\Jcal]$ where $\Jcal\in{\frak a}$ is a representative
of the equivalence class ${\frak a}$ and $E[\Jcal] =
\bigcap_{f\in\Jcal}\ker{f}$ is the simultaneous kernel of $\Jcal$.

The working base field $\Fbb_p$ for CSIDH is carefully selected such
that $p=4\ell_1\cdots\ell_n-1$ where each $\ell_i>2$ is a small prime
generally referred to as an {\em Elkies prime}. This allows one to
generate a heuristically large enough sub-covering $\{{\frak
  l}_1^{e_1}\dots{\frak l}_n^{e_n}| \forall i: |e_i|\leq b_i\}$ of
$\Cl$ where each prescribed $b_i$ is small\footnote{For
 CSIDH-512 \cite{castryck2018csidh} proposes
$b_1=\dots=b_n=5$.} and each ${\frak
  l}_i^{\pm 1}$ is the class of ideal $\langle \pi_p\mp
1,\ell_i\rangle$. The indices $(e_1,\dots,e_n)$ thus 
represent the ideal class ${\frak l}_i^{e_1}\cdots{\frak l}_n^{e_n}$,
making it easier to compute the co-domain curve. In particular, for a
curve $E\in\Ecal\ell\ell_p(\Ocal,\pi_p)$ and any choice of $\ell_i$,
the curve ${\frak l}_i\cdot E:=E/E[\langle\pi_p-1,\ell_i\rangle]$ is
computed by sampling a generator of the kernel,
\ifnum\submission=1\vspace{-0.5em}\fi
\begin{align*}
E[\langle \pi_p-1,\ell_i\rangle] = E(\Fbb_p)[\ell_i] = \{P\in E(\Fbb_p)| \ell_i P = 0\},
\end{align*}
\ifnum\submission=1\vspace{-0.2em}\fi
which is a one dimensional $\Zbb/\ell$-linear eigen-subspace of
$\pi_p$ within the $\ell_i$-torsion $E[\ell_i]$. For the opposite 
direction, one can compute ${\frak l}_i^{-1}\cdot E=\left({\frak
    l}_i\cdot E^t\right)^t$ where the superscript $^t$ is referred to
as the quadratic twist of the specified curve, by taking the
convention that the curve is fixed when its $j$-invariant is $1728$, or
equivalently, this can be done by sampling from the other
$\Zbb/\ell_i$-linear eigen-subspace of $\pi_p$ in $E[\ell_i]$, which
sits in the quadratic extension $E(\Fbb_{p^2})$.

We also list here some well-known properties for the considered class
group action. First, the class group $\Cl$ {\em commutes}, which is a
direct result of the fact that the $\Fbb_p$-rational endomorphism ring
$\mathsf{End}_p(E)$ commutes. Second, as noted in
\cite[Theorem~7]{castryck2018csidh}, $\Cl$ acts {\em freely and
  transitively} on $\Ecal\ell\ell_p(\Ocal,\pi_p)$, which means that
for all $E_1, E_2\in \Ecal\ell\ell_p(\Ocal,\pi_p)$, there exists a
unique ${\frak a}\in\Cl$ such that ${\frak a}\cdot E_1 = E_2$.
Finally, elements in $\Ecal\ell\ell_p(\Ocal,\pi_p)$ can be {\em
  efficiently verified}. We note that a curve $E$ is supersingular if
and only if it has $p+1$ points over $\Fbb_p$. This can be efficiently
tested by finding some $P\in E(\Fbb_p)$ with order
$\mathsf{ord}(P)\geq 4\sqrt{p}$ dividing $p+1$.  A random point $P$
sampled from $E(\Fbb_p)$ satisfies such a condition with high
probability if $E$ is supersingular, and whether it does can be
verified efficiently as follows. If $(p+1)P\neq 0$, then
$\mathsf{ord}(P)$ does not divide $p+1$ and $E$ is ordinary.
Otherwise, we can perform the so-called {\em batch co-factor
  multiplication} computing $P_i = \frac{p+1}{\ell_i}P$ for each $i$,
by using convention that $\ell_0=4$. This allows us to determine
$\mathsf{ord}(P)=\prod_i \mathsf{ord}(P_i)$.

For typical cryptographic constructions such as CSIDH,
additional heuristic assumptions are required to sample a random
element from the class group (as in Definition~\ref{def:DDH}). This is
because the ``CSIDH-way'' for doing this is by sampling exponents 
$(e_1,\dots,e_n)$ satisfying $\forall i: |e_i|\leq b_i$,
and the resulting distribution for ideals ${\frak
  l}_1^{e_1}\dots{\frak l}_n^{e_n}$ is generally non-uniform within
$\Cl$. To get rid of such heuristics, one could instead work with
specific parameters, where a bijective (yet efficient) representation
of ideals is known. For instance, in \cite{beullens2019csi}, the
structure of $\Cl$ is computed, including a full generating set of
ideals ${\frak l}_1,\dots,{\frak l}_n$ and the entire lattice
$\Lambda:=\{(e_1,\dots,e_n)| {\frak l}_1^{e_1}\dots{\frak
  l}_n^{e_n}=\mathsf{id}\}$. Evaluating the group action is just a
matter of approximating a {\em closest vector} and then evaluating the
residue as in CSIDH. In this work, we will be working with such a
``perfect'' representation of ideals, unless otherwise specified.

As a remark, we note that the D-CSIDH problem for characteristic
$p=1\mod 4$ is known to be broken \cite{castryck2020breaking}.
Nevertheless, the attack is not applicable to the standard CSIDH
setting where $p=3\mod 4$.

\section{Sigma protocol}\label{sec:sigpro}

A sigma protocol should satisfy the following three properties.

\begin{definition} (Correctness)
A sigma protocol is correct if for any $(x,w)\in R$, the probability
\begin{equation*}
\Pr\left[
(\mathsf{com}, st)\leftarrow P_1(x,w),
    \mathsf{ch}\Usample\Ccal,
    \mathsf{resp}\leftarrow P_2(st, \mathsf{ch}),
    0\leftarrow V(x, \mathsf{com}, \mathsf{ch}, \mathsf{resp})
\right]
\end{equation*}
is negligible.
\end{definition}

\begin{definition} (Honest Verifier Zero Knowledge/HVZK) Let
  $\Trans(x,w)\rightarrow (\mathsf{com}, \mathsf{ch}, \mathsf{resp})$
  be a function that honestly executes the sigma protocol and outputs
  a transcript. We say that the sigma protocol is HVZK if there exists
  a simulator $\Sim(x)\rightarrow (\mathsf{com}, \mathsf{ch},
  \mathsf{resp})$ such that the output distribution of $\Trans(x,w)$
  and $\Sim(x)$ is indistinguishable.
\end{definition}

\begin{definition} ($\mu$-special soundness) A sigma protocol is
  $\mu$-special sound if there exist an efficient extractor $\Ext$ such
  that, for any set of $\mu$ transcripts with the same
  $(x,\mathsf{com})$, denoted as $(x, \mathsf{com},
  \aset{\mathsf{ch}}{i}{[\mu]}, \aset{\mathsf{resp}}{i}{[\mu]})$, where
  every $\mathsf{ch}_i$ is distinct, the probability
\begin{equation*}
  \Pr\left[(x,s)\notin R\wedge \forall i\in[\mu], \mathsf{acc}_i=1:
    \substack{
      \forall i\in[\mu],\ \mathsf{acc}_i \leftarrow V(x, \mathsf{com}, \mathsf{ch}_i, \mathsf{resp}_i),\\
      s \leftarrow\mathbf{Ext}(x, \mathsf{com}, \aset{\mathsf{ch}}{i}{[\mu]}, \aset{\mathsf{resp}}{i}{[\mu]})
    }
  \right]
\end{equation*}
is negligible.
\end{definition}

Here, we formulate a more general form of special soundness. While
most sigma protocol constructions in the literature adopt 2-special
soundness, any $\mu$-special sound protocol with constant $\mu$ can be
similarly transformed into a signature scheme, simply by applying more
rewinding trials.

\section{Analysis in CROM}\label{sec:cROM}
\subsection{The forking lemma}\label{sec:forkl}
The concept of the forking lemma is as follows. In the random oracle
model, let $A$ be an adversary that can with non-negligible
probability generate valid transcripts $(m, \mathsf{com}, \mathsf{ch},
\mathsf{resp})$ with $\mathsf{ch}=H(m,\mathsf{com})$. Since $H$ is a
random oracle, for some $(m,\mathsf{com})$, $A$ should be able to
succeed on sufficiently many different $\mathsf{ch}'$ from $H$ in
order to achieve an overall non-negligible success probability. If
we can rewind and rerun $A$ with different oracle outputs on
$H(m,\mathsf{com})$, we should be able to get multiple accepting
transcripts.

To dig a little bit deeper, we can construct an efficient algorithm
$B$ that runs $A$ as a subroutine, where $A\rightarrow
(m,\mathsf{com}, \mathsf{ch}, \mathsf{resp})$ has at most $Q$ oracle
queries. The tuple $(m,\mathsf{com})$ should, with all but negligible
probability, be among one of the $Q$ queries. $B$ first guesses the {\em
  critical query} $i\in[Q]$, the index where
$\Qcal_i=(m,\mathsf{com})$ is being queried. Then, $B$ replays $A$
with fixed random tape, fixed oracle outputs for the first $i-1$
queries, and fresh random oracle outputs for the remaining queries. If
the query guess $i$ and fixed randomness are ``good,'' which should
happen with non-negligible probability, then among sufficiently many
retries we should get $t$ successful outputs of $A$, which are
transcripts with identical $(m,\mathsf{com})$ with distinct challenges
$\mathsf{ch}$'s. For a rigorous proof, we refer the reader to
\cite{pointcheval1996security,pointcheval2000security} for the forking
lemma with 2 transcripts and \cite{brickell2000design} for a
$\mu$-transcript version.

Here, we give a reformulated version of the improved forking lemma
proposed by \cite{brickell2000design}. We renamed the variables to fit
our notion and restricted parameters to the range that is sufficient for
our proof.

\begin{theorem} (The Improved Forking Lemma\cite{brickell2000design},
  Reformulated)
\label{forkinglemma}
Let $A$ be a probabilistic polynomial-time algorithm and ${\mathsf Sim}$ be a
probabilistic polynomial-time simulator which can be queried by $A$.
Let $H$ be a random oracle with image size $|H| \geq 2^\lambda$. If
$A$ can output some valid tuple $(m, \mathsf{com}, \mathsf{ch},
\mathsf{resp})$ with non-negligible probability $\varepsilon \geq
1/poly(\lambda)$ within less than $Q$ queries to the random oracle,
then with $O(Q\mu\log \mu/\varepsilon)$ rewinds of $A$ with different
random oracles, $A$ will, with at least constant probability, output
$\mu$ valid tuples $(m, \mathsf{com}, \mathsf{ch}_i, \mathsf{resp}_i)$
with identical $(m,\mathsf{com})$ and pairwise distinct
$\mathsf{ch}_i$'s.
\end{theorem}

\subsection{Classical Security}\label{sec:classical_sec}

For the proof of Theorem \cref{ARSsecure} we again break down the theorem into proving each security property, i.e. correctness, anonymity and unforgeability. For correctness, there is no difference between classical and quantum settings, but since the proof does not exploit ``quantum-ness'' of an adversary, we put it in this section as well.

\paragraph{Proof of \cref{lem:correct}}
\begin{proof}
  For any master key pair $(\mathsf{mpk},
  \mathsf{msk})\in\Kcal\Pcal_m$, any key pair
  $(\pk,\sk)\in\Kcal\Pcal$, and any set of public keys $S$ such that
  $\pk\in S$, we directly have $(\mathsf{mpk}, \mathsf{msk})\in R_m$
  and $(S, \sk)\in R_n$ where $n=|S|$. Let
  $\sigma\leftarrow\Sign(\mathsf{mpk}, S, m, \sk)$ be an honest
  signature on message $m$ and ring $S$. Notice that in an honest
  execution of $\Sign$, each $\mathsf{com}_j$ and $\mathsf{resp}_j$ is
  honestly generated according to $\Sigma$. Thus by the
  correctness of $\Sigma$, we know for ${\sf ch}:=H({\sf com},m)$ and every $j\in[t]$ with
  probability $1-\mathsf{negl}(\lambda)$, that $1\leftarrow
  \Sigma.\Ver(\mathsf{mpk}, S, \mathsf{com}_j, \mathsf{ch}_j,
  \mathsf{resp}_j)$ and $\pk\leftarrow\Sigma.\Open(\mathsf{mpk},
  S, \mathsf{com}_j)$. Hence we directly obtain that, with probability
  $1-t\cdot \mathsf{negl}(\lambda) = 1-\mathsf{negl}(\lambda)$, we have that
  $1\leftarrow \Ver(\mathsf{mpk}, S, m, \sigma)$ and $\pk\leftarrow
  \Open(\mathsf{msk}, S, m, \sigma)$. This concludes the proof that
  $\mathcal{ARS}_{\Sigma}$ is correct.
\end{proof}

\paragraph{Proof of \cref{lem:anonymity_ROM}}
\begin{proof}
The anonymity of $\mathcal{ARS}_{\Sigma}$ follows immediately from the CWI
property of $\Sigma$. For any efficient adversary $A$ with at most $q$
queries to the random oracle, it can have at most
$q\cdot{\sf negl}(\lambda)\leq \mathsf{negl}(\lambda)$ advantage on distinguishing $\Sign^*$
and $(\Trans^*)^t$. And by CWI from $\Sigma$, we have
$\Trans^*(\mathsf{mpk}, S, \sk_{id_0})\approx_c \Trans^*(\mathsf{mpk},
S, \sk_{id_1})$. Hence we can directly conclude that
$\Sign^*(\mathsf{mpk}, S, \sk_{id_0})\approx_c \Sign^*(\mathsf{mpk},
S, \sk_{id_1})$, which proves that $\mathcal{ARS}_\Sigma^t$ is anonymous.
\end{proof}

\paragraph{Proof of \cref{lem:unforgeable_ROM}}
\label{sec:lem10}
\begin{proof}
Assume that there exists an efficient adversary $A'$ that wins
$G^{\mathsf{UF}}_{A'}(\mathsf{mpk}, \mathsf{msk})$ on some valid
key pair $(\mathsf{mpk}, \mathsf{msk})\in \Kcal\Pcal_m$ with
non-negligible probability. We aim to show that we can construct some
algorithm $B$ which runs $A'$ as a subroutine and extract an un-corrupted secret key.

As it doesn't hurt for a signing oracle to produce the challenges, let's abuse the notation as say the signing oracle returns not only the signature, but also those corresponding challenges.
First, we replace the $\Sign$ oracle with a simulator, so that
$A^{H}:=A'^{\Sim,H}$ can emulate the oracle responses to $A'$. We consider a
modified game $G^{\mathsf{UF},1}_{A'}$ which replaces the signing
oracle $\Sign(\bullet, \bullet, \bullet, \sk)$ by a
simulator $\Sim$ defined as follows:

\begin{itemize}
    \item $\Sim(\mathsf{mpk}, S, m, \pk\in S)$:
    \begin{algorithmic}[1]
        \STATE $t:=t(\lambda,|S|)$
        \STATE for $j\in[t]$, $(\mathsf{com}_j, \mathsf{ch}_j, \mathsf{resp}_j)\leftarrow \Sigma_{GA}.\Sim(\mathsf{mpk}, S, \pk\in S)$
        \STATE $\mathbf{program}$ $H(m,S,\mathsf{com}_1,\dots,\mathsf{com}_t) := (\mathsf{ch}_1,\dots, \mathsf{ch}_t)$
        \RETURN $\sigma = ({\mathsf{com}}, {\mathsf{resp}}) := ((\mathsf{com}_1,\dots, \mathsf{com}_t), (\mathsf{resp}_1, \dots, \mathsf{resp}_t)) $
    \end{algorithmic}
\end{itemize}
Since $\Sigma.\Sim$ is a statistical HVZK simulator, any adversary with
$Q={\sf poly}(\lambda)$ queries to $H$ cannot distinguish $\Sign$ from
$\Sim$ with non-negligible probability. Without loss of generality we assume
$A'^{\Sim,H}$ never produces a signature $\sigma^*$ from previous queries
to $\Sim({\sf mpk},S^*,m^*,\pk)$. Therefore for an forgery $({\sf com}^*,{\sf resp}^*)\gets A'^{\Sim,H}$ accepted with respect to the potentially reprogrammed $H$, $H(S^*,m^*,{\sf com}^*)$ must have not been reprogrammed, otherwise there must be a prior query of form $({\sf com^*},{\sf resp})\gets\Sim({\sf mpk},S^*,m^*,{\sf pk})$ for some ${\sf resp}$, but then ${\sf resp}\neq{\sf resp}^*$ is hard to find due to the computational unique-response property. Thus, $A$ should also win $G^{\mathsf{UF},1}_{A}$ with non-negligible probability.

Now, since $A$ wins $G^{\mathsf{UF},1}_{A}(\mathsf{mpk},
\mathsf{msk})$ only if it outputs some $(R, m^*, \sigma^*)$ such that
$\mathsf{out}^*\leftarrow \mathsf{Open}(\mathsf{msk}, R, m, \sigma^*)$
satisfies $\mathsf{out}^* = \pk$ or $\mathsf{out}^* = \perp$,
either $A$ wins with non-negligible probability with $\mathsf{out}^* =
\perp$, or $A$ wins with non-negligible
probability with $\mathsf{out}^* = \pk$. We deal with
these cases separately.

We first prove that there cannot exist efficient $A_\perp$ that wins
$G^{\mathsf{UF},1}_{A}(\mathsf{mpk}, \mathsf{msk})$ with
non-negligible probability with $\mathsf{out}^* = \perp$. If such
$A_\perp$ exists, we can construct an algorithm $B$ that honestly
generates $(\pk, \sk)$ and runs $A_\perp(\pk)$.
With non-negligible probability, $A_\perp$ will output
valid $(S, m, \sigma = ({\mathsf{com}}, {\sf ch},
{\mathsf{resp}}))$ such that $\perp\leftarrow\Open(\mathsf{msk},
S, m, \sigma)$. By applying the improved forking lemma (Theorem
\ref{forkinglemma}), with $r=O(Q/\varepsilon)$ rewinds of $A_\perp$,
it would, with constant probability, output $\mu$ valid signatures $(S,
m, \sigma^1,\dots, \sigma^\mu)$ with identical ${\mathsf{com}}$ and
pairwise distinct ${\mathsf{ch}}^c$, and that
$\perp\leftarrow\Open(\mathsf{msk}, S, m, \sigma^c)$ for all
$c\in[\mu]$. We now claim that with high probability, we can find some
parallel session $j\in[t]$ such that
$\perp\leftarrow\Sigma.\Open(\mathsf{msk}, S, \mathsf{com}_j)$
and $\mathsf{ch}^1_j, \dots, \mathsf{ch}^\mu_j$ are distinct. Note that
this is not trivially true, as the forking lemma only promises that
${\mathsf{ch}}^1, \dots, {\mathsf{ch}}^\mu$ are pairwise
distinct as vectors, so they might not be pairwise distinct on any
index $j$.

Let $T$ be the set of indices $j$ where
$\perp\leftarrow\Sigma.\Open(\mathsf{msk}, S, \mathsf{com}_j)$.
Since $\perp\leftarrow\Open(\mathsf{msk}, S, m, \sigma)$, by the
definition of $\Open$, $\perp$ must be (one of) the majority output
among the $t$ parallel sessions. Thus $|T|\geq t/(|S|+1) \geq
\lambda$. We say that $\mu$ challenges ${\mathsf{ch}'}^1, \dots,
{\mathsf{ch}'}^\mu$ are \textbf{good} on $T$ if there exists some $j\in
T$ such that $\mathsf{ch}'^1_j, \dots, \mathsf{ch}'^\mu_j$ are distinct.
For $\mu$ independently random challenges in $[\mu]^t$, the probability that
they are good on $T$ is $1-(1-(\mu!/\mu^\mu))^{|T|}=
1-\mathsf{negl}(\lambda)$.

Unfortunately, the challenges ${\mathsf{ch}}^1, \dots,
{\mathsf{ch}}^\mu$ obtained from rewinding $A$ are not necessarily
independent. To cope with this, we will need the fact that in each
rewind of $A$, the \textit{valid} ${\mathsf{ch}}$ is a new random
output from the new random oracle $H$. Thus, the finally output
$\mu$-tuple ${\mathsf{ch}}^1, \dots, {\mathsf{ch}}^\mu$ must be a
subset of $r=O(Q/\varepsilon)$ independent random samples from
$[\mu]^t$. By the union bound, the probability that \textit{all}
$\mu$-tuples in the $r$ samples are good on $T$ is
$1-\binom{r}{\mu}\mathsf{negl}(|T|)\geq 1-\mathsf{negl}(\lambda)$. Thus
we can find $j\in T$ such that $\mathsf{ch}^1_j, \dots,
\mathsf{ch}^\mu_j$ are distinct with probability $
1-\mathsf{negl}(\lambda)$.

For such $j$, we without loss of generality let $(\mathsf{ch}^1_j,
\dots, \mathsf{ch}^\mu_j) = (1,\dots, \mu)$ and consider $(S, \mathsf{com}_j,
\mathsf{resp}^1_j, \dots, \mathsf{resp}^\mu_j)$. Now $B$ achieves $\forall
c\in [\mu], 1\leftarrow\Sigma.\Ver(S, \mathsf{com}_j, c,
\mathsf{resp}^c_j)$, and 
$\perp\leftarrow\Sigma.\Open(\mathsf{msk}, S, \mathsf{com}_j)$.
Thus $B$ violates the $\mu$-special soundness property of $\Sigma$ and
brings a contradiction. Hence such $A_\perp$ cannot exist.

Now we consider the case where some $A$ wins $G^{\mathsf{UF},1}_{A}(\mathsf{mpk}, \mathsf{msk})$ with non-negligible probability with $\mathsf{out}^* = \pk$.

For such $A$, we can similarly
construct an algorithm $B$ that runs $A$ with input
$\pk$. Then again
by applying the improved forking lemma, with the same probability,
$r=O(Q/\varepsilon)$ rewinds of $A$ will output $\mu$ valid signatures
$(S, m, \sigma^1,\dots, \sigma^\mu)$ with identical
${\mathsf{com}}$ and pairwise distinct ${\mathsf{ch}}^c$,
so that $\pk\leftarrow\Open(\mathsf{msk}, S, m, \sigma^c)$ for all
$c\in[\mu]$. Again by the same argument as in the case of $A_\perp$, we can
with high probability find some $j\in[t]$ such that
$\pk\leftarrow\Sigma_{GA}.\Open(\mathsf{msk}, S, \mathsf{com}_j)$ and
$\mathsf{ch}^1_j, \dots, \mathsf{ch}^\mu_j$ are distinct.

Now, without loss of generality let $(\mathsf{ch}^1_j, \dots,
\mathsf{ch}^\mu_j) = (1,\dots, \mu)$ and consider $(S, \mathsf{com}_j,
\mathsf{resp}^1_j, \dots, \mathsf{resp}^\mu_j)$. We have $\forall
c\in [\mu], 1\leftarrow\Sigma_{GA}.\Ver(S, \mathsf{com}_j, c,
\mathsf{resp}^c_j)$, and that the challenge statement
$\pk\leftarrow\Sigma_{GA}.\Open(\mathsf{msk}, S, \mathsf{com}_j)$.
Thus by the $\mu$-special soundness property of $\Sigma_{GA}$, we can
extract the matching secret key $\sk\leftarrow\Sigma_{GA}.\Ext(S, \mathsf{com}_j,
\mathsf{resp}^1_j,\dots, \mathsf{resp}^\mu_j)$, such that $(\pk, \sk)\in
R$.

From the previous arguments, we see that if such efficient $A$
exists, then we can obtain an algorithm $B$ based on $A$ that, on
inputting random $\pk\in {\cal PK}$, output $\sk$ such that
$(\pk,\sk)\in R$ with non-negligible probability. Thus, we successfully construct
a secret extractor from adversary $A$ that wins the unforgeability game,
which concludes the proof that our $\mathcal{ARS}_{\Sigma}$ is unforgeable assuming
the instance relations are hard (to extract witness) for $\Sigma$.

\end{proof}

\section{Analysis in QROM}
\subsection{Proof of Lemma~\ref{lem:INT2EXT_QROM}}\label{sec:INT2EXT_QROM_proof}
\begin{proof}
We adopt the generalized Unruh's rewinding, as described in \cite[Lemma~29]{DFMS19}. Let $B(\pk)$ run as follows. First, execute $(S,m,{\sf com},{\sf st}_0)\gets A(\pk)$ as usual. Then, perform the following computation for $\mu$ times. For the $j$th time, freshly sample a challenge ${\sf ch}_j\gets\Sigma^{\otimes t}.\Ccal$ and then produce ${\sf resp}_j\gets A({\sf st}_{j-1},{\sf ch}_j)$, where the computation is {\em projectively executed}, i.e. after ${\sf resp}_j$ is produced, the computation is rewinded to where it started with ${\sf st}_{j-1}$, but with the internal state collapsed to ${\sf st}_j$ for the next run. After $\mu$ trials of rewinding, $B$ obtains $\mu$ samples of transcripts ${\sf com},\{{\sf ch}_j,{\sf resp}_j\}_{j\in[\mu]}$ sharing the same first message ${\sf com}$. Denote ${\sf com}^i,{\sf ch}_j^i,{\sf resp}_j^i$ to be the $i$th repetition of the $j$th rewinding. If there is some repetition (the $i$th) such that the corresponding transcript $({\sf com}^i,{\sf ch}^i_j,{\sf resp}^i_j)$ are distinct valid responses opened to $\pk$ or $\bot$ for all $j\in[\mu]$, then output $s\gets \Sigma.\Ext(S,{\sf com}^i,\{{\sf ch}^i_j\}_{j\in\mu},\{{\sf resp}_j^i\}_{j\in[\mu]})$, and abort otherwise.

As described earlier, by (\ref{eq:openable_special_sound}), we know that the output $s$ of $B(\pk)$ is always such that $(\pk,s)\in R$ whenever $B$ does not abort. Let ${\sf out}$, ${\sf out}^i$ and ${\sf acc}_j$ respectively be the output of $\Sigma^{\otimes t}.\Open({\sf msk},S,{\sf com})$, $\Sigma.\Open({\sf msk},S,{\sf com}^i)$ and $\Sigma^{\otimes t}.\Vrfy({\sf mpk},S,{\sf com},{\sf ch}_j,{\sf resp}_j)$.

The non-abort probability of $B$ can be union-bounded by two parts, namely
$$
\Pr\left[B(\pk)\text{ non-abort}\right]\geq \Pr\left[\substack{{\sf out}\in\{\pk,\bot\}\text{ and }\\ \forall j\in[\mu]: {\sf acc}_j=1}\right] - \Pr\left[\substack{\forall i\in[\mu]:{\sf out}^i\not\in\{\pk,\bot\}\text{ or }\\ {\sf ch}_1^i,\dots,{\sf ch}_\mu^{i}\text{ not distinct}} \middle| {\sf out}\in\{\pk,\bot\}\right]\;.
$$

We bound $\Pr\left[\substack{{\sf out}\in\{\pk,\bot\}\text{ and }\\\forall j\in[\mu]: {\sf acc}_j=1}\right]$ first. For every fixed choice $x^\circ:=(\pk^\circ,S^\circ,{\sf com}^\circ)$ of the joint random variable $x:=(\pk,S,{\sf com})$, identifying the (mixed) state of ${\sf st}_0$ conditioned on $x=x^\circ$ as $\rho_{x^\circ}$, and thus the un-conditioned state would be $\rho:=\sum_{x^\circ}\Pr[x=x^\circ] \rho_{x^\circ}$. Due to the perfect unique-response property, every time when a valid response ${\sf resp}_j$ is produced, it only disturbs the running state as a projector. Thus, we can define a family of projectors $\{P^{x^\circ}_{{\sf ch}^\circ}\}_{{\sf ch}^\circ\in\Sigma^{\otimes t}.\Ccal}$ where each projector $P_{{\sf ch}^\circ}^{x^\circ}$ on input ${\sf st}_{j-1}$ serve as the predicate that $({\sf com}^\circ,{\sf ch}_j,{\sf resp}_j)$ is an accepted transcript, i.e. $1\gets \Sigma^{\otimes t}.\Vrfy({\sf mpk},S^\circ,{\sf com}^\circ,{\sf ch}_j,{\sf resp}_j)$. Then
$$
\Pr\left[ \substack{\forall j\in[\mu]: \\1\gets \Sigma^{\otimes t}.\Vrfy({\sf mpk},S^\circ,{\sf com}^\circ,{\sf ch}_j,{\sf resp}_j)}\middle| x=x^\circ\right]=\sum_{{\sf ch}_1^\circ,\dots, {\sf ch}^\circ_\mu\in\Sigma^{\otimes t}.\Ccal}{\sf tr}\left( P^{x^\circ}_{{\sf ch}^\circ_1}\dots P^{x^\circ}_{{\sf ch}^\circ_\mu}\rho_{x^\circ}\right)\;.
$$
Expanding $\rho_{x^\circ}:=\sum_i\alpha_i\ket{\psi_i}\bra{\psi_i}$ via singular-value decomposition, we get
\begin{align*}
    &\sum_{{\sf ch}_1^\circ,\dots, {\sf ch}^\circ_\mu\in\Sigma^{\otimes t}.\Ccal}{\sf tr}\left( P^{x^\circ}_{{\sf ch}^\circ_1}\dots P^{x^\circ}_{{\sf ch}^\circ_\mu}\rho_{x^\circ}\right)
    = \sum_i\alpha_i\sum_{{\sf ch}_1^\circ,\dots, {\sf ch}^\circ_\mu\in\Sigma^{\otimes t}.\Ccal}\left\Vert P^{x^\circ}_{{\sf ch}^\circ_1}\dots P^{x^\circ}_{{\sf ch}^\circ_\mu}\ket{\psi_i}\right\Vert^2\\
    &\geq \sum_i \alpha_i \left(\sum_{{\sf ch}^\circ\in\Sigma^{\otimes t}.\Ccal}\left\Vert P^{x^\circ}_{{\sf ch}^\circ}\ket{\psi_i}\right\Vert^2\right)^{2\mu-1}
    \geq \left(\sum_{{\sf ch}^\circ\in\Sigma^{\otimes t}.\Ccal}{\sf tr}\left( P^{\pk^\circ}_{{\sf ch}^\circ}\rho_{x^{\circ}}\right)\right)^{2\mu-1}\;,
\end{align*}
where the first inequality is by \cite[Lemma~29]{DFMS19} and the second inequality is by Jensen's inequality. Summing over $x^\circ=(\pk^\circ,S^\circ,{\sf com}^\circ)$ such that $\{\pk^{\circ},\bot\}\ni{\sf out}^\circ:=\Sigma^{\otimes t}.\Open({\sf msk}, S^\circ,{\sf com}^\circ)$ with suitable probability, we obtain
\begin{align*}
    &\Pr\left[ \substack{{\sf out}\in\{\pk,\bot\}\text{ and }\\ \forall j\in[\mu]: {\sf acc_j}=1}\right]
    \geq \sum_{x^\circ\text{ s.t. }{\sf out}^\circ\in\{\pk^{\circ},\bot\}}\Pr\left[x=x^\circ\right]\left(\sum_{{\sf ch}^\circ\in\Sigma^{\otimes t}.\Ccal}{\sf tr}\left( P^{x^\circ}_{{\sf ch}^\circ}\rho_{x^{\circ}}\right)\right)^{2\mu-1}\\
    &\geq \left(\sum_{\substack{x^\circ\text{ s.t. }{\sf out}^\circ\in\{\pk^{\circ},\bot\}\\ {\sf ch}^\circ\in\Sigma^{\otimes t}.\Ccal}}\Pr\left[x=x^\circ\right]{\sf tr}\left( P^{x^\circ}_{{\sf ch}^\circ}\rho_{x^{\circ}}\right)\right)^{2\mu-1}
    =\Pr\left[A\text{ wins }G^{\sf int}_A({\sf mpk},{\sf msk})\right]^{2\mu-1}\;,
\end{align*}
where the second inequality is again via Jensen's inequality.

Next, for every $x^\circ=(\pk^\circ,S^\circ,{\sf com}^\circ)$ we define
$$
\Ical_{x^\circ}:=\big\{i\in[\mu]\,\big|\,\pk^\circ\text{ or }\bot\gets \Sigma.\Open({\sf msk},S^\circ,{\sf com}^\circ)\big\}\;,
$$
in order to bound $\Pr\left[\substack{\forall i\in[\mu]:{\sf out}^i\not\in\{\pk,\bot\}\text{ or }\\ {\sf ch}_1^i,\dots,{\sf ch}_\mu^{i}\text{ not distinct}} \middle| {\sf out}\in\{\pk,\bot\}\right]$
$$
\leq \max_{x^\circ\text{ s.t. }{\sf out}^\circ\in\{\pk^\circ,\bot\}}\Pr\left[\substack{\forall i\in\Ical_{x^\circ}:\\ {\sf ch}_1^i,\dots,{\sf ch}_\mu^{i}\text{ not distinct}} \middle| x=x^\circ\right]
= \max_{x^\circ\text{ s.t. }{\sf out}^\circ\in\{\pk^\circ,\bot\}}\Pr\left[\substack{\forall i\in\Ical_{x^\circ}:\\ {\sf ch}_1^i,\dots,{\sf ch}_\mu^{i}\text{ not distinct}}\right]
\;,
$$
where the last equality is due to the freshly sampled $\{{\sf ch}^i_j\}_{(i,j)\in[\mu]\times[\mu]}$ being independent with $x$. Note that, by the pigeonhole principle, ${\sf out}^\circ\in\{\pk^\circ,\bot\}$ implies $\#\Ical_{x^\circ}\geq\kappa$, thus the above can be bounded by
$$
\leq \left(1-\frac{{C \choose \mu}}{C^{\mu}}\right)^{\kappa}\leq\exp\left(\frac{-\kappa}{\mu^\mu}\right)\;,
$$
where $C:=\#\Sigma.\Ccal$ is the size of the challenge space.

Putting things together,
\begin{align*}
    &\Pr\left[\sk\gets B(\pk)\right]
    \geq \Pr[B(\pk)\text{ non-abort}]\\
    &\geq \Pr\left[A\text{ wins }G^{\sf int}_A({\sf mpk},{\sf msk})\right]^{2\mu-1}-\exp\left(\frac{-\kappa}{\mu^\mu}\right)\;,
\end{align*}
we conclude the proof.
\end{proof}

\subsection{Proof of Lemma~\ref{lem:unforgeable_QROM}}\label{sec:unforgeable_QROM_proof}
\begin{proof}
    Let $t(\lambda,n)=(n+1)\kappa$ where $\kappa={\sf poly}(\lambda)$ and $A$ be an efficient quantum adversary against $G^{\sf UF}_A({\sf mpk},{\sf msk})$, making at most $q$ queries to the random oracle $H$. By Lemma~\ref{lem:CMA2NMA_QROM},~\ref{lem:NMA2INT_QROM},~\ref{lem:INT2EXT_QROM}, we know that for every $({\sf mpk},{\sf msk})\in\Kcal\Pcal_m$, there exists efficient quantum adversaries $A_1, A_2, A_3$ respectively such that
    \begin{align*}
        &\Pr\left[A \text{ wins }G^{\sf UF}_A({\sf mpk},{\sf msk})\right]
        \leq \Pr\left[A_1 \text{ wins }\widetilde G^{\sf UF}_A({\sf mpk},{\sf msk})\right] + {\sf negl}(\lambda)\\
        &\leq \Pr\left[A_2 \text{ wins }\widetilde G^{\sf int}_A({\sf mpk},{\sf msk})\right](2q+1)^2 + {\sf negl}(\lambda)\\
        &\leq \left(\Pr\left[\substack{(\pk,\sk)\gets R(1^\lambda)\\ \sk\gets A_3(\pk)}\right] + \exp\left(\frac{-\kappa}{\mu^\mu}\right)\right)^{\frac{1}{2\mu-1}}(2q+1)^2 + {\sf negl}(\lambda)\;.
    \end{align*}
    By assumption $R$ is hard and $\kappa\geq {\sf poly}(\lambda)$, making the right-most term negligible. This concludes the proof.
\end{proof}

\section{Group signature}\label{sec:gsig}
A group signature scheme consists of one manager and $n$ parties. The
manager can set up a group and provide secret keys to each party.
Every party is allowed to generate signatures on behalf of the whole
group. Any party can verify the signature for the group without
knowing the signer, while the manager party can open the signer's
identity with his master secret key.

\noindent\textbf{Syntax.} A group signature scheme $\mathcal{GS}$
consists of the following four algorithms.
\begin{itemize}
\item $\mathbf{GKeygen}(1^\lambda,1^n)\rightarrow
  (\mathsf{gpk},\nset{\sk}{i},\mathsf{msk})$: The key generation
  algorithm $\mathbf{GKeygen}$ takes $1^\lambda$ and $1^n$ as inputs
  where $\lambda$ is the security parameter and $n \in \bbN$ is the
  number of parties in the group, and outputs $(\mathsf{gpk},
  \{\sk_i\}_{i\in[n]}, \ok)$ where $\mathsf{gpk}$ is the public key
  for the group, $\sk_i$ being the secret key of the $i$-th player for
  each $i \in [n]$, and $\ok$ is the master secret key held by the
  manager for opening.
\item $\mathbf{GSign}(\mathsf{gpk},m, \sk_k)\rightarrow\sigma$:
  The signing algorithm $\mathbf{GSign}$ takes a secret key $\sk_k$
  and a message $m$ as inputs, and outputs a signature $\sigma$
  of $m$ using $\sk_k$.
\item $\mathbf{GVerify}(\mathsf{gpk},m,\sigma)\rightarrow
  y\in\{0,1\}$: The verification algorithm $\mathbf{GVerify}$ takes
  the public key $\mathsf{gpk}$, a message $m$, and a candidate
  signature $\sigma$ as inputs, and outputs either $1$ for
  accept or $0$ for reject.
\item $\mathbf{GOpen}(\mathsf{gpk},\mathsf{msk},m,
  \sigma)\rightarrow k\in[n]$: The open algorithm
  $\mathbf{GOpen}$ takes the public key $\mathsf{gpk}$, the manager's
  master secret key $\mathsf{msk}$, a message $m$, and a signature
  $\sigma$ as inputs, and outputs an identity $k$ or abort with
  output $\bot$.
\end{itemize}

A group signature scheme should satisfy the following security
properties.

\noindent\textbf{Correctness.} A group signature scheme is said to be
correct if every honest signature can be correctly verified and
opened.
\begin{definition}
  A group signature scheme $\mathcal{GS}$ is correct if for any
  tuple of keys  
  $(\mathsf{gpk},\nset{\sk}{i},\mathsf{msk})\leftarrow
  \mathbf{GKeygen}(1^\secparam,1^n)$, any $i\in[n]$ and any message
  $m$,
  \begin{equation*}
    \Pr\left[
      \substack{
        \mathsf{acc}=1\land \mathsf{out}=i
      } : 
      \substack{
        \sigma\leftarrow\mathbf{GSign}(\mathsf{gpk}, m, \sk_i),\\
        \mathsf{acc} \leftarrow \mathbf{GVerify}(\mathsf{gpk}, m, \sigma),\\
        \mathsf{out} \leftarrow\mathbf{GOpen}(\mathsf{gpk}, \mathsf{msk}, m, \sigma)
      } \right]> 1-\mathsf{negl}(\lambda)
  \end{equation*}
\end{definition}

\noindent\textbf{Anonymity.} A group signature is said to be anonymous
if no adversary can determine the signer's identity among the 
group of signers given a signature, without using the master's secret key
($\mathsf{msk}$).
\begin{definition}
  A group signature scheme $\mathcal{GS}$ is anonymous if for any efficient
  adversary $A$ and any $n=poly(\lambda)$,
  \begin{equation*}
    \left|
      \Pr[1\leftarrow G_{A,0}^{\mathsf{Anon}}(\lambda, n)] - \Pr[1\leftarrow G_{A,1}^{\mathsf{Anon}}(\lambda, n)]
    \right|\leq \mathsf{negl}(\lambda),
  \end{equation*}
where the game $G_{A,b}^{\mathsf{Anon}}(\lambda, n)$ is defined below.
\vspace*{-0.7cm}
\end{definition}
\begin{algorithm}[H]
  \caption{$G_{A,b}^{\mathsf{Anon}}(\lambda, n)$: Anonymity game}
  \begin{algorithmic}[1]
    \STATE
    $(\mathsf{gpk},\nset{\sk}{i},\mathsf{msk})\leftarrow\mathbf{GKeygen}(1^\lambda,
    1^n)$
    \STATE $(st, i_0, i_1)\leftarrow A(\mathsf{gpk},\nset{\sk}{i})$\\
    \STATE $b \gets \{0,1\}$
    \RETURN $\mathsf{out}\leftarrow A^{\mathbf{GSign}(\mathsf{gpk},
      \cdot, \sk_{i_b})}(st)$
  \end{algorithmic}
\end{algorithm}
\vspace*{-0.7cm}
\noindent\textbf{Unforgeability.}  A group signature is said to be
unforgeable if no adversary can forge a valid signature that fails to
open or opens to some non-corrupted parties, even if the manager
has also colluded.
{\ifnum\submission=1
    \vspace*{-1em} 
\fi
\begin{definition}
  A group signature scheme $\mathcal{GS}$ is unforgeable if for any
  efficient adversary $A$ and any $n=poly(\lambda)$,
  \begin{equation*}
    \Pr[A\ \text{wins}\ G_{A}^{\mathsf{UF}}(\lambda, n)] < \mathsf{negl}(\lambda),
  \end{equation*}
where the game $G_{A}^{\mathsf{UF}}(\lambda, n)$ is defined below.
\end{definition}
\vspace*{-0.7cm}
\begin{algorithm}[H]
  \caption{$G_{A}^{\mathsf{UF}}(\lambda, n)$: Unforgeability game}
  \begin{algorithmic}[1]
    \STATE
    $(\mathsf{gpk},\nset{\sk}{i},\mathsf{msk})\leftarrow\mathbf{GKeygen}(1^\lambda,1^n)$,
    $\mathsf{Cor} = \{\}$ \STATE $( m^*, \sigma^*)\leftarrow
    A^{\mathbf{GSign}(\mathsf{gpk}, \bullet, \sk_i\notin \mathsf{Cor}),
      \textbf{Corrupt}(\bullet)}(\mathsf{gpk},\mathsf{msk})$\\\COMMENT{$\textbf{Corrupt}(i)$
      returns $\sk_i$ stores query $i$ in list $\mathsf{Cor}$} \STATE
    $A$ wins if $\sigma^*$ is not produced by querying
    $\mathbf{GSign}({\sf gpk},m^*,\sk_i\not\in{\sf Cor})$, $1\leftarrow \mathbf{GVerify}(\mathsf{gpk}, m^*,
    \sigma^*)$\\ and $i\leftarrow
    \mathbf{GOpen}(\mathsf{gpk},\mathsf{msk}, m^*, \sigma^*)$
    satisfies $i\notin \mathsf{Cor}$
      \end{algorithmic}
    \end{algorithm}
}

\section{Fiat-Shamir with Aborts Flaw in Related Works}\label{sec:FSwA_flaw}

We briefly introduce several relevant notions before describing the Fiat-Shamir with Aborts (FSwA) flaw.
In FSwA signatures (resp. NIZKs), one considers a Sigma protocol $({\sf com},{\sf ch},{\sf resp})\gets\Sigma$ that may abort (in such case ${\sf resp}=\bot$) with a certain probability. Such a protocol is called an {\em aborting Sigma protocol}. Typically, the transcript $({\sf com},{\sf ch},{\sf resp})$ may leak information about the secret key $\sk$, and the transcript can only be simulated conditioned on it not aborting (${\sf resp}\neq\bot$). Let ${\sf Sim}_\Sigma$ be such a simulator indistinguishable from $\Sigma|_{{\sf resp}\neq\bot}$ as specified in Fig.~\ref{fig:FSwA}. Typically, a FSwA signature is then constructed by repeating $\Sigma$ but replacing the challenge with some hash output (as produced by ${\sf FSwA}[\Sigma]$ in Fig.~\ref{fig:FSwA}). To show that it is hard to forge a signature, even given existing signatures, one typically has to perform a so-called CMA-to-NMA reduction, which makes up the signatures via $\Sim(\pk,m)$ as in Fig.~\ref{fig:FSwA} and give them to the forger. In such a reduction, it is then crucial to argue the signatures simulated by $\Sim$ are indistinguishable from the real signatures as generated by ${\sf FSwA}[\Sigma]$.

\begin{figure}
    \centering
    \begin{minipage}[t]{0.54\linewidth}
    $\Sigma(\sk)$: leaks $\sk$
    \begin{algorithmic}[1]
        \STATE ${\sf com},{\sf st}\gets\Com(\sk)$
        \STATE ${\sf ch}\gets\Ccal$
        \STATE ${\sf resp}\gets\Resp(\sk)$
        \RETURN $({\sf com},{\sf ch},{\sf resp})$
    \end{algorithmic}
    \ \\
     $\Sigma|_{{\sf resp}\neq\bot}(\sk)$: simulatable by $\Sim_\Sigma(\pk)$\\
    \vspace{-1em}\begin{algorithmic}[1]
        \REPEAT
            \STATE ${\sf com},{\sf st}\gets\Com(\sk)$
            \STATE ${\sf ch}\gets\Ccal$
            \STATE ${\sf resp}\gets\Resp(\sk)$
        \UNTIL{${\sf resp}\neq\bot$}
        \RETURN $({\sf com},{\sf ch},{\sf resp})$
    \end{algorithmic}
    
\end{minipage}
\begin{minipage}[t]{0.44\linewidth}
    $\Sign(\sk,m)$
    \begin{algorithmic}[1]
        \vspace{-1em}\REPEAT
            \STATE ${\sf com},{\sf st}\gets\Com(\sk)$
            \STATE ${\sf ch}\gets H(m,{\sf com})$
            \STATE ${\sf resp}\gets\Resp(\sk)$
        \UNTIL{${\sf resp}\neq\bot$}
        \RETURN $({\sf com},{\sf resp})$
    \end{algorithmic}
    \ \\
    $\Sim(\pk,m)$:
    \vspace{-1em}\begin{algorithmic}[1]
        \STATE $({\sf com},{\sf ch},{\sf resp})\gets \Sim_\Sigma(\pk)$
        \STATE $H(m,{\sf com}):={\sf ch}$
        \RETURN $({\sf com},{\sf resp})$
    \end{algorithmic}
\end{minipage}    
    \caption{An aborting Sigma protocol $\Sigma$ that may leak $\sk$, its non-abort transcripts $\Sigma|_{{\sf resp}=\bot}$ that does not leak $\sk$ (formally, simulatable by $\Sim_\Sigma$), and its FSwA signatures ${\sf FSwA}[\Sigma]$ that can be simulated by $\Sim$.}
    \label{fig:FSwA}
\end{figure}

Now, we give a high-level description of the Fiat-Shamir with Aborts (FSwA) flaw. In order to argue the closeness between ${\sf FSwA}[\Sigma]$ and $\Sim$, essentially in all the existing analyses, an intermediate oracle $\Trans$ is introduced, that (1) generates a non-abort transcript $({\sf com},{\sf ch},{\sf resp})\gets\Sigma|_{{\sf resp}\neq\bot}$, (2) reprograms $H(m,{\sf com}):={\sf ch}$, and (3) returns the transcript $({\sf com},{\sf resp})$. From the simulatability of non-abort transcripts, the closeness between $\Sim$ and $\Trans$ immediately follows, and hence it remains to argue the closeness between ${\sf FSwA}[\Sigma]$ and $\Trans$. The FSwA flaw lies in those analyses that argue (on a high level) that as long as the input $(m,{\sf com})$ where $H$ is being reprogrammed has not been queried prior to the reprogramming, then both oracles cannot be distinguished. However, this is not the case, for even without any prior query to the oracle ${\sf FSwA}[\Sigma]$ or $\Trans$, there is still {\em positive advantage} of distinguishing both oracles. Indeed, each time the oracle $\Trans$ is  queried, a non-abort transcript $({\sf com},{\sf ch},{\sf resp}\neq\bot)$ is reprogrammed to the oracle, thereby biasing the distribution of $H$ toward having more non-aborting input-output pairs. This flawed argumentation occurs not only in Dilithium (as \cite{FSwA23} have pointed out) but also in \cite{beullens2021group} and likely even in \cite{beullens2020calamari,lai2021collusion}.

Since the flaw that appeared in the analyses of Dilithium has been fixed by \cite{FSwA23}, it is natural to ask if similar techniques fix relevant isogeny-based ring/group signatures.\footnote{The work \cite{FSwA23_Devevey} also fixes the FSwA flaw for Lyubashevsky-style signatures, but it premises a stronger simulator that is not available in isogeny-based constructions.} To the best of our knowledge, the fix as provided in \cite{FSwA23} {\em does not} immediately fix these works. Indeed, in each of \cite{beullens2020calamari,lai2021collusion,beullens2021group}, the construction crucially relies on a non-standard variant of aborting Sigma protocols $\Sigma^H$ (specified by $\Com^H, \Resp^H)$ that is given additional query access to the random oracle $H$. However, in showing the closeness between $\Sign$ and $\Trans$, for both games, the sub-procedures $\Com^H$ and $\Resp^H$ are now given query access to the random oracle $H$. It is then conceivable that such additional access may help an adversary to distinguish both oracles.

Indeed, zooming into the argumentation of the Dilithium fix \cite{FSwA23}, one relies (information-theoretically) on the fact that a distinguisher interacting with $H$ and $\Trans$ cannot learn those aborting transcripts generated in $\Trans$. However, since now the aborting transcripts are partly determined by the randomness of $H$, by making queries to $H$, the distinguisher may actually learn something about those aborting transcripts. Therefore, for fixing the FSwA flaw in all currently available isogeny-based ring/group signatures \cite{beullens2020calamari,lai2021collusion,beullens2021group} (besides ours), a new idea or a very different proof is necessary.

\subsection{Details of the Flaw in \cite{BDKLP22}}
Here, we elaborate on the technical details of how the Fiat-Shamir with Aborts flaw affects the analyses of \cite{BDKLP22}. We point to the proof in its ePrint version, \cite[Theorem~6.4]{beullens2021group}, in which the closeness of two specific games ${\sf Game}_1$ and ${\sf Game}_2$ is argued.

To start, we briefly recap the definitions of both games. In ${\sf Game}_1$, the adversary $\Acal$ interacts with the oracle ${\sf Prove}$, while in ${\sf Game}_2$, such oracle is replaced by another oracle $\Scal_{\sf int}$, which runs the underlying aborting Sigma protocol (as $P_1$ and $P_2$ specified in \cite[Fig.~4]{beullens2021group}) and then reprograms the transcript to the random oracle, which is formally realized by maintaining a list $L$ of previously defined inputs.

Indeed, if taken at the face value, ${\sf Game}_1$ and ${\sf Game}_2$ are easily distinguishable, because in the former game the oracle ${\sf Prove}$ always returns non-$\bot$ transcripts, whereas in the latter game $\Scal_{\sf int}$ may return an aborting transcript (with $\bot$).

A natural way to fix this is by insisting that $\Scal_{\sf int}$ always generates a non-aborting transcript and then reprograms such transcript to the random oracle. Concretely, this can be done by adding a for-loop in ${\sf Sim}_{\sf int}$ (see \cite[Fig.~9]{beullens2021group}) that terminates after ${\sf resp}\neq\bot$, and only then executes the reprogramming $L[{\sf FS}\;\Vert\; {\sf lbl}\;\Vert\; X\;\Vert\; {\sf com}]:={\sf chall}$. However, if this is their intended approach, the FSwA flaw re-appears. This is because such reprogramming biases the random oracle toward having more non-$\bot$ input-output pairs. Hence, even if $Q_2=1$ and no query to $\Ocal/{\sf Sim}_0$ is made prior to the ${\sf Prove}/\Scal_{\sf int}$ query, there is still a positive advantage of distinguishing ${\sf Game}_1$ and ${\sf Game}_2$, contradicting the reasoning in, quote “the view of $\Acal$ is identical to the previous game unless ${\sf Sim}_{\sf int}$ outputs $\bot$ in Line 4.”

\paragraph{Is it fixable?} To the best of our knowledge, there is no immediate solution to the flaw. Below, we show two natural alternatives that do not (directly) fix the flaw.

First, if we modify $\Scal_{\sf int}$ in such a way that does the reprogramming $L[{\sf FS}\;\Vert\; {\sf lbl}\;\Vert\; X\;\Vert\; {\sf com}]:={\sf chall}$ in every iteration of the for-loop (instead of at the end), then one may then be able to show the closeness between ${\sf Game}_1\approx{\sf Game}_2$, but the very next hybrid step ${\sf Game}_2\approx{\sf Game}_3$ falls apart. This is because, under the considered non-abort HVZK property, only non-abort transcripts are guaranteed to be efficiently simulatable, and hence, the reprogramming taking place at abort iterations cannot be simulated in (any imaginable twists of) ${\sf Game}_3$ efficiently. Therefore, this approach does not (immediately) work.

Second, one may be tempted to apply generic results provided in \cite{FSwA23}, but the (flawed) analyses in \cite[Theorem~6.4]{beullens2021group} are performed in a non-blackbox manner, and simply un-covered by such results. In addition, note that \cite{beullens2021group} considers a twisted variant of aborting Sigma protocol, where the prover is given access to the random oracle, which is not the case in \cite{FSwA23}. It is conceivable that such additional access to the random oracle helps an adversary to distinguish ${\sf Game}_1$ and ${\sf Game}_2$. Therefore, it is unclear how (and whether it is possible) to fix such a flaw, even with the reasoning provided in \cite{FSwA23}.
\end{document}